\documentclass[aps,superscriptaddress,twocolumn,10pt,prx]{revtex4-2}

\usepackage{amsmath,mathtools,amsthm,amssymb}
\usepackage{tabularx}
\usepackage{tabularray}

\usepackage[colorlinks=true,citecolor=teal,linkcolor=teal,urlcolor=teal,breaklinks,pdftex]{hyperref}

\usepackage{graphicx}
\usepackage{multirow}
\usepackage{color}
\usepackage[dvipsnames]{xcolor}
\usepackage{bbold}
\usepackage{enumitem}
\usepackage{physics}
\usepackage{comment}
\usepackage{array}
\usepackage{graphics}
\usepackage{wrapfig}
\graphicspath{{}}
\usepackage{biolinum}
\usepackage{bbm}
\usepackage{dsfont}
\usepackage{mathrsfs}
\usepackage{mathdots}
\usepackage{enumitem}
\usepackage{pifont}

\usepackage[capitalize]{cleveref}

\usepackage{physics}

\newcommand{\cut}{{\mathrm{cut}}}

\usepackage[linesnumbered,ruled,vlined]{algorithm2e}
\SetAlFnt{\small}
\SetKwComment{Comment}{// }{}
\crefname{algocf}{Algorithm}{Algorithms}

\newtheorem{theorem}{Theorem}
\newtheorem{lemma}{Lemma}

\newtheorem{definition}{Definition}
\newtheorem{corollary}{Corollary}

\newtheorem{remark}{Remark}
\newtheorem{fact}{Fact}
\newtheorem{open}{Open Question}

\newcommand{\haar}[0]{\operatorname{Haar}}

\newcommand{\id}{\ensuremath{\mathbb{I}}}
\newcommand{\poly}{\operatorname{poly}}
\newcommand{\E}{{\mathbb{E}}}
\newcommand{\negl}{{\mathsf{negl}}}

\newcommand{\prf}{\mathsf{PRF}}
\newcommand{\prp}{\mathsf{PRP}}
\newcommand{\pru}{\mathsf{PRU}}
\newcommand{\gue}{{\mathrm{GUE}}}
\newcommand{\locc}{\mathrm{LOCC}}
\newcommand{\pgue}{{\mathcal{E}_{\tilde{d}}}}
\DeclareMathOperator{\loe}{LOE}
\DeclareMathOperator{\OTOC}{OTOC}
\DeclareMathOperator{\de}{d\!}

\definecolor{airforceNonE} {rgb}{0.36, 0.54, 0.66}

\allowdisplaybreaks

\newcommand{\be}{\begin{equation}\begin{aligned}\hspace{0pt}}
\newcommand{\ee}{\end{aligned}\end{equation}}
\newcommand{\ba}{\begin{eqnarray}}
\newcommand{\ea}{\end{eqnarray}}

\newcommand{\nocontentsline}[3]{}
\let\origcontentsline\addcontentsline
\newcommand\stoptoc{\let\addcontentsline\nocontentsline}
\newcommand\resumetoc{\let\addcontentsline\origcontentsline}

\begin{document}

\title{Simulating quantum chaos without chaos}

\author{Andi Gu}
\affiliation{Department of Physics, Harvard University, 17 Oxford Street, Cambridge, MA 02138, USA}

\author{Yihui Quek}

\affiliation{Department of Physics, Harvard University, 17 Oxford Street, Cambridge, MA 02138, USA}
\affiliation{Department of Computer Science, Harvard John A. Paulson School Of Engineering And Applied Sciences, 150 Western Ave, Boston, MA 02134, USA}

\author{Susanne Yelin}

\affiliation{Department of Physics, Harvard University, Cambridge, Massachusetts 02138, USA}

\author{Jens Eisert}
\affiliation{Dahlem Center for Complex Quantum Systems, Freie Universit\"at Berlin, 14195 Berlin, Germany}
\affiliation{Helmholtz-Zentrum Berlin f\"ur Materialien und Energie, 14109 Berlin, Germany}

\author{Lorenzo Leone}
\affiliation{Dahlem Center for Complex Quantum Systems, Freie Universit\"at Berlin, 14195 Berlin, Germany}

\begin{abstract}
Quantum chaos is a quantum many-body phenomenon that is associated with a number of intricate properties, such as level repulsion in energy spectra or distinct scalings of out-of-time ordered correlation functions. In this work, we introduce a novel class of ``pseudochaotic'' quantum Hamiltonians that fundamentally challenges the conventional understanding of quantum chaos and its relationship to computational complexity. Our ensemble is computationally indistinguishable from the Gaussian unitary ensemble (GUE) of strongly-interacting Hamiltonians, widely considered to be a quintessential model for quantum chaos. Surprisingly, despite this effective indistinguishability, our Hamiltonians lack all conventional signatures of chaos: it exhibits Poissonian level statistics, low operator complexity, and weak scrambling properties. This stark contrast between efficient computational indistinguishability and traditional chaos indicators calls into question fundamental assumptions about the nature of quantum chaos. We, furthermore, give an efficient quantum algorithm to simulate Hamiltonians from our ensemble, even though simulating Hamiltonians from the true GUE is known to require exponential time. Our work establishes fundamental limitations on Hamiltonian learning and testing protocols and derives stronger bounds on entanglement and magic state distillation. These results reveal a surprising separation between computational and information-theoretic perspectives on quantum chaos, opening new avenues for research at the intersection of quantum chaos, computational complexity, and quantum information. Above all, it challenges conventional notions of what it fundamentally means to actually observe complex quantum systems.
\end{abstract}

\maketitle
\stoptoc
\section{Introduction}\label{intro}
Quantum chaos, at the intersection of quantum mechanics and classical chaos theory, plays a crucial role in our understanding of complex quantum systems, from atomic nuclei to black holes~\cite{haake2018quantum}. Traditionally characterized by level repulsion in energy spectra~\cite{mehta2004random}, growth of operator complexity~\cite{prosen2007efficiency}, and rapid decay of \emph{out-of-time-ordered correlators} (OTOCs)~\cite{maldacena2016chaos}, these features have become standard diagnostics for identifying chaotic behavior in quantum systems. 
Along with the 
advancement of quantum computation, there is a growing interest in simulating complex quantum systems, including chaotic ones. This intersection raises fundamental questions about the nature of complexity in quantum systems and the resources required to simulate them.

In this work, we present a surprising result that challenges basic intuitions about quantum chaos and its relationship to computational complexity. Concretely, we construct a quantum algorithm to simulate a ``pseudochaotic'' ensemble of Hamiltonians that is computationally indistinguishable from the \emph{Gaussian unitary ensemble} (GUE), often considered the quintessential model of quantum chaos~\cite{liu2018spectral}.  The GUE's importance as a model for complex quantum systems is underscored by its wide-ranging applications across various domains of physics. These applications span from condensed matter physics, where it describes electron transport~\cite{beenakker1997random}, to quantum chromodynamics~\cite{verbaarschot2000random}, nuclear physics for modeling energy levels~\cite{wigner1951statistical}, and even to the understanding of anti-de Sitter black hole dynamics~\cite{cotler2017black}. This ubiquity motivates our investigation into pseudochaotic ensembles as computationally efficient alternatives to the GUE for simulating and studying complex quantum systems that exhibit chaotic behavior.

\begin{figure}
    
   \includegraphics[width=1.04\columnwidth]{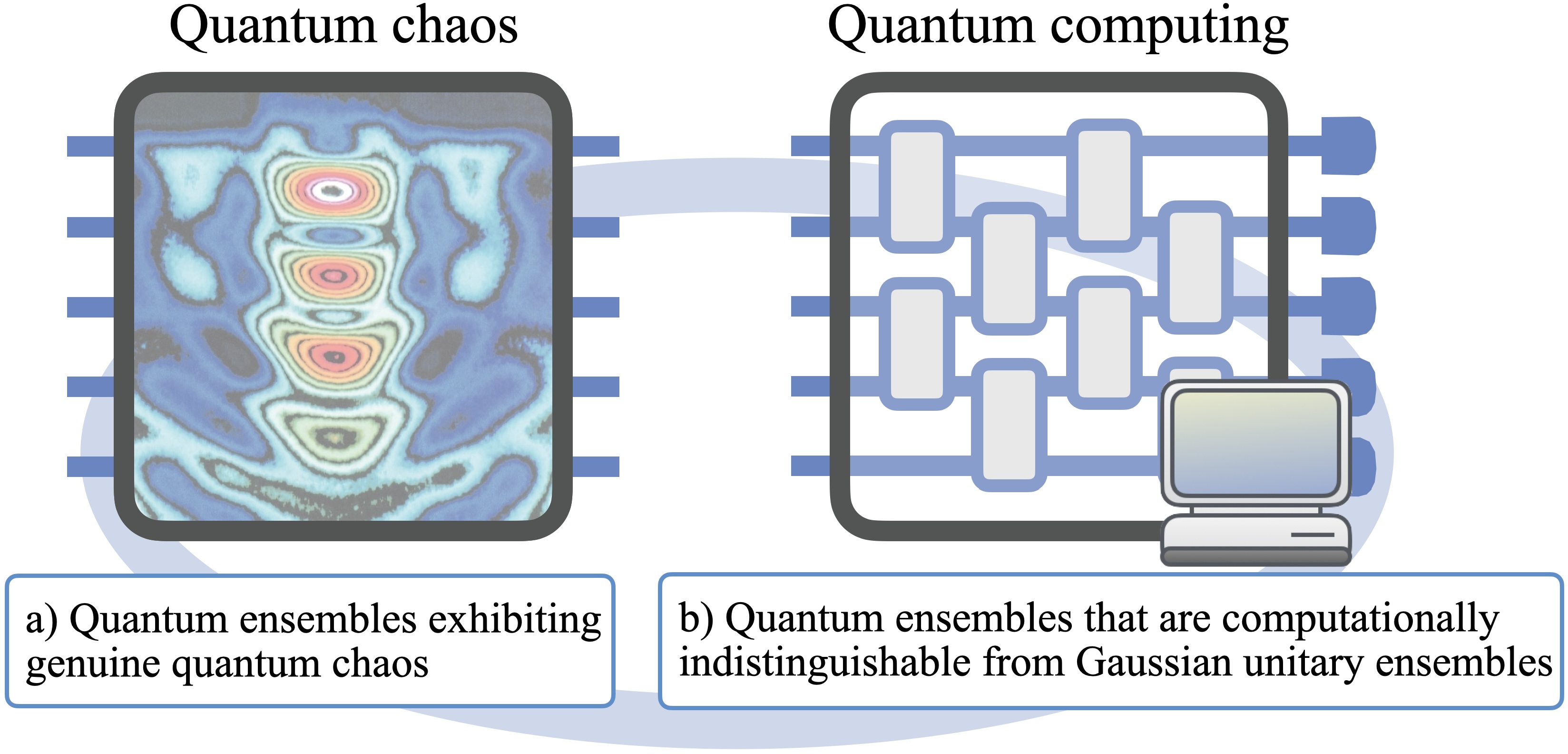}

    \caption{We consider quantum ensembles that are indistinguishable 
    from a) Gaussian unitary ensembles featuring signatures of
    genuine quantum chaotic dynamics 
    by 
    any computationally restricted observer: Yet, b) this new ensemble that is computationally efficiently preparable on quantum computers does in many ways not exhibit quantum chaotic features.}
    \label{fig:ame_circuit}
\end{figure}

Our ensemble is notable not just for the features it displays, but those it lacks: surprisingly, it fails to exhibit several hallmarks of chaos. It shows no level repulsion, contrary to the Wigner-Dyson statistics of the GUE, demonstrates significantly larger late-time OTOC values ($\sim \exp(-\poly\log n))$) compared to chaotic systems ($\sim \exp(-n)$)~\cite{cotler2017black}. Similarly, the complexity of time-evolved local operators under our ensemble, as quantified by the local operator entanglement, saturates at significantly lower values compared to the GUE. Despite this, no physical experiment can distinguish our ensemble from the GUE, except at the expense of an extraordinary amount of time. These results suggest that the heuristics previously thought to indicate the onset of quantum chaos are, in fact, not essential, as their presence or absence leads to the same emergent phenomena for any computationally bounded observer.

Our work extends far beyond challenging conventional notions of quantum chaos to practical results concerning quantum simulation, resource theory and quantum learning theory.
\begin{itemize}
    \item {\em Simulating GUE.} Remarkably, our pseudochaotic ensemble is efficiently simulable on a quantum computer, overcoming known exponential lower bounds for simulating true GUE Hamiltonians~\cite{kotowski2023extremal}.

    \item {\em Hamiltonian learning.} We establish a no-go result for Hamiltonian learning from black-box access to the time evolution. In particular, even if there exists an efficient circuit that generates the time evolution for a given $H$, there is no polynomial-time quantum algorithm that learns the time evolution operator for $H$. We can make this result stronger and show that the no-go result persists even if the Hamiltonian is sparse in the computational basis.
    
    \item {\em Hamiltonian property testing.} We establish sample-complexity lower bounds for algorithms which aim to determine key properties of quantum systems, such as entanglement and magic production, scrambling (as quantified by out-of-time-order correlators), and the spreading of local operator entanglement. %with black-box access to Hamiltonians \cite{bluhm2024hamiltonianpropertytesting} 
    \item {\em Testing spectral properties.} We show that there is no efficient quantum algorithm for determining whether an ensemble of Hamiltonians obeys Wigner-Dyson or Poisson spectral statistics, which are commonly associated with different classes of chaotic and integrable systems. Surprisingly, this no-go result holds even if we are guaranteed that the Hamiltonians are sparse.
    \item {\em Pseudoresourceful unitaries and tighter resource distillation bounds.} The unitary dynamics generated by pseudochaotic Hamiltonians introduce the stronger concept of \textit{pseudoresourceful} unitary operators. These are unitary operators that, while indistinguishable from those producing highly entangled and highly magical states, in fact generate states with low entanglement and low magic. We prove that the near exponential suppression of conversion rates in efficient resource distillation protocols still persists, even when the distiller has black-box query access to the unitary which prepares the given resource state. This demonstrates that even this form of prior knowledge is too weak to enable practical distillation algorithms.
\end{itemize}

%These results can be viewed as a converse to recent work on efficient algorithms for testing and learning local Hamiltonians \cite{bluhm2024hamiltonianpropertytesting,gutiérrez2024simplealgorithmstestlearn}, as GUE Hamiltonians are non-local. Notably, these findings establish sharp limits on our ability 
%to extrapolate the long-time behavior of quantum systems from early-time dynamics.

More broadly, our work adds to the toolbox of quantum pseudorandomness, an emerging theme in the complexity of physical systems. While recent research has focused on pseudorandom states~\cite{ji2018pseudorandom}, pseudorandom unitaries~\cite{chen2024efficient,schuster2024random,bostanci2024efficient,metger2024simple,ma2024constructrandomunitaries}, and ``pseudoresourceful'' states~\cite{aaronson2024quantum,gu2024pseudomagic}, our work extends these initial ideas to the realm of Hamiltonians and quantum dynamics. Moreover, we rigorously establish that the late-time dynamics generated by GUE Hamiltonians is indistinguishable from Haar-random unitaries, thus providing a rigorous recipe for constructing pseudoresourceful unitaries that are also pseudorandom.
 This extends the landscape of quantum pseudorandomness beyond 
states and unitaries to encompass the generators of quantum dynamics and to notions that are naturally
associated with the realm of quantum many-body physics.

%\pseudoyihui{I am deleting the next paragraph because all points made there have already been made in a sharper way in previous paragraphs}
%These results have significant implications for our understanding of quantum chaos and its relationship to computational complexity. They provide a new tool for quantum simulation and challenges conventional wisdom about quantum chaos indicators and computational complexity. Our results suggest a potential hierarchy in chaotic behavior, with computational indistinguishability representing a different level of chaos than that captured by traditional measures.

\section{Framework}

\textbf{What is quantum chaos?} While chaos theory in classical mechanics is well-established and elegantly formulated~\cite{eckmann_ergodic_1985,schuster_deterministic_2005,buzzi_chaos_2009}, its quantum counterpart remains more elusive. The concept of quantum chaos is inherently ambiguous, largely due to the absence of direct quantum analogs for key classical notions such as the butterfly effect (characterized by rapidly diverging trajectories in phase space) and integrals of motion (conserved quantities during dynamics). This fundamental difference has led researchers to postulate that any comprehensive theory of quantum chaos, if it exists, must be qualitatively distinct from its classical counterpart. The challenge lies in developing a framework that captures the essence of chaotic behavior in quantum systems without relying on classical intuitions that may not apply at the quantum scale.

Since the 1950s, numerous approaches have been explored, contributing to a rich body of insightful results. One of the oldest and most promising indicators of chaos in this context is found in the statistical properties of the Hamiltonian spectra~\cite{haake2018quantum}. Wigner-Dyson statistics of energy levels are traditionally associated with chaotic systems, reflecting the phenomenon of \textit{level repulsion}, where energy 
levels tend to avoid each other~\cite{bohigas1984characterization,guhr1998random}. 
%We can mention the Berry-Tabor conjecture, for the Poissonian energy levels to put a bit of historical background.
In contrast, Poissonian statistics, typical of integrable systems as suggested by the Berry-Tabor conjecture, exhibit no level repulsion~\cite{berry1997level}. In Poissonian systems, energy levels are uncorrelated, with a non-zero probability of finding levels arbitrarily close together, hinting at the presence of underlying symmetries (see \cref{fig:spacing}). This distinction reflects the dynamics at play: Wigner-Dyson statistics indicate complex, chaotic behavior with many interacting degrees of freedom, while Poissonian statistics suggest simpler, more predictable out-of-equilibrium quantum dynamics~\cite{dalessio2016quantum,deutsch2018eigenstate,gogolin2016equilibration,PolkovnikovReview,OutOfEquilibrium}. This signature, or lack thereof, of chaos has been extensively studied numerically in numerous many-body models considered chaotic~\cite{rabson2004crossover,leblond2021universality,zeng1994scaling}.
%such as X and Y.
%\je{[We could cite, e.g.,  ]}
% Added this, please disagree. These are great references Jens. I moved them to show that WD statistic is not sufficient for QC, but (maybe) only necessary as it has been observed in integrable models. See few paragraph below.

\begin{table*}
    \setlength{\tabcolsep}{6pt} % Default value: 6pt
    \renewcommand{\arraystretch}{1.4} % Default value: 1
    \begin{centering}
        \begin{tabular}{|c|c|c}
                \hline
                {\textbf{Probe to chaos}}  & {\textbf{Definition}} \\
                \hline
                \hline
                $4$-point OTOC~\cite{kitaev2014HiddenCorrelationsHawking} & $\OTOC(H,t)\coloneqq\frac{1}{d}\tr[O_1(t)O_2O_1(t)O_2]$                                      \\
                \hline
                $2$-R\'enyi entropy~\cite{zanardi2001entanglement}   &

                $S_{A|B}(H,t;\rho_0)\coloneqq-\log \tr_A[(\tr_B(e^{-iHt} \rho_0 e^{iHt}))^2]$                             \\
                \hline
                Operator entanglement~\cite{zanardi2001entanglement}  & $\loe(H,t)\coloneqq S_2(\tr_{B,B'}\ketbra{O_1(t)})$       \\
                \hline
                Stabilizer entropy~\cite{leone2022stabilizer}   &

                $M_{\alpha}(H, t; \rho_0)\coloneqq\frac{1}{1-\alpha}\log \frac{1}{d} \sum_P \tr^{2\alpha}(P e^{-iHt} \rho_0 e^{iHt})$                             \\
                
                \hline
            \end{tabular}
        \caption{\label{table:magicmonotone} Probes to chaos considered in this work. We typically associate chaos with systems for which there exists times $t=O(\poly n)$ such that the 2-R\'enyi entropy, operator entanglement, and stabilizer entropy are extensive in system size, and the 4-point OTOC is exponentially vanishing in system size.} 
        %and their definitions.}
        \label{tableprobestochaos}
    \end{centering}
\end{table*}

Modern approaches to quantum chaos have shifted focus from properties of the Hamiltonian to emergent many-body phenomena produced by unitary dynamics $e^{-iHt}$. After all, all we need to understand is how emergent physical laws arise from many-body unitary dynamics. A particularly notable development in this direction is the identification of a quantum analog of the butterfly effect, marked by the rapid spread of quantum correlations into non-local degrees of freedom~
\cite{kitaev2014HiddenCorrelationsHawking}.
%\cite{kitaev_hidden_2014}. 
% Is this the right reference? Yes.
%
\textit{Quantum scrambling} describes how quantum information spreads through a system under unitary evolution. Initially localized information, encoded in a spatially local operator $O_1$, spreads over time due to the unitary dynamics $O_1(t)\coloneqq e^{-iHt}O_1e^{iHt}$. The effect of this spreading can be captured by a \textit{test} local operator $O_2$ via the decay of \emph{out-of-time-ordered correlators} (OTOCs), with faster decay generally associated with stronger scrambling~\cite{hashimoto2017otoc,swingle2016measure}. 

Another information-theoretic measure of this spreading is the \textit{local operator entanglement} (LOE) of $O_{1}(t)$, which quantifies how quantum information spreads between two subsystems $S=A\cup B$. Formally, the LOE is defined as the bipartite entanglement of the Choi state vector $\ket{O_1(t)}\coloneqq O_{1}(t)\otimes I_{S'}\ket{S,S'}$ associated with $O_{1}(t)$, where $\ket{S,S'}$ is the maximally entangled state vector between the system $S$ and an identical copy $S'$ (\cref{tableprobestochaos}). Rapid growth in LOE, 
as well as large saturation values, are associated with stronger scrambling of quantum information under $e^{-iHt}$.

That being said, we can so far only assert that chaos is \textit{diagnosed} through indicators, or ``probes'', that signal the presence of quantum chaos, but do not fully characterize it. Some quantum systems do not exhibit chaotic behavior, yet still display these same indicators. For instance, Wigner-Dyson statistics of energy 
levels have been observed in integrable systems~\cite{PhysRevLett.80.3996,PhysRevE.68.045201,PhysRevE.103.062211}, meaning that level repulsion is not a defining feature of chaotic systems. Similarly, there are quantum systems that do not exhibit chaotic behavior, yet display strong scrambling~\cite{PhysRevB.97.144304,PhysRevLett.132.080402}. Thus, the current understanding suggests that neither scrambling nor level-repulsion are sufficient to define quantum chaos; they merely serve as indicators or, more precisely, are often considered only necessary conditions for the emergence of quantum chaos~\cite{Dowling_2023}.

%That being said, we can so far only assert that scrambling is an indicator, or \textit{probe}, of quantum chaos. There are quantum systems that do not exhibit chaotic behavior, yet display strong scrambling. Thus, the current understanding suggests that scrambling is a necessary, but not sufficient, indicator for the onset of quantum chaos~\cite{Dowling_2023}.
In a similar vein, quantum chaotic dynamics often produce highly entangled and highly magical states, indicating the production of states lying beyond the reaching of classical simulability. However, again, these features alone do not fully characterize chaotic dynamics, as non-chaotic dynamics can also generate entangled and magical states~\cite{PhysRevA.106.042426,Tarabunga2024criticalbehaviorsof}.

What, then, are we left with? Current approaches to quantum chaos are not designed to fully characterize chaotic dynamics. Instead, they focus on identifying the conditions necessary for its emergence, primarily through the observation of emergent phenomena, or \textit{signatures} of chaotic dynamics. While these traditional approaches have provided valuable insights, they do not fully capture the role of observers in characterizing quantum chaos. To address this, we introduce the concept of computationally bounded observers. 
\smallskip
%After all, the primary scope of a quantum theory of chaos is to give ergodicity a solid theoretical ground and the emergent of thermodynamics via many-body properties. 

%This presents a double-edged sword: while these approaches cannot definitively distinguish chaotic dynamics, they serve as a bridge, linking the theoretical framework to experimentally observable behaviors driven by emergent physical phenomena.

\textbf{The crucial role of computationally bounded observers.} Modern approaches to quantum chaos subtly introduce a new perspective: if the essence of quantum chaos needs to be found in emergent physical phenomena produced by unitary dynamics, then observers are not merely passive verifiers but integral parts of the theory. This makes it crucial to account for the intrinsic limitations of physical observers.

Consider two $n$-particle dynamics, $e^{-iH_1t}$ and $e^{-iH_2t}$, which are predicted to exhibit distinct physical behaviors characterized by some emergent phenomenon $C$ (for instance, different late-time correlation functions). However, if the time required to experimentally distinguish between these dynamics using any measurement scheme is on the order of the age of the universe, what practical relevance does $C$ truly have? In such a case, even though $e^{-iH_1t}$ and $e^{-iH_2t}$ differ by some attribute $C$, they are indistinguishable within any reasonable timeframe. This raises the question: in what sense can $C$ be considered a truly \emph{physical} phenomenon?

Given the central role of observers in quantum chaos theory, and in light of their fundamental computational constraints, we are motivated to introduce the concept of \textit{pseudochaotic Hamiltonians}. These are systems whose dynamics, while failing to meet the necessary conditions commonly associated with quantum chaos, remain indistinguishable from chaotic systems for any observer limited to measurement schemes that operate efficiently within a reasonable amount of time. In computational terms, a measurement scheme is efficient if the time required (in terms of the number of elementary operations) scales polynomially with the number $n$ of particles, which we write as $\poly(n)$. 

In order to formalize the idea of pseudochaotic Hamiltonians, we introduce the concept of black-box Hamiltonian access, 
which provides a general and natural way to interact with quantum systems without requiring detailed knowledge. 
%of their internal structure. 
The notion of black-box Hamiltonian access has emerged as a standard model in the field of 
Hamiltonian learning and quantum system characterization~\cite{BenchmarkingReview,GoogleHamiltonianLearning,anshu2021,stilck2024efficient,gu2024practical,bakshi2023learning,bakshi2024structure,huang2023learning}. Simply, an algorithm $\mathcal{A}^{H}$ has black-box Hamiltonian access to $H$ if it can query the time evolution $e^{-iHt}$ generated by $H$ at any given set of times $\abs{t}=O(\poly n)$ (as well as its controlled version) and also has access to samples of the Gibbs state $\propto e^{-\beta H}$ for any inverse temperature $\beta = O(\poly n)$.
%\begin{definition}[Black-box Hamiltonian access]
%For a given $n$-qubit Hamiltonian $H$, let $U_H$ be an oracle that, when queried with any $t\in \mathcal{R}$ such that $\abs{t} \leq \exp(O(n))$, implements the unitary $e^{-iH t}$ (as well as its controlled version). We say an algorithm has black-box Hamiltonian access to $H$ if it can query $U_H$ and also has access to samples of the Gibbs states $e^{-\beta H}/\tr(e^{-\beta H})$ for any $\beta \leq O(\poly n)$. We will denote algorithms $\mathcal{A}$ with black-box Hamiltonian access as $\mathcal{A}^H$.
%\end{definition}
%This access model provides a general framework for interacting with quantum systems in a way that abstracts away the specific details of the Hamiltonian's structure. 
Building upon this concept, we now introduce the core concept of our work: \emph{pseudochaotic Hamiltonians}. 
%More specifically, we construct and give an algorithm to implement an ensemble of Hamiltonians that, while structurally distinct from truly chaotic systems, exhibits behavior computationally indistinguishable from them. The following definition formalizes this notion:

\begin{definition}[Pseudochaotic Hamiltonians] Let $\mathcal{E}_1=\{H_{k}\}$ be a ``chaotic'' ensemble of Hamiltonians which exhibits all the necessary properties commonly associated to quantum chaos (see \cref{tableprobestochaos}) with high probability. A ``pseudochaotic'' ensemble of Hamiltonians $\mathcal{E}_2=\{H_l\}$ has two key properties: 
\begin{enumerate}[label=(\roman*)]
        \item It does not exhibit such chaotic properties: it does not exhibit level repulsion, strong scrambling capabilities, nor does it generate high entanglement and magic states. 
        \item It is computationally indistinguishable from $\mathcal{E}_1$ for any efficient algorithm $\mathcal{A}^H$ with black-box Hamiltonian access. That is,
        \begin{equation}
        \abs{\Pr_{H \sim \mathcal{E}_1}[\mathcal{A}^{H}()=1] - \Pr_{H \sim \mathcal{E}_2}[\mathcal{A}^{H}()=1]} = \negl(n).
        \end{equation}
    \end{enumerate}
\end{definition}
%\begin{definition}[pseudochaotic Hamiltonians] A pseudochaotic pair of Hamiltonians consists in two ensembles of Hamiltonians:
 %   \begin{enumerate}[label=(\roman*)]
  %      \item a ``chaotic'' ensemble of Hamiltonians $\mathcal{E}_1=\{H_{k}\}$, which exhibit all the necessary properties commonly associated to quantum chaos with high probability;
   %     \item a ``pseudochaotic'' ensemble $\mathcal{E}_2=\{H_l\}$, which lacks in fully satisfying chaotic properties with high probability.
    %\end{enumerate}
    %Yet, the two ensembles are computationally indistinguishable. That is, for any efficient algorithm $\mathcal{A}^H$ with black-box Hamiltonian access
   % \begin{equation}
    %    \abs{\Pr_{H \sim \mathcal{E}_1}[\mathcal{A}^{H}()=1] - \Pr_{H \sim \mathcal{E}_2}[\mathcal{A}^{H}()=1]} = \negl(n).
    %\end{equation}
%\end{definition}
If pseudochaotic Hamiltonians exist, this would pose a strong challenge to our understanding of quantum chaos. The crucial observation is this: if Hamiltonians with vastly different spectral statistics, scrambling properties, and capacities for entanglement and magic generation still produce the same emergent phenomena for every physical observer, how can these features be considered necessary for the onset of quantum 
chaos?
\smallskip
%To appreciate the distinction between chaotic and pseudochaotic ensembles, understanding the differences between Wigner-Dyson and Poissonian statistics is crucial. 

\section{Results}
\textbf{Construction of pseudochaotic Hamiltonians.} Having defined pseudochaotic Hamiltonians, we now prove their existence with an explicit construction. 
To begin with, we focus on the chaotic ensemble. The \emph{Gaussian unitary ensemble} (GUE) of Hamiltonians is the archetypal example of a chaotic quantum system exhibiting Wigner-Dyson statistics. Formally, 
the Gaussian unitary ensemble $\mathcal{E}_{\gue}$ is an ensemble of $d \times d$ Hermitian matrices $H$ that is distributed according to
$\Pr(H) \propto e^{-\frac{d}{2} \tr H^2}$. Beyond obeying Wigner-Dyson spectral statistics, it is easy to show that it also fulfills the other necessary properties for quantum chaotic Hamiltonians with overwhelming probability. That is, they exhibit high LOE, fast decay of OTOCs, and generate maximally entangled and magical states (see \cref{th:separations}). Therefore, we identify the ``chaotic'' ensemble as $\mathcal{E}_{\gue}$.

We now proceed with the construction of the pseudochaotic ensemble of Hamiltonians. A general Hamiltonian can always be diagonalized with $H = U \Lambda U^\dagger$, where $\Lambda$ is a diagonal matrix containing the spectrum of the Hamiltonian and $U$ encodes the eigenvectors of $H$. Our construction starts with the observation that we can build an ensemble of Hamiltonians 
$\mathcal{E}$ by choosing eigenvalues $\Lambda$ and eigenvectors $U$ independently. We borrow some known results on the GUE ensemble. First, it is well known that the eigenvectors of the GUE are Haar random. Hence, we generate them using an ensemble $\mathcal{U}$ of pseudorandom unitaries,  which are efficiently implementable yet computationally indistinguishable from Haar random unitaries~\cite{chen2024efficient,schuster2024random,bostanci2024efficient,metger2024simple,ma2024constructrandomunitaries}. 
Second, the marginal distribution of a single eigenenergy is described by \emph{Wigner's famous semi-circle distribution} given by $p(\lambda) = {\sqrt{4 - \lambda^2}}/({2\pi})$. It turns out that, despite the fact that the full GUE 
eigenvalue distribution is strongly correlated, it can be ``spoofed'' by a joint distribution consisting of independent semi-circle distributions. Crucially, 
we can go one step further and define a highly \emph{degenerate}, but permutation-invariant, spoofing distribution $\tilde{p}_{\tilde{d}}$ such that each eigenenergy has degeneracy $d / \tilde{d}$. That is, there are exactly $\tilde{d}$ unique eigenenergies (hence $\tilde{d}$ defines a new ``effective'' dimension), each of which is independently sampled from the Wigner semi-circle distribution, and each eigenenergy is repeated exactly $d / \tilde{d}$ times. We, therefore, 
define the pseudochaotic ensemble of Hamiltonians as 
\be 
\mathcal{E}_{\tilde{d}} \coloneqq \{ U\Lambda U^{\dag} \, : \, \Lambda \sim \tilde{p}_{\tilde{d}}, U \sim \mathcal{U} \}, 
\ee 
which we refer to as the \textit{pseudo-GUE} ensemble, as we demonstrate that it is computationally indistinguishable from the true GUE ensemble. 

\begin{theorem}[Pseudo-GUE is indistinguishable from GUE]\label{th:pseudo-GUE} For any $\tilde{d}=\omega(\poly n)$, the ensemble $\mathcal{E}_{\tilde{d}}$ is computationally indistinguishable from $\mathcal{E}_{\gue}$ given black-box Hamiltonian access.
\end{theorem}
The proof of \cref{th:pseudo-GUE} is in \cref{App:pseudochaoticconstruction}. As we will see in the next section, this ensemble is designed to mimic certain properties of the GUE while exhibiting fundamentally different behavior in key aspects of quantum chaos, hence qualifying as a pseudochaotic ensemble of Hamiltonians.

The ensemble of chaotic Hamiltonians we selected in our construction, namely the GUE ensemble, lacks a fundamental property of many-body Hamiltonians, namely the locality of interactions, and as such, may not qualify as a truly `physical' many-body chaotic ensemble. However, we emphasize that the GUE is merely an example of a set of Hamiltonians that satisfy the necessary conditions for the onset of quantum chaos, exhibiting level repulsion in its spectral statistics, rapid and strong scrambling, and generating highly complex states. That said, the exploration of pseudochaotic Hamiltonians that also exhibit locality of interaction would be a fruitful direction for future research.
\smallskip

\textbf{``Signatures" of quantum chaos?}
Our pseudochaotic ensemble lacks four key signatures of quantum chaos: (i) its spectral statistics show no level repulsion; (ii) the typical operator entanglement never saturates to its maximal value; (iii) 
\emph{out-of-time-ordered correlators} (OTOCs) decay much more slowly than in truly chaotic systems; and (iv) it produces less complex, low-entanglement, and low-magic states. Yet, the fact that our ensemble can still mimic chaos for a computationally bounded observer, or equivalently produce the same emergent phenomena, raises an important question: how can these be considered signatures of quantum chaos? Our results suggest that these four properties are in fact unnecessary conditions for chaos.

\begin{figure}
    \centering
    \includegraphics[width=0.75\linewidth]{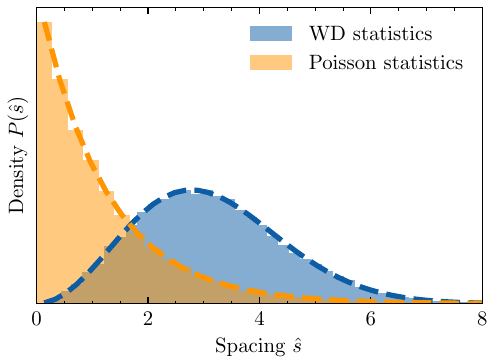}
    \caption{Level spacing statistics for GUE Hamiltonians (blue) and pseudo-GUE Hamiltonians (orange) with $d=2^6$. The dashed curves show a fitted Wigner-Dyson and exponential distribution.} %respectively.}
    \label{fig:spacing}
\end{figure}

At the heart of our investigation are spectral statistics, particularly level repulsion, which is often considered a smoking gun signature for quantum chaos~\cite{bohigas1984characterization,berry1997level}. Level repulsion, a phenomenon where energy levels tend to ``avoid" each other, is a key feature of random matrix theory and is closely tied to the ergodic properties of chaotic systems~\cite{haake2018quantum,guhr1998random,reichl2021transition}.
\begin{theorem}[Dichotomy of spectral statistics]\label{thm:spectrum}
The level spacing distribution $P(s)$ of GUE Hamiltonians obeys $P(\hat{s}=0) = 0$, where $s$ is the difference between any energy and the next largest energy level and $\hat{s} \coloneqq d \cdot s$. In contrast, the level spacing distribution of pseudo-GUE Hamiltonians is monotonically decreasing in $\hat{s}$, and $P(\hat{s}=0) = \frac{8}{3\pi^2}$. \end{theorem}

We present details of the proof of \cref{thm:spectrum} in  \cref{App:behaviorprobestochaos}. The GUE, as expected, exhibits characteristic level repulsion statistics. 
In sharp contrast, our pseudochaotic ensemble completely lacks level repulsion. The intuitive reason for this is that the spacing between $d$ iid random variables is approximated by an exponential distribution~\cite{livan2018introduction}. 
Beside lack of level repulsion, the pseudo-GUE ensemble exhibits weak scrambling behavior (as quantified by the LOE and the OTOC), generates low entangled states (as measured by the second R\'enyi entanglement entropy $S_{2}$), and low-magic states, measured by the second stabilizer entropy $M_2$. 

%The absence of this feature in a system that is computationally indistinguishable from the GUE is remarkable. This suggests that spectral statistics, despite their theoretical importance, may not be sufficient to characterize quantum chaos in physically relevant contexts. 

\begin{theorem}[Separation in probes of quantum chaos]\label{th:separations} Let $\mathcal{P}=\{\operatorname{LOE}, -\log\abs{\OTOC_4}, S_{2}, M_{2}\}$ be the set of probes of quantum chaos defined in \cref{tableprobestochaos}. For any $f\in \mathcal{P}$, $f(H,t)$ quickly grows very large for GUE Hamiltonians: 
\begin{equation}
    \exists t=O(1)\,:\,\Pr_{H \sim \mathcal{E}_{\gue}}[f(H,t) = \Theta(n)] \geq 1 - \exp(-\Omega(n)).
\end{equation}
Conversely, for pseudo-GUE Hamiltonians with $\tilde{d} = \exp(\Theta(\poly \log n))$, we have
\begin{equation}
    \forall t\,:\,\Pr_{H \sim \mathcal{E}_{\pgue}}[f(H,t) \leq O(\poly \log n)] \geq 1 - \negl(n).
\end{equation}
\end{theorem}

These contrasts between the GUE and our pseudochaotic ensemble fundamentally challenge our understanding of quantum chaos. While our pseudochaotic Hamiltonians differ from true chaotic systems in some aspects, they reproduce key features of quantum chaos relevant for many practical applications. For any polynomial-time quantum algorithm, including those measuring common chaos indicators, our systems are indistinguishable from true GUE Hamiltonians. They exhibit the same rapid growth of complexity measures at early times and reproduce emergent phenomena associated with quantum chaos, such as rapid thermalization and apparent randomness of observables, for all practically accessible timescales and measurements.

Traditional indicators like information scrambling or level repulsion are undoubtedly important in many contexts. However, our results suggest that they may not be necessary conditions for a system to exhibit \textit{emergent} chaotic behavior from a computational perspective. The key distinctions between our pseudochaotic ensemble and true chaotic systems only emerge in the long-time, high-precision limit, detectable only with exponential resources. This discrepancy opens up new avenues for research into the nature of quantum chaos. It suggests that there might be a hierarchy of chaotic behaviors, with computational indistinguishability representing a different --- and perhaps more fundamental --- level of chaos than that captured by traditional statistical measures. Moreover, it raises intriguing questions about the role of these traditional chaos indicators in quantum information processing and quantum computation.

Our work implies that many important features of quantum chaos, particularly those relevant to quantum information processing and efficient simulation, can be captured without requiring all traditional signatures of chaos. This opens new avenues for studying and simulating complex quantum systems with significantly reduced computational resources, while still maintaining fidelity to the most practically relevant aspects of quantum chaotic behavior. We refer to \cref{App:behaviorprobestochaos} for the proof of \cref{th:separations}.

\section{Implications of our results}

Having established the existence and properties of pseudochaotic Hamiltonians, we now turn to their broader implications across quantum information science. Our construction not only challenges conventional understanding of quantum chaos but also leads to concrete results in quantum simulation, quantum learning theory, and resource theories. These applications demonstrate how the concept of computational indistinguishability can fundamentally reshape our approach to various quantum information protocols. \smallskip

\textbf{Quantum simulation for GUE Hamiltonians.} It 
has been shown that simulating actual 
GUE Hamiltonians is infeasible even on a quantum computer, as the circuits required to implement (even an $\epsilon$-approximate) time evolution by a GUE Hamiltonian have exponential gate complexity~\cite{kotowski2023extremal}. However, the dynamics and finite-temperature behavior of our ensemble of pseudochaotic Hamiltonians $\mathcal{E}_{\tilde{d}}$ \emph{can} be efficiently simulated. This provides a way to efficiently investigate properties of the GUE that would otherwise be out of 
reach for quantum simulators. 

Our time evolution and Gibbs state simulation algorithms both leverage the special iid form of $\tilde{p}_{\tilde{d}}$. For time evolution, we implement $e^{-iHt}$ by decomposing it into $U e^{-i \Lambda t} U^\dagger$, where $U$ is a pseudorandom unitary that can be implemented using known techniques~\cite{ma2024constructrandomunitaries} and $\Lambda$ is the spectrum we construct. We implement $e^{-i \Lambda t}$ using the ``phase-kickback'' trick~\cite{ji2018pseudorandom}, applying appropriate phases to computational basis states. Crucially, the time evolution algorithm reveals that our pseudochaotic Hamiltonians can be \emph{exponentially} fast-forwarded, with complexity scaling polylogarithmically in $t$. This rare property, known only for a small set of Hamiltonians~\cite{atia2017,gu2021fast}, makes our pseudochaotic ensemble particularly notable. Indeed, since the late-time dynamics of GUE Hamiltonians generate unitaries that are indistinguishable from Haar-random unitaries, the fast-forwardable feature of pseudo-GUE allows us to reach exponentially large times, thereby generating pseudorandom unitaries as well. See \cref{App:GUEHaar} and the discussion therein. Our approach to Gibbs state preparation uses rejection sampling, a classical technique adapted to the quantum setting. The key insight is to propose states from a simple uniform distribution over computational basis states, and then accept or reject them with probabilities determined by their Gibbs weights. While naively this process might seem inefficient due to potentially small acceptance probabilities, we show that for our pseudo-GUE ensemble, the acceptance probability remains sufficiently large to ensure efficient sampling. This is made possible by careful bounds on the ratio between the target Gibbs distribution and our proposal distribution. See \cref{App:simulationpseudo-GUE} for detailed algorithms.

%Crucially, this provides a physically well-motivated application of pseudorandom unitaries.

The capability to efficiently simulate our pseudochaotic ensemble naturally raises the question: given that many properties of the GUE can be calculated analytically (see, 
for example, Ref.~\cite{cotler2017chaos}), what is the value of such simulations? The answer lies in the limitations of analytical methods when dealing with complex scenarios and high-order properties of GUE systems. While low-order moments of the GUE are analytically tractable, higher moments present a substantial challenge, with computational costs scaling factorially. This limitation becomes particularly relevant in several contexts. For instance, in the analysis of finite temperature properties of the GUE, one can use the fact that the Gibbs states of GUE Hamiltonians can be expressed in terms of polynomial moments of the ensemble. Yet, at low temperatures, the polynomial order increases, and analytical evaluation very rapidly becomes prohibitively expensive. Here, direct simulation allows us to access these high-order moments empirically. This approach provides a practical means to explore complex thermal properties that would be excessively difficult %or impossible 
to calculate analytically, offering insights into the behavior of GUE %-like 
systems at finite temperatures.

Moreover, our interest often extends beyond isolated GUE systems to scenarios where GUE elements interact with other systems. Consider, for instance, a 
two-dimensional lattice Hamiltonian composed of a locally interacting background bath, interspersed with small regions governed by GUE dynamics. Such a system might model, for example, a quantum device with regions of controlled interactions punctuated by areas of complex, chaotic behavior. The overall dynamics of these hybrid systems falls outside the realm of current analytical tools. Our simulation method provides a way to explore these complex, heterogeneous systems, offering insights into how local chaos influences global behavior.

These scenarios parallel the development of classical Markov chain Monte Carlo methods, which emerged as practical alternatives to prohibitively expensive analytical evaluations in statistical physics. Our quantum simulation approach offers a means to extract desired behaviors empirically, leveraging quantum pseudorandomness to efficiently simulate GUE-like behavior. While analytical expressions for GUE behavior exist in principle, direct simulation on quantum computers often provides a more practical approach, particularly for complex scenarios involving high-order moments or interactions with non-GUE systems. Our method thus bridges the gap between theoretical understanding and practical exploration of GUE-like systems.
\smallskip

\textbf{Hamiltonian property testing.} 
Our results have implications for Hamiltonian property testing, a task in which one aims to determine whether or not an unknown Hamiltonian possesses a specific property~\cite{bluhm2024hamiltonian}. These implications challenge several commonly held assumptions about quantum systems and reveal fundamental limitations in our ability to infer long-time behavior from early-time dynamics.

A prevailing assumption in the study of quantum systems is that rapid growth of complexity measures (such as entanglement or OTOCs) at early times implies continued growth until saturation at maximal (or minimal) allowed values. Our findings decisively refute this assumption. The pseudochaotic ensemble we introduce exhibits rapid growth of complexity measures at early times, indistinguishable from the growth observed in the GUE. However, contrary to expectations, this growth in our ensemble saturates at $O(\poly \log n)$, falling far short of the $\Omega(n)$ saturation that a naive extrapolation from early-time behavior would suggest. This discrepancy between early and late-time behavior highlights more precise limitations in Hamiltonian property testing. With our construction, we can lower bound the resources required for testing each of the chaos signatures listed in \cref{tableprobestochaos}. In the following, we focus specifically on property testing for OTOCs.

\begin{corollary}[Scrambling property testing]\label{cor:scramblingtesting} Any algorithm $\mathcal{A}^{H}$ with black-box access to a Hamiltonian $H$ which aims to distinguish between whether (i) $ \abs{\OTOC(H,t)}< k$ or (ii) $ \abs{\OTOC(H,t)}\ge K$ requires $\Omega(K^{-1/2})$ queries to the time evolution operator $e^{-iHt}$. This result extends to any algorithm $\mathcal{A}^U$ having black-box access to a unitary operator $U$.
\end{corollary}
The above lower bound, proven in \cref{App:scramblingtesting}
as a direct corollary of \cref{th:pseudo-GUE} and \cref{th:separations}, imposes quantitative limitations on testing the decay of OTOCs for general Hamiltonians $H$, and consequently, for general unitaries as well. 
%Our results demonstrate that it can be very difficult to determine several crucial properties of quantum systems without. For instance, given a general Hamiltonian $H$ and a time $t$ of interest, no efficient algorithm can reliably determine whether the entanglement (or magic) of the state vector $e^{-iHt} \ket{0}$ scales extensively with system size or merely polylogarithmically. Perhaps most strikingly, . This result is particularly significant given the importance of OTOCs in characterizing quantum chaos and information scrambling. 

%Furthermore, our findings reveal a fundamental inability to efficiently determine whether a given Hamiltonian is exponentially fast-forwardable or requires exponential gate complexity for its simulation. This highlights a significant gap in our ability to assess the computational resources needed for quantum simulation tasks, as the difference between exponentially fast-forwardable and exponentially complex Hamiltonians represents an extreme in terms of simulation efficiency.

\smallskip

\textbf{Testing eigenvalues spectrum statistics.} The implications of our results extend beyond individual Hamiltonians to entire ensembles. Given an ensemble of Hamiltonians $\mathcal{E}$ and black-box access to polynomially many samples from this ensemble, \cref{thm:spectrum} shows that no efficient algorithm can determine whether $\mathcal{E}$ follows Wigner-Dyson or exhibits Poissonian level statistics. However, a crucial property of many-body Hamiltonians is the locality of the interaction terms, and it is well-known that GUE Hamiltonians are non-local. Therefore, to make a more meaningful statement, we may consider modifying the setup compared to \cref{th:pseudo-GUE}. While sparsity of Hamiltonians is fundamentally different from locality, it intuitively serves as a necessary condition for locality. Specifically, sparsity, being a basis-dependent property, takes on a meaning akin to locality when the preferred basis is a tensor product basis. Surprisingly, we can strengthen our results for the spectral statistics of eigenvalues and derive the following corollary, proven using ingredients  from \cref{th:pseudo-GUE}. See \cref{App:testingESS} for the formal proof.

\begin{corollary}[Eigenvalues spectrum statistics]\label{cor:ESSth}
Consider an efficient quantum algorithm $\mathcal{A}^{\mathcal{E}}$ which has black-box Hamiltonian access to a polynomial number of Hamiltonians $H$, each sampled at random from a Hamiltonian ensemble $\mathcal{E}$. No such algorithm can successfully distinguish whether the eigenvalue spectrum statistics of $\mathcal{E}$ exhibit level repulsion or not --- this holds even if the Hamiltonians in the ensemble $\mathcal{E}$ are $O(1)$ sparse in the computational basis. %There exist two classes of Hamiltonians, $\mathcal{E}_1$ and $\mathcal{E}_2$, both of which are $O(1)$-sparse in the computational basis and indistinguishable from each other. However, the ESS for $\mathcal{E}_1$ follows Wigner-Dyson (WD) statistics, while the ESS for $\mathcal{E}_2$ follows Poisson statistics. Consequently, the resources required to distinguish between WD and Poisson statistics are superpolynomial, even when it is known that the Hamiltonians in each class are sparse in the computational basis.
\end{corollary}
%Say that it is not surprising, making argument about 2^n energies in O(n) space. And it is just rigorous.
%The proof of \cref{cor:ESSth} follows from \cref{thm:spectrum}. 
The above result implies that we cannot efficiently distinguish between spectral statistics associated with chaotic and integrable systems.

\smallskip

\textbf{Hamiltonian learning.} Our results have implications for the field of Hamiltonian learning, a crucial area in quantum information science that aims to characterize unknown quantum systems. The black-box query access model we assume in our work aligns precisely with the access model typically assumed in Hamiltonian learning protocols~\cite{anshu2021,stilck2024efficient,gu2024practical,bakshi2023learning,bakshi2024structure,huang2023learning}. Traditional Hamiltonian learning algorithms often make an additional physically-motivated assumption: that the Hamiltonian being learned is ``local'' or has a low-degree interaction graph. Our theorem demonstrates that this structural assumption is not only convenient but \emph{absolutely crucial} for the success of these algorithms.

We will say an algorithm $\mathcal{A}^H(t,\epsilon)$ with black-box query access to $H$ has successfully learned a Hamiltonian if for any $\epsilon=\Omega(1/\poly n)$ and a given $t=O(\poly n)$, $\mathcal{A}^H(\ket{\psi})$ outputs a circuit description of any unitary that approximates $e^{-iHt}$ up to an error (measured by the diamond norm) $\epsilon$, with failure probability $o(1)$. 
Almost all existing Hamiltonian learning algorithms succeed by this definition, as they learn an explicit Pauli decomposition of the Hamiltonian, after which they can simply output a Trotterized time evolution circuit~\cite{childs2021thory}.
%to satisfy the above requirements.

\begin{theorem}[No general Hamiltonian learning. Informal version of \cref{th:hardnessdiamond}]\label{th:nohamiltonianlearning}
There is no efficient quantum algorithm $\mathcal{A}^H(t,\epsilon)$ which, given black-box query access to a Hamiltonian $H$, can successfully learn the Hamiltonian $H$ --- this holds even if the circuit which generates the time evolution of $H$ has polynomial size. In fact, this hardness of learning persists even if the Hamiltonian is $O(1)$ sparse in the computational basis.
\end{theorem}
We present details of the proof in \cref{applemmahardnesslearning}.
%See \cref{applemmahardnesslearning} for the proof. 
This finding underscores the importance of structure in quantum systems for their learnability. It suggests that future research in Hamiltonian learning should focus on identifying and leveraging specific structures or symmetries in quantum systems that make them learnable, rather than pursuing general-purpose learning algorithms for arbitrary Hamiltonians. Moreover, our result highlights a noteworthy connection between computational indistinguishability in quantum systems and the limits of quantum learning. It shows that there exist quantum systems which, despite being efficiently implementable, are essentially ``unlearnable''. Our findings parallel recent work in quantum state learning which establishes connections between the hardness of learning quantum states and the existence of quantum cryptographic primitives~\cite{hiroka2024computational}. While their work focuses on state learning, our results extend similar concepts to the realm of Hamiltonian learning, further emphasizing the deep relationship between quantum pseudorandomness and the limitations of quantum learning algorithms.
\smallskip

\textbf{Pseudoresourceful unitaries and tighter bounds on black-box resource distillation.} 
Finally, our results have implications for resource theories. The existence of pseudoresourceful states, such as pseudoentangled~\cite{aaronson2024quantum} and pseudomagic states~\cite{gu2024pseudomagic}, has already posed challenges to the resource theories of entanglement and magic, respectively, when limited to efficient, real-world protocols. Along similar lines, the very existence of pseudochaotic ensembles of Hamiltonians implies even stricter challenges for resource theories.

In fact, the unitary dynamics generated by pseudochaotic Hamiltonians introduce the stronger concept of \textit{pseudoentangling} or \textit{pseudomagic} unitary operators. We can informally define pseudoresourceful unitaries as (a class of) unitary operators that, while indistinguishable from a class of operators producing highly entangled and highly magical states (starting from any computational basis state), actually generate states with low entanglement and low magic.

%the task of resource distillation when limited to efficient, real-world protocols. Along similar lines, the very 
%existence of pseudochaotic ensembles of Hamiltonians implies stricter bounds on resource distillation. 

%The unitary dynamics generated by pseudochaotic Hamiltonians introduce the stronger concept of \textit{pseudoentangling} or \textit{pseudomagic} unitary operators. These are unitary operators that, while indistinguishable from those producing highly entangled and highly magical states, in fact generate states with low entanglement and low magic.

\cref{th:separations} proves the existence of pseudoentangling and pseudomagic unitaries, with a maximal gap in the generated entanglement and magic $O(\poly \log n)$ vs. $\Omega(n)$ (see \cref{App:distillationbounds}). Their existence provides more stringent bounds on resource distillation. The task of resource distillation is to start with multiple copies of a resourceful state, $\rho$, and, depending on the specific details of the resource theory, transform it into a certain number of copies of a pure and useful resourceful state. This task has been extensively studied in the context of entanglement theory, where the goal is to distill clean \textit{ebits} (i.e., perfect Bell pairs)~\cite{PureBipartiteBennett}, and in magic-state resource theory, where such procedures are often referred to as \textit{magic-state factories}, aimed at distilling noiseless magic states for fault-tolerant quantum computation~\cite{Bravyi_2005}.

%The existence of pseudoresourceful states, such as pseudoentangled~\cite{aaronson2024quantum} and pseudomagic states~\cite{gu2024pseudomagic}, has already posed challenges to the task of resource distillation when limited to efficient, real-world protocols. Along similar lines, the very existence of pseudochaotic ensembles of Hamiltonians implies stricter bounds on resource distillation. 

%The unitary dynamics generated by pseudochaotic Hamiltonians introduce the stronger concept of \textit{pseudoentangling} or \textit{pseudomagic} unitary operators. These are unitary operators that, while indistinguishable from those producing highly entangled and highly magical states, in fact generate states with low entanglement and low magic. 
We demonstrate that the (quasi-)exponential suppression of the rate at which distillable states can be produced persists, even if the distiller has access to the unitary operation that prepares the state from some known reference state $\rho_0$ -- a significant advantage compared to 
mere query access to the state.

\begin{corollary}[Stronger distillation bounds]\label{cor:distillationbounds}
Consider a general unitary $U$, a reference state $\rho_0$ and let $\rho \coloneqq U \rho_0 U^\dagger$. Consider any efficient ``designer'' quantum algorithm $\mathcal{D}^U$ with query access to $U$, which outputs a circuit description of an efficient LOCC algorithm $\mathcal{A}_{\locc}$, between two parties $A|B$, for distilling $m$ 
copies of a target state per copy of $\rho$. \cref{th:separations} implies that
\begin{equation}
    m = O(\log^{1+c} S_2(\rho_A)) \label{eq:distill-bound}
\end{equation}
where $\rho_A=\tr_B\rho$. This holds for any constant $c>0$ and every extensive bipartition $A|B$. An analogous statement can be obtained for magic distillation, by replacing $S_2$ with $M_2$ in \cref{eq:distill-bound}, and letting the output of the designer algorithm be an (efficient) stabilizer protocol. 
\end{corollary}
This strengthens existing bounds on general distillation algorithms~\cite{aaronson2024quantum,gu2024pseudomagic,gu2024magicinducedcomputationalseparationentanglement} to include the case where we are allowed to \emph{tailor} our distillation algorithms to the unitaries that generate the resource state that we are given. We discuss the proof in \cref{App:distillationbounds}. \cref{cor:distillationbounds} shows that even this form of prior knowledge is too weak to realize useful distillation algorithms.
\smallskip

{\bf Potential implication for quantum advantage.}
Recent research has highlighted the potential of quantum Gibbs states to achieve quantum advantage in sampling tasks, outperforming classical sampling methods \cite{SupremacyReview, GibbsSampling1, GibbsSampling2}. While sampling from high-temperature Gibbs states can be efficiently accomplished classically \cite{GibbsSampling3, Intensive}, there exist ensembles of Hamiltonians for which this efficiency breaks down for sufficiently large constant values of $\beta$~\cite{GibbsSampling1, GibbsSampling2}. The Hamiltonians we consider in this work are promising candidates for demonstrating such quantum advantage. As shown in \cref{app:sim}, there are efficient quantum algorithms to generate samples from the Gibbs distribution for any $\beta = O(\poly n)$. An intriguing open problem is to prove the hardness of classical sampling for $\beta = O(\poly n)$. This conjecture is plausible, particularly because the Hamiltonians in question exhibit a strong sign problem. Consequently, Monte Carlo sampling techniques applied to these systems would require a sample complexity of $\exp(\Omega(n))$ with respect to the system size $n$, as we demonstrate in~\cref{App:SP}.

\smallskip

\section{Discussion} Our exploration of pseudochaotic Hamiltonians reveals a surprising disconnect between computational indistinguishability and traditional indicators of quantum chaos. This discovery challenges long-held assumptions about the nature of quantum chaos and its relationship to computational complexity, opening up new avenues for research and raising profound questions about our understanding of complex quantum systems.

Our pseudochaotic ensemble provides a novel perspective on the relationship between scrambling, computational complexity, and random matrix universality. It suggests that a comprehensive quantum information-theoretic definition of chaos may need to incorporate computational indistinguishability alongside traditional chaos indicators and measures of information scrambling. The existence of our pseudochaotic ensemble suggests a potential hierarchy within quantum chaotic systems. At the top of this hierarchy might be systems exhibiting all traditional hallmarks of chaos, such as level repulsion and extensive operator entanglement. Our pseudochaotic systems, while lacking these features, occupy a distinct tier characterized by computational indistinguishability from truly chaotic systems. This hierarchy invites us to reconsider what truly defines quantum chaos in a computational context.

One intriguing implication of our work is the potential decoupling of information scrambling from computational complexity. The conventional wisdom that maximal scrambling is necessary for complex quantum behavior is challenged by our results. This decoupling suggests that quantum algorithms or protocols requiring chaotic dynamics~\cite{hu2022hamiltonian,liu2024predicting} might be implementable on a broader class of systems than previously thought.

Looking forward, our work opens up several exciting research directions. Exploring the boundaries of pseudochaos could help isolate which chaotic properties are truly essential for various quantum information processing tasks. The concept might also extend to other quantum phenomena, such as many-body localization or topological order, potentially revealing new insights into these complex behaviors. From a practical standpoint, pseudochaotic systems might offer novel approaches to quantum algorithm design. Finally, the experimental realization of pseudochaotic systems presents an intriguing challenge, providing new tools for quantum simulation while raising questions about distinguishing pseudochaotic from truly chaotic systems in practice.

In conclusion, our work on pseudochaotic Hamiltonians challenges conventional understanding of quantum chaos and its relationship to computational complexity. By introducing a computationally indistinguishable yet structurally distinct ensemble, we highlight the limitations of traditional indicators for quantum chaos. This research opens new avenues for exploring the connections between quantum dynamics, computational indistinguishability, and the practical aspects of quantum simulation, inviting further investigation into these fundamental aspects of quantum systems. Moreover, it raises the philosophical question of what it actually means to observe complex quantum systems when the computationally accessible view may differ so drastically 
from the underlying reality.
\vspace*{-.2cm}

\let\oldaddcontentsline\addcontentsline% Store \addcontentsline
\renewcommand{\addcontentsline}[3]{}% Make \addcontentsline a no-op
%\smallskip

\begin{acknowledgments}
%\paragraph{Acknowledgements.} 
The authors would like to thank Soumik Ghosh for suggesting the search of a physically motivated application of pseudorandom unitaries, Alvaro M.\ Alhambra for discussions regarding Gibbs state preparation, Lennart Bittel and Christian Bertoni for discussions about asymptotic properties of random Hamiltonian ensembles, Silvia Pappalardi for preliminary discussions on the possible role of computationally bounded observers in quantum chaos, Piet Brouwer on notions of quantum chaos, and Nazli Koyluoglu and Varun Menon for discussions on chaos in the quantum many-body setting. Support is also acknowledged from the U.S. Department of Energy, Office of Science, National Quantum Information Science Research Centers, Quantum Systems Accelerator.
%\je{[Any US acknowledgements?]}\ynote{I put in mine but I believe Andi and Susanne have yet to do so}
The Berlin team 
has been supported by the BMBF (FermiQP, MuniQC-Atoms, DAQC), 
the Munich Quantum Valley, the Quantum Flagship (PasQuans2),
the ERC (DebuQC) and 
the DFG (CRC 183, FOR 2724).

\end{acknowledgments}

\bibliographystyle{apsrev4-2}


%\bibliography{refs}
\let\addcontentsline\oldaddcontentsline

\clearpage
\onecolumngrid
\appendix
\tableofcontents

\resumetoc

\section{The Gaussian unitary ensemble}\label{app:GUE}

In this section, we will introduce the \emph{Gaussian unitary ensemble} (GUE) 
as it is commonly discussed in the physics literature in quantum chaos and present some known results from the literature that will be useful for the technical proof in this manuscript.
The Gaussian unitary ensemble%, denoted as GUE, 
is defined by a distribution over $d \times d$ Hermitian matrices $H$, where $d = 2^n$ is the dimension of the Hilbert space. This distribution is defined as follows: the diagonal elements of $H$ are real Gaussian random variables with zero mean and variance $d^{-1}$, while the off-diagonal elements are complex numbers whose real and imaginary parts are independent Gaussian random variables with zero mean and variance $(2d)^{-1}$. The overall probability distribution is
\be
p(H)\propto e^{-\frac{d}{2}\tr(H^2)}.\label{eq:measurePH}
\ee
This is a unitarily invariant measure over Hermitian matrices which defines the Gaussian unitary ensemble. Another common form of the GUE is $p(H)\propto e^{-\frac{1}{2}\tr(H^2)}$; we opt for the one in \cref{eq:measurePH} because the second form results in Hamiltonians which have spectral norm $\Theta(d)$ with high probability. This is unphysical, since one expects energy to scale at most linearly in system size, not exponentially. Thanks to the unitary invariance of \cref{eq:measurePH}, we can factorize the measure as
\begin{equation}
    p(H) = p(U \Lambda U^\dagger) = p(U) p(\Lambda),
\end{equation}
where $p(U)$ is the Haar measure over unitary operators, and $\Lambda=\operatorname{diag}(\lambda_1,\ldots,\lambda_d)$ is a diagonal matrix which contains the spectrum of the Hamiltonian. The distribution over eigenvalues $p(\Lambda)$ follows the well-known GUE eigenvalue distribution
\be
p(\lambda_1,\ldots, \lambda_d) \propto e^{-\frac{d}{2} \sum_{i} \lambda_{i}^{2}} \prod_{i > j} (\lambda_i - \lambda_j)^2\,.\label{eq:eigenvaluedistirbution}
\ee
Therefore, any integral of a generic function $f(H)$ under the measure $p(H)$ \eqref{eq:measurePH}, can be expressed as
\be
\int \dd{H} p(H) f(H) = \int \dd{\Lambda} p(\Lambda) \int_{\haar} \dd{U} g(\Lambda, U),
\ee
where $\int_{\haar} \dd{U}$ is the Haar measure, and $g(\Lambda, U) \coloneqq f(U \Lambda U^\dagger)$. 

In the remainder of the text, we will use $p(\cdot)$ to refer to the distribution over GUE eigenvalues $p(\lambda_1,\ldots,\lambda_d)$ unless otherwise specified. It is useful to define the $k$th marginals
\be
p^{(k)}(\lambda_1,\ldots,\lambda_k)\coloneqq \int \dd{\lambda_{k+1}} \ldots \dd{\lambda_d} p(\lambda_1,\ldots,\lambda_d)
\ee
of the eigenvalue distribution.
Note that since $p(\lambda_1,\ldots,\lambda_d)$ is a permutation invariant distribution, the choice of which eigenvalues to integrate out does not affect the form of the marginal distribution $p^{(k)}$. For $k=1$, one obtains, up to an  $O(d^{-1})$ correction factor, the famous Wigner semi-circle distribution, i.e.,
\be
p^{(1)}(\lambda)=\frac{1}{2\pi}\sqrt{4-\lambda^2}+O(d^{-1})\,.
\ee
In general, we have the following lemma.
\begin{lemma}[Marginals of $p(\lambda_1,\ldots,\lambda_d)$~\cite{gotze_rate_2005,fyodorov2010introductionrandommatrixtheory}]\label{lem:marginals} The $k$th marginal of the eigenvalue distribution in \cref{eq:eigenvaluedistirbution} can be expressed as
    \begin{equation}
    p^{(k)}(\lambda_1,\ldots,\lambda_k) = \frac{(d-k)! (d/\pi)^k}{d!} \det(A+B) + O(d^{-1})\, 
\end{equation}
where $A$ and $B$ are $k\times k$ matrix-valued functions of $\lambda_1,\ldots,\lambda_k$ with components
\begin{equation}
A_{i,i} = \frac{\sqrt{4-\lambda_i^2}}{2} \qc B_{i,j} 
    = \begin{cases}
        0 &\qq{for $i=j$,} \\
        \frac{\sin(d(\lambda_i-\lambda_j))}{d(\lambda_i-\lambda_j)} &\qq{for $i \neq j$.}
    \end{cases}
\end{equation}
\end{lemma}
Due to unitary invariance, it is straightforward to see that sampling $H \sim \text{GUE}$ will, with overwhelming probability, yield a non-local Hamiltonian. In other words, unitary invariance erases any notion of locality. Indeed, as shown in Ref.~\cite{cotler2017black}, while the GUE scrambles information in a time $O(1)$, any local Hamiltonian has a scrambling time lower bounded by $\Omega(\log n)$~\cite{maldacena2016chaos}. That said, although the GUE ensemble violates locality, it serves as an effective \emph{model} which can describe and help us understand relevant behavior of complex systems in condensed matter physics, quantum chromodynamics, nuclear physics, and black hole dynamics.

Recent results have highlighted both the importance and the computational challenges associated with the GUE ensemble in quantum information processing. On one hand, \citet{chen2024efficient} demonstrated that time evolution under two independently sampled GUE matrices for a time $t=O(1)$ is sufficient to generate unitary $k$-designs \cite{Designs}, potentially offering a more efficient method than standard approaches using random quantum circuits. However, this promising result is tempered by a fundamental computational obstacle: implementing the GUE has been shown to require exponential gate complexity. Specifically,~\citet{kotowski2023extremal} proved that the time evolution $e^{-iHt}$ generated by $H \sim \text{GUE}$ exhibits exponential gate complexity for any $t \geq t^{*}$, where $t^{*} = O(1)$. This result demonstrates an \textit{extremal jump in quantum complexity}, severely limiting the practical applicability of GUE-based methods. To be precise, we describe their technical results below. 

We define the $\varepsilon$-unitary complexity as usual. Let $\mathcal{G}$ be a discrete universal gate set, and let $\mathcal{G}^{l}$ be the set of all unitaries built out from $l$ gates from $\mathcal{G}$. The $\varepsilon$-unitary complexity of the unitary $U$ is defined as 
\be
C_{\varepsilon}(U)\coloneqq \min \qty{l : \exists V \in \mathcal{G}^l \ \text{such that} \ \norm{U-V}_\diamond \leq \varepsilon}.
\ee
A similar definition for
(approximate) state complexity is
\be
C_{\varepsilon}(\ket{\psi})\coloneqq \min \qty{l : \exists V \in \mathcal{G}^l \ \text{such that} \ \norm{\ket{\psi} - V \ket{0}} \leq \varepsilon}.
\ee
\begin{lemma}[Exponential gate complexity of GUE~\cite{kotowski2023extremal}]\label{lem:kotowski2023extremal} Let $H$ be sampled from the Gaussian unitary ensemble. Then there exists a time $t^{*}=O(1)$ such that for any $t\ge t^{*}$
\be
C_{\varepsilon}(e^{-iHt})=2^{\Omega(n)}\,,
\ee
for any $\varepsilon^{-1}=O(\poly n)$, with overwhelming probability over the choice of $H$. Similarly, $C_{\varepsilon}(e^{-iHt}\ket{0})=2^{\Omega(n)}$ with overwhelming probability.
\end{lemma}

To address this challenge of implementing the GUE in practice, researchers have explored various alternatives. \citet{chen2024efficient}, for instance, showed that simple Hamiltonians formed from sums of iid Hermitian matrices can mimic the behavior of GUE Hamiltonians. Specifically, they showed that time evolution under two Hamiltonians formed from these random sums results in unitary $t$-designs. These developments form the basis for our motivation in finding pseudo-alternatives for implementing an efficient version of time evolution that is computationally indistinguishable from the GUE. Our work builds on these insights, aiming to bridge the gap between the theoretical power of GUE-based methods and their practical implementability.

Given the exponential gate complexity 
of unitary evolution generated by the GUE, there are strong similarities between these evolutions and Haar-random unitaries, which also exhibit exponential gate complexity. This raises a natural question: can evolutions generated by the GUE resemble the Haar-random distribution for some value of time $t$? This question has already been posed within the framework of unitary $k$-designs—ensembles of unitaries that replicate the first $k$ moments of the distribution induced by the Haar measure over the unitary group with some accuracy. In particular, Ref.~\cite{cotler2017chaos} heuristically argued that for certain times, the GUE might resemble a $k$-design for high $k$. However, two important points should be noted: first, the argument of Ref.~\cite{cotler2017chaos} is mostly heuristic, and second, they argued that at very late times $t \rightarrow \infty$, the GUE would deviate from the Haar measure due to differences in the spectral distribution of the eigenphases between the two ensembles. Below, we clarify this heuristic by rigorously proving the following result: dynamics generated by GUE Hamiltonians are indistinguishable from Haar-random unitaries for \textit{any} time $t$ scaling superpolynomially with the system size. Specifically, any quantum algorithm would need $k = \Omega(t^{1/16})$ samples from either the GUE or Haar unitaries to distinguish between the two. In the language of unitary designs, this means that GUE dynamics form a $\frac{k^2}{t^{1/8}}$-approximate unitary $k$-design. In particular, this implies that for any time $t = \omega(\poly n)$ GUE generates an $o\left(\frac{1}{\poly n}\right)$-approximate unitary $k$-design for $k = O(\poly n)$ (see \cref{Appcor:unitarykdesign}). This result is summarized in the following theorem and proven in \cref{App:GUEHaar}.

\begin{theorem}[\cref{th:GUE=Haarlatetimes} in \cref{App:GUEHaar}]\label{th:GUE=Haarlatetimes2} Consider the ensemble of unitaries generated by GUE, as $\mathcal{U}_{t}=\{e^{-iHt}\,:\, H\sim \gue\}$. Then the ensemble $\mathcal{U}_t$ is adaptively statistically indistinguishable from Haar random unitaries for any $t=\omega(\poly n)$, even if one has access to the adjoint evolution. In particular, any 
quantum algorithm would 
need at least 
\be
k=\begin{cases}
\Omega(t^{1/8}),\quad & t=o(d^4),\\
\Omega(d^{1/4}) &\text{otherwise}\,
\end{cases}
\ee
many queries to the time evolution and/or its inverse to distinguish the two ensembles.
\end{theorem}

\section{Construction of pseudochaotic Hamiltonians}\label{App:pseudochaoticconstruction}
In this section, we rigorously construct the pseudochaotic Hamiltonian following the procedure outlined in the main text, and provide a proof of \cref{th:pseudo-GUE}. Recall that we define the pseudo-GUE ensemble as
\be\label{eq:pgue}
\mathcal{E}_{\tilde{d}} = \{ U \Lambda U^{\dagger} \,\, : \, \Lambda \sim \tilde{p}_{\tilde{d}}, \, U \sim \mathcal{U} \},
\ee
where $\tilde{p}_{\tilde{d}}$ is a permutation-invariant degenerate iid distribution, with each eigenvalue following the Wigner semicircle law $\frac{1}{2\pi} \sqrt{4 - \lambda^2}$. The set $\mathcal{U}$ represents a collection of inverse-secure pseudorandom unitaries, which are employed to mimic the eigenvectors of the GUE ensemble, which are distributed according to the Haar measure. However, although the use of pseudorandom unitaries is sufficient, it is not strictly necessary, as a weaker condition already suffices. We refer the reader to the discussion in \cref{App:GUEHaar} for further details on this point.

The Appendix is organized as follows: in \cref{app:spoof}, we demonstrate that the distribution $\tilde{p}_{\tilde{d}}$, rigorously introduced therein, successfully mimics the correlated distribution $p$ of the GUE spectrum. Specifically, we show that the polynomial marginals of both distributions are close in total variation distance. In \cref{App:Gibbs}, we leverage this result to prove that the Gibbs states generated by the true GUE distribution and those generated by the pseudo-GUE are close in trace distance. Finally, in \cref{app:sim}, we combine these two findings with the use of pseudorandom unitaries to establish \cref{th:pseudo-GUE}.

\subsection{Spoofing the GUE spectrum}\label{app:spoof}
In this section, we prove the claim that a $k$-fold product of Wigner semi-circle distributions closely approximates the true GUE eigenenergy distribution. We write the spoofing product distribution as
\begin{equation}
    \tilde{p}^{(k)}(\lambda_1,\ldots,\lambda_k) = \prod_{i=1}^k p^{(1)}(\lambda_i) = (1/\pi)^k \det(A),
\end{equation}
where $A$ is the matrix defined in \cref{lem:marginals}. 

Our approach is as follows. We want to show that the total variation distance between $p^{(k)}$ and $\tilde{p}^{(k)}$ is small. To do this, we will bound the KL divergence between the two distributions. Since the divergence depends on a term 
\begin{equation}
\frac{\det(A+B)}{\det(A)}= \det(\id + A^{-1} B), 
\end{equation}
we will rely on an identity that bounds the determinant of a perturbed identity matrix. However, we cannot immediately do this because we may not be able to treat $A^{-1} B$ as a perturbation. To do so, we will need to cut out the regions of probability space that are problematic, which will enforce the condition that $A$ must be large and $B$ must be small. The problematic regions consist of two events.
\begin{enumerate}
    \item We require $\abs{\lambda_i - \lambda_j} \geq \frac{1}{\sqrt{d}}$, which ensures that each entry of $\abs{B_{i,j}} \leq \frac{1}{\sqrt{d}}$. We label the event that the condition $\abs{\lambda_i - \lambda_j} \geq \frac{1}{\sqrt{d}}$ \emph{fails} to hold as $E_1$.
    \item We require $\lambda_i^2 \leq 4-\frac{1}{\sqrt{d}}$, which ensures that $(A^{-1})_{i,i} \leq d^{1/4}$. We label the event that the condition $\lambda_i^2 \leq 4-\frac{1}{\sqrt{d}}$ \emph{fails} to hold as $E_2$. This corresponds to regions near the edges of the semi-circle distribution.
\end{enumerate}

\begin{lemma}[Upper bound to probability of event $E_1$]\label{lem:e1}
Under both distributions $p^{(k)}$ and $\tilde{p}^{(k)}$,
\begin{equation}
    \Pr(E_1) \leq \frac{k^2}{\sqrt{d}}.
\end{equation}
\end{lemma}
\begin{proof}
We begin with the $p^{(k)}$ distribution first, to get
\begin{equation}
\begin{aligned}
    \Pr_{\lambda \sim p^{(k)}}(E_1) &\leq \sum_{i=1}^k \sum_{j=i+1}^k \Pr_{\lambda \sim p^{(k)}}\qty(\abs{\lambda_i - \lambda_j} < \frac{1}{\sqrt{d}}) \\
    &= \frac{k(k-1)}{2} \Pr_{\lambda \sim p^{(2)}}\qty(\abs{\lambda_1 - \lambda_2} < \frac{1}{\sqrt{d}}) \\
    &= \frac{d^2}{\pi^2 d(d-1)} \frac{k(k-1)}{2} \int_{-2}^2 \int_{\lambda_1-\frac{1}{\sqrt{d}}}^{\lambda_1+\frac{1}{\sqrt{d}}} \qty[\frac{\sqrt{4-\lambda_1^2}\sqrt{4-\lambda_2^2}}{4} - \frac{\sin^2(d(\lambda_1-\lambda_2))}{(d (\lambda_1-\lambda_2))^2}] \dd[2]{\lambda} \\
    &\leq \frac{4}{\pi^2 \sqrt{d}}\frac{1}{(1-d^{-1})} \frac{k(k-1)}{2} \leq \frac{k^2}{\sqrt{d}}.
\end{aligned}
\end{equation}
The case of $\tilde{p}^{(k)}$ 
follows identically because
\begin{equation}
    \Pr_{\lambda \sim \tilde{p}^{(2)}}\qty(\abs{\lambda_1-\lambda_2} < \frac{1}{\sqrt{d}}) \leq \int_{-2}^2 \int_{\lambda_1-\frac{1}{\sqrt{d}}}^{\lambda_1+\frac{1}{\sqrt{d}}} \frac{\sqrt{4-\lambda_1^2} \sqrt{4-\lambda_2^2}}{4\pi^2} \leq \frac{8}{\pi^2 \sqrt{d}} \leq \frac{1}{\sqrt{d}}.
\end{equation}
\end{proof}

\begin{lemma}[Upper bound 
to probability of event $E_2$]\label{lem:e2}
Under both distributions $p^{(k)}$ and $\tilde{p}^{(k)}$, 
\begin{equation}
    \Pr(E_2) \leq \frac{2k}{d^{1/4}}.
\end{equation}
\end{lemma}
\begin{proof}
The event that any $\lambda^2 \geq 4-\frac{1}{\sqrt{d}}$ is subsumed by the event that $\abs{\lambda} \geq 2-\frac{c}{\sqrt{d}}$ for any $c > \frac{1}{4}$. The following proof uses only $\Pr\qty(\abs{\lambda_1} \geq 2-c/\sqrt{d})$, and since $\lambda_1$ is distributed the same way under both $p$ and $\tilde{p}$, the 
proof
\begin{equation}
    \begin{aligned}
        \Pr(E_2) &\leq k \Pr\qty(\abs{\lambda_1} \geq 2-\alpha)\qc \alpha \coloneqq c/\sqrt{d} \\
    &= k \qty(2\pi \cos^{-1}(1-\alpha/2) - (2-\alpha)\sqrt{4-(2-\alpha)^2}) \\
    &= k \qty(2\pi \cos^{-1}(1-\alpha/2) - (2-\alpha) \sqrt{4\alpha -\alpha^2}) \\
    &\approx k \qty(2\pi \qty(\sqrt{\alpha} + \frac{\alpha^{3/2}}{24}) + \sqrt{\alpha}(2\alpha-4) + o(\alpha^{3/2})) \\
    &\approx 2k (\pi-2) \sqrt{\alpha} + O(\alpha^{3/2}) \\
    &\leq 3k \sqrt{\alpha} = \frac{2k}{d^{1/4}}        
    \end{aligned}
\end{equation}
holds for both distributions.
\end{proof}

Combining \cref{lem:e1,lem:e2}
gives $\Pr(E_1 \lor E_2) \leq \Pr(E_1) + \Pr(E_2) \leq \frac{3k}{d^{1/4}}$. These bounds hold for both $p$ and $\tilde{p}$. Define $p_{\cut}$ as 
\begin{equation}
    p^{(k)}_{\cut}(\lambda_1,\ldots,\lambda_{k}) = \begin{cases}
        \frac{p^{(k)}(\lambda_1,\ldots,\lambda_{k})}{1-\Pr(E_1 \lor E_2)} &\qq{if neither $E_1$ or $E_2$,} \\
        0 &\qq{otherwise,}
    \end{cases}
\end{equation}
with an identical definition for $\tilde{p}^{(k)}_{\cut}$.

\begin{lemma}[Negligible effect of excluding $E_1$ and $E_2$]\label{lem:cut}
Excluding $E_1$ and $E_2$ has a negligible effect on both $p^{(k)}$ and $\tilde{p}^{(k)}$. That is,
\begin{equation}
    D_{KL}(p^{(k)}_{\mathrm{cut}} \parallel p^{(k)}),\; D_{KL}(\tilde{p}^{(k)}_{\mathrm{cut}} \parallel \tilde{p}^{(k)}) \leq \frac{4k}{d^{1/4}}.
\end{equation}    
\end{lemma}
\begin{proof}
We start from
    \begin{equation}
    D_{KL}(p^{(k)}_{\cut} \parallel p^{(k)}) = \int p^{(k)}_{\cut}(\vec{\lambda}) \log(\frac{p^{(k)}_{\cut}(\vec{\lambda})}{p^{(k)}(\vec{\lambda})})\dd[k]{\vec{\lambda}} = -\ln(1-\Pr(E_1 \lor E_2)) \leq \frac{4k}{d^{1/4}}.
\end{equation}
The same proof holds for $D_{KL}(\tilde{p}_{\cut} \parallel \tilde{p})$.
\end{proof}

\begin{lemma}[Closeness of $E_1$ and $E_2$]\label{lem:tilde-rho}
The two distributions with $E_1$ and $E_2$ excluded are close. For any $k \leq d/2$,
we find
    \begin{equation}
        D_{KL}(\tilde{p}^{(k)}_{\mathrm{cut}} \parallel p^{(k)}_{\mathrm{cut}}) \leq \frac{k^4}{\sqrt{d}}.
    \end{equation}
\end{lemma}
\begin{proof}
This proof makes use of the fact that for any real matrix $X$, $\abs{\det(\id + X)} \leq \exp(\tr(X) + \frac{\norm{X}_2^2}{2})$~\cite{RUMP2018101}. Observe that 
    \begin{equation}
    \ln(\frac{p^{(k)}_{\cut}}{\tilde{p}^{(k)}_{\cut}}) = \ln(\frac{(d-k)!d^k}{d!} \frac{\det(A+B)}{\det(A)}) = \ln(\frac{(d-k)! d^k}{d!}) + \ln(\det(\id + A^{-1} B)).
\end{equation}
Since $B_{i,i}=0$, we have $(A^{-1} B)_{i,i} = A^{-1}_{i,i} B_{i,i} = 0$. Finally, over the domain of both $p^{(k)}_{\cut}$ and $\tilde{p}^{(k)}_{\cut}$ (i.e., whenever $\lnot E_1 \land \lnot E_2$), $\abs{A^{-1} B}_{i,j} \leq \frac{1}{\sqrt{4-\lambda_i^2}} \sum_j \abs{B_{i,j}} \leq \frac{k}{\sqrt{d}\sqrt{4-\lambda_i^2}} \leq \frac{k}{d^{1/4}}$, so $\norm{A^{-1}B}_2^2 \leq \frac{k^4}{\sqrt{d}}$, and $\abs{\ln(\det(\id + A^{-1} X))} \leq \frac{k^4}{2\sqrt{d}}$.
In this way, we get
\begin{equation}
    \begin{aligned}
        D_{KL}(p^{(k)}_{\cut} \parallel \tilde{p}^{(k)}_{\cut}) &= \int p^{(k)}_{\cut}(\vec{\lambda}) \ln(\frac{p^{(k)}_{\cut}(\vec{\lambda})}{\tilde{p}^{(k)}_{\cut}(\vec{\lambda})}) \dd[k]{\vec{\lambda}} \\
        &\leq \frac{k^4}{2\sqrt{d}} + \ln(\frac{(d-k)! d^k}{d!}) \\
        &= \frac{k^4}{2\sqrt{d}} - \sum_{j=0}^{k-1} \ln(1-\frac{j}{d})
        \\
        &\leq \frac{k^4}{2\sqrt{d}} + \frac{k(k-1)}{d} \leq \frac{k^4}{\sqrt{d}},
    \end{aligned}
\end{equation}
where we have used the fact that $-\ln(1-x) \leq 2x$ for any $x \leq \frac{1}{2}$ in the second last inequality (this relies on the assumption $k/d \leq 1/2$).
\end{proof}

\begin{theorem}[Spoofing the GUE spectrum]\label{thm:spoof}
For any $k=O(\poly n)$, the two distributions $p^{(k)}$ and $\tilde{p}^{(k)}$ are close 
in total variation distance as
\begin{equation}
    \delta(p^{(k)},\tilde{p}^{(k)}) \leq \sqrt{\frac{k^4}{d^{1/2}}} + 2 \sqrt{\frac{2k}{d^{1/4}}} \leq \frac{3\sqrt{k}}{d^{1/8}}.
\end{equation}
\end{theorem}
\begin{proof}
The total variation distance between any two distributions is upper bounded by 
\begin{equation}\delta(P,Q) \leq \sqrt{\frac{1}{2} D_{KL}(P \parallel Q)}.
\end{equation}
Using a triangle inequality for the sequence $p^{(k)} \overset{\ref{lem:cut}}{\to} p^{(k)}_{\cut} \overset{\ref{lem:tilde-rho}}{\to} \tilde{p}^{(k)}_{\cut} \overset{\ref{lem:cut}}{\to} \tilde{p}^{(k)}$, we get
\begin{equation}
    \delta(p^{(k)},\tilde{p}^{(k)}) \leq \sqrt{\frac{k^4}{d^{1/2}}} + 2 \sqrt{\frac{2k}{d^{1/4}}} \leq \frac{3\sqrt{k}}{d^{1/8}}.
\end{equation}
\end{proof}
This is similar in spirit to a more general theorem, which says that for any distribution 
over $d$ exchangeable variables, the $k$-marginal distribution can be approximated by a mixture of \emph{independent and identically distributed} (iid) 
variables up to total variation distance $O(k^2/d)$~\cite{diaconis1980finite}. However, 
this more general result does not say what the iid distribution is; here, we explicitly construct the spoofing distribution.
Now, recall the definition of our ``degenerate'' distribution $\tilde{p}_{\tilde{d}}$: it 
is a permutation invariant distribution where each eigenenergy has degeneracy $d/\tilde{d}$, so there are exactly $\tilde{d}$ unique eigenenergies, each of which is distributed independently according to the Wigner semi-circle distribution. We are now ready to show that the degenerate distribution $\tilde{p}_{\tilde{d}}$ can also spoof $p$. 
\begin{theorem}[Spoofing the GUE spectrum, with degeneracies]\label{lem:degen}
    For any $k=O(\poly n)$ and $\tilde{d} = \omega(\poly n)$, the true GUE eigenenergy distribution is close to the degenerate spoofing distribution in total variational distance:
    \begin{equation}\label{eq:degen-close}
        \delta(p^{(k)}, \tilde{p}_{\tilde{d}}^{(k)}) \le \frac{k^2}{\tilde{d}} + \frac{3 \sqrt{k}}{d^{1/8}}= \negl(n).
    \end{equation}
\end{theorem}
\begin{proof}
We will first show that the degenerate and non-degenerate spoofing distributions are close in the sense of
\begin{equation}
    \delta(\tilde{p}^{(k)}, \tilde{p}_{\tilde{d}}^{(k)}) < \negl(n),
\end{equation}
after which \cref{eq:degen-close} will follow using \cref{thm:spoof} in a simple triangle inequality. We use a similar trick of excluding certain low probability events from the support of $\tilde{p}_{\tilde{d}}^{(k)}$. Define a ``cut-out'' version of $\tilde{p}_{\tilde{d}}^{(k)}$ which excludes the event $E$ that any pair $\lambda_i=\lambda_j$ for $i \neq j$. Observe that this distribution is \emph{exactly} equal to $\tilde{p}^{(k)}$, since when there are no duplicated energies, the two distributions are exactly the same. The probability of $E$ \emph{not} occurring (i.e., all the energies are unique, meaning they are not sampled from the same degenerate subspace) is simply $\prod_{j=1}^{k-1} \qty(1-\frac{j}{\tilde{d}})$. Since $k=O(\poly n) < \frac{\tilde{d}}{2}$, we apply the bound $1-x \geq e^{-2x}$ for $x \leq \frac{1}{2}$ (i.e., $1-\frac{j}{\tilde{d}} \geq e^{-2j/\tilde{d}}$), so that
\begin{equation}
    \prod_{j=1}^{k-1} \qty(1-\frac{j}{\tilde{d}}) \geq e^{-k^2/\tilde{d}} \geq 1 - \frac{k^2}{\tilde{d}}.
\end{equation}
Therefore, the probability of $E$ is at most $\frac{k^2}{\tilde{d}}$, meaning that
\begin{equation}\label{Eq.th6}
    \delta(\tilde{p}^{(k)}, \tilde{p}_{\tilde{d}}^{(k)}) \leq \frac{k^2}{\tilde{d}} < \negl(n),
\end{equation}
and a triangle inequality shows that $\delta(p^{(k)}, \tilde{p}_{\tilde{d}}^{(k)}) \leq \frac{k^2}{\tilde{d}} + \frac{3\sqrt{k}}{d^{1/8}}$. 
\end{proof}

\subsection{Indistinguishability of Gibbs states}\label{App:Gibbs}

We now turn to showing that the respective Gibbs states 
are indistinguishable for computationally bounded observers.

\begin{lemma}[Bounds to 
moments of partition functions]\label{lem:z-bound}
For any $\beta=O(\poly n)$, the first two moments of $\tr(e^{-\beta H})$ are
\begin{subequations}
    \begin{gather}
        \mathbb{E}_{\Lambda \sim p}[\tr(e^{-\beta \Lambda})] = \frac{d I_1(2\beta)}{\beta} ,\\
        \mathbb{E}_{\Lambda \sim p}[\tr^2(e^{-\beta \Lambda})] = \frac{d^2 I_1(2\beta)^2}{\beta^2} \qty(1 + O(\beta^3 d^{-1/8})),
    \end{gather}
\end{subequations}
where $I_1(x) \coloneqq \frac{1}{\pi} \int_0^\pi e^{x \cos \theta} \cos \theta \dd{\theta}$ is the modified Bessel function of the first kind. 
\end{lemma}
\begin{proof}
For the first moment, it follows from 
the fact that $\frac{1}{2\pi} \int_{-2}^2 \sqrt{4-\lambda^2} e^{-\beta \lambda} \dd{\lambda} = \frac{I_1(2\beta)}{\beta}$, where $I_1$ is the modified Bessel function of the first kind. Since $\mathbb{E}_{\Lambda}[\tr^2(e^{-\beta \Lambda})]$ depends only on the marginal 
distribution $p^{(2)}(\lambda_i,\lambda_j)$, 
we have
\begin{equation}
    \abs{\mathbb{E}_{\Lambda \sim p}[\tr^2(e^{-\beta \Lambda})] - \mathbb{E}_{\Lambda \sim \tilde{p}}[\tr^2(e^{-\beta \Lambda})]} \leq d^2 e^{4\beta} \cdot \delta(p^{(2)},\tilde{p}^{(2)}) \leq \frac{3\sqrt{2} d^2 e^{4\beta}}{d^{1/8}}.
\end{equation}
Also, we find
\begin{equation}
\begin{gathered}
    \mathbb{E}_{\Lambda \sim \tilde{p}}[\tr^2(e^{-\beta \Lambda})] = d \mathbb{E}_{\lambda \sim p^{(1)}}[e^{-2\beta \lambda}] + d (d-1) \mathbb{E}_{\lambda \sim p^{(1)}}[e^{-\beta \lambda}]^2\qq{while} \\
    \qty(\mathbb{E}_{\Lambda \sim p}[\tr(e^{-\beta \Lambda})])^2 = d^2 \mathbb{E}_{\lambda \sim p^{(1)}}[e^{-\beta \lambda}]^2,
\end{gathered}
\end{equation}
so that $\abs{\qty(\mathbb{E}_{\Lambda \sim p}[\tr(e^{-\beta \Lambda})])^2 - \mathbb{E}_{\Lambda \sim \tilde{p}}[\tr^2(e^{-\beta \Lambda})]} = d\qty(\mathbb{E}[e^{-\beta \lambda}]^2 + \mathbb{E}[e^{-2\beta \lambda}]) \leq 2de^{4\beta}$. From this, it follows that 
\begin{equation}
    \frac{\text{var}(\tr(e^{-\beta H}))}{\qty(\mathbb{E}[\tr(e^{-\beta H})])^2} \leq \frac{3 \sqrt{2} e^{4\beta}}{d^{1/8} I_1(2\beta)^2/\beta^2} + \frac{2de^{4\beta}}{d^2 I_1(2\beta)^2/\beta^2} = O\qty(\beta^3 d^{-1/8}).
\end{equation}
\end{proof}
\begin{fact}
The modified Bessel function $I_1(x)$ obeys the asymptotic behaviour
\begin{equation}
    I_1(x) = \frac{e^x}{\sqrt{2\pi x}} \qty(1 - \frac{3}{8x} - \frac{15}{128x^2} - O(x^{-3})).
\end{equation}
For $x \geq e$, it can be upper and lower bounded by $\frac{e^x}{2\sqrt{x}}$ and $\frac{e^x}{3\sqrt{x}}$~\cite{abramowitz1965handbook}.
\end{fact}

\begin{lemma}[Closeness of matrix exponential]\label{lem:exp-approx}
    For any $\beta,\epsilon \geq 0$, if $m \geq 32\beta + \ln \epsilon^{-1}$ then
    \begin{equation}
        \sup_{E \in [-2,2]} \abs{e^{-\beta E} - \sum_{j=0}^m \frac{(-\beta E)^j}{j!}} \leq \epsilon.
    \end{equation}
\end{lemma}
\begin{proof}
Taylor's theorem says that the error of the degree $m$ Taylor expansion of $e^{-\beta E}$ is at most $\frac{e^{2\beta} (2\beta)^{m+1}}{(m+1)!}$. Using the bound $(m+1)! \geq (m/e)^{m+1}$, we require 
that 
   \begin{equation}
\qty(\frac{2\beta e}{m})^{m+1} \leq \epsilon e^{-2\beta}.
    \end{equation}
    We see that if $m \geq 4 \beta e^2$ (a stronger condition being $m \geq 30\beta$), then $\qty(\frac{2\beta e}{m})^{m+1} \leq e^{-(m+1)}$, in which case it suffices to have $m+1 \geq \ln \epsilon^{-1} + 2\beta$. Therefore, $m \geq \max(30\beta, \ln \epsilon^{-1} + 2\beta)$ suffices to achieve error at most $\epsilon$.
\end{proof}

\begin{lemma}[Taylor bound]\label{lem:1/1+x}
    For any $\beta,\epsilon \geq 0$, and any $a \in [0,1)$, if $m \geq \ln(a/\epsilon)$,
    \begin{equation}
        \sup_{x \in [0,a]} \abs{\frac{1}{1+x} - \sum_{j=0}^m (-1)^j x^j} \leq \epsilon .
    \end{equation}
\end{lemma}
\begin{proof}
Applying Taylor's theorem again, the error of the degree-$m$ Taylor expansion of $\frac{1}{1+x}$ is simply $a^{m+1}$. Therefore $m \geq \ln(a/\epsilon)$ suffices to guarantee the error is at most $\epsilon$.
\end{proof}

\begin{lemma}[Closeness of distributions]\label{lem:q-approx}
For any $\beta$, there exists two polynomials $f: \mathbb{R} \to \mathbb{R}$ and $g: \mathbb{R}^d \to \mathbb{R}$ with total degree $m=O(\poly(\beta,n))$, such that the two distributions
\begin{equation}
    (G(\Lambda))_i \coloneqq \frac{e^{-\beta \lambda_i}}{\tr(e^{-\beta \Lambda})} \qq{and} (\tilde{G}_m(\Lambda))_i \coloneqq f(\lambda_i) g(\lambda_1,\ldots,\lambda_d)
\end{equation}
are close in total variation distance: $\delta(G(\Lambda), \tilde{G}_m(\Lambda)) = O(d^{-1})$ with failure probability at most $O(d^{-1/2} \beta^3)$.
\end{lemma}
\begin{proof}
We approximate both $e^{-\beta \lambda_i}$ and $\frac{1}{Z} \coloneqq \frac{1}{\tr(e^{-\beta \Lambda})}$ with $\poly(\beta,n)$-degree polynomials, and show that their errors can be individually made very small. 
We start with $\frac{1}{Z}$, and here, we begin with
%
%We start by 
writing 
\begin{equation}
\frac{1}{\tr(e^{-\beta H})} = \frac{(\beta/d I_1(2\beta))}{\frac{\beta}{d I_1(2\beta)} \tr(e^{-\beta H})}.
\end{equation}
Define $\tilde{Z} \coloneqq \frac{\beta}{d I_1(2\beta)} \tr(e^{-\beta H})$. Note that $\mathbb{E}[\tilde{Z}]=1$. By Chebyshev's inequality, 
\begin{equation}
    \Pr\qty(\abs*{\tilde{Z} - 1} \geq d^{-1/2}) \leq O(d^{-1/2} \beta^3).
\end{equation}
We will therefore assume $\abs*{\tilde{Z}-1} \leq d^{-1/2}$, and allow our approximation to fail otherwise. Now, we write 
\begin{equation}
\qty(\frac{\beta}{d I_1(2\beta)} \tr(e^{-\beta H}))^{-1} = \qty(1+(\tilde{Z}-1))^{-1}. 
\end{equation}
From \cref{lem:1/1+x}, up to an error $\epsilon_1$, this can be approximated by $\sum_{j=0}^{m_1} (-1)^j (\tilde{Z}-1)^j$ for $m_1 = O(\ln \epsilon_1^{-1} - n)$. Now, since $\abs*{\tilde{Z}-1} \leq d^{-1/2}$ by assumption, we can approximate $(\tilde{Z}-1)^j$ to error $\epsilon_2/m_1$ so long as we can approximate $\tilde{Z}$ to error within $\frac{\epsilon_2}{2j m_1}$. This is possible if we approximate each of the summands $e^{-\beta \lambda_i}$ in $Z$ using a Taylor series with error at most $\frac{I_1(2\beta)}{\beta} \frac{\epsilon_2}{2jm_1}$. A sufficient condition is to achieve an error $\frac{e^{2\beta} \epsilon_2}{9m_1^2 \beta^{3/2}}$. From \cref{lem:exp-approx}, this requires a degree
\begin{equation}
    m_2 \leq 32\beta + \ln(9 m_1^2 \beta^{3/2}) + \ln(\epsilon_2^{-1}) \leq O(\beta + \ln \ln \epsilon^{-1} + \ln \epsilon_2^{-1}).
\end{equation}
Setting $\epsilon_1=\epsilon_2=d^{-2} e^{-2\beta}/2$, we see that $m_1,m_2=O(n+\beta)$ is sufficient. This defines the function $g(\lambda_1,\ldots,\lambda_d)$.

Having approximated the denominator, we need to approximate the Gibbs factor $e^{-\beta \lambda_i}$ to within error $\epsilon_3 = d^{-1}$. Again, from \cref{lem:exp-approx}, this requires $m_3=O(n + \beta)$. Simple algebra shows that the overall error is then at most 
\begin{equation}
    (\epsilon_1+\epsilon_2) \cdot \epsilon_3 + \frac{\epsilon_3}{Z} + (\epsilon_1+\epsilon_2) e^{2\beta}.
\end{equation}
Since $Z \geq \frac{I_1(2\beta) d}{\beta} \qty(1 - d^{-1/2}) \geq \frac{e^{2\beta} d}{6\beta^{3/2}}$, the overall error $\epsilon$ and polynomial degree $m$ is
\begin{subequations}
    \begin{gather}
        \epsilon \leq d^{-3} e^{-2\beta} + 6 \beta^{3/2} d^{-2} e^{-2\beta} + d^{-2} \leq O(d^{-2}), \\
        m \leq m_1 m_2 m_3 = O((n+\beta)^3).
    \end{gather}
\end{subequations}
Therefore, each $\abs{f(\lambda_i) g(\lambda_1,\ldots,\lambda_d) - \frac{e^{-\beta \lambda_i}}{\tr(e^{-\beta \Lambda})}} \leq O(d^{-2})$, with $f$ and $g$ both having degree $O(\poly(\beta,n))$. This allows us to conclude that the total variation distance between $G$ and $\tilde{G}_m$ is $O(d^{-1})$.
\end{proof}

\begin{theorem}[Indistinguishability of Gibbs states]\label{thm:gibbs}
For any $\beta$ and any $k$,
\begin{equation}
    \norm{\mathbb{E}_{\Lambda \sim p}\qty[\qty(\frac{e^{-\beta \Lambda}}{\tr(e^{-\beta \Lambda})})^{\otimes k}] -  \mathbb{E}_{\Lambda \sim \tilde{p}}\qty[\qty(\frac{e^{-\beta \Lambda}}{\tr(e^{-\beta \Lambda})})^{\otimes k}]}_1 \leq O\qty(\frac{\poly(k,\beta,n)}{d^{1/8}}).
\end{equation}
\end{theorem}
\begin{proof}
We show that each approximation in the following incurs a small error: $\mathbb{E}_{\Lambda \sim p}[G(\Lambda)^{\otimes k}] \approx \E_{\Lambda \sim p}[\tilde{G}_m(\Lambda)^{\otimes k}] \approx \E_{\Lambda \sim \tilde{p}}[\tilde{G}_m(\Lambda)^{\otimes k}] \approx \E_{\Lambda \sim \tilde{p}}[G(\Lambda)^{\otimes k}]$. By \cref{lem:q-approx}, for some degree $m=O(\poly(\beta,n))$,
\begin{equation}
    \norm{\mathbb{E}_{\Lambda \sim p}\qty[G(\Lambda)^{\otimes k} - \tilde{G}_m(\Lambda)^{\otimes k}]}_1 \leq O(k \cdot d^{-1} + \beta^3 d^{-1/2}),
\end{equation}
which holds for $\Lambda \sim \tilde{p}$ as well. This takes care of the first and last approximations. For the second approximation,
\begin{equation}
    \begin{aligned}
        \norm{\mathbb{E}_{\Lambda \sim p}[\tilde{G}_m(\Lambda)^{\otimes k}] - \mathbb{E}_{\Lambda \sim \tilde{p}}[\tilde{G}_m(\Lambda)^{\otimes k}]}_1 &\leq \int \norm{\tilde{G}_m(\vec{\lambda})^{\otimes k}}_1 \abs{p^{(mk)}(\vec{\lambda}) - \tilde{p}^{(mk)}(\vec{\lambda})} \dd[(mk)]{\vec{\lambda}} \\
        &\leq O\qty(\frac{\sqrt{m k}}{d^{1/8}}) \\
        &\leq O\qty(\frac{\poly(k,\beta,n)}{d^{1/8}}).
    \end{aligned}
\end{equation}
\end{proof}

\subsection{Pseudo-GUE is indistinguishable from GUE}\label{app:sim}

%\je{[Here it would be important to state more clearly what the efficient simulation of pseudo-GUE Hamiltonians refers to, concerning sample and computational complexity, the assumed oracle access, and all that.]}
In this section, we proceed to prove \cref{th:pseudo-GUE}. To begin, we rigorously define the black-box Hamiltonian access model, which is used throughout the paper and outlined in the main text.

\begin{definition}[Black-box Hamiltonian access]\label{def:black-box-ham}
An algorithm has black-box Hamiltonian access to a Hamiltonian $H$ if it has query access to
\begin{itemize}[noitemsep]
    \item samples of the Gibbs states $e^{-\beta H}/\tr(e^{-\beta H})$ for any $\beta \leq O(\poly n)$, and
    \item $e^{-iH t}$ (as well as its controlled version) for any $\abs{t} \leq \exp(O(n))$. However, we will not allow the time $t$ to be controlled coherently, only the application or non-application of the unitary $e^{-iHt}$.
\end{itemize}
We will denote algorithms $\mathcal{A}$ with black-box Hamiltonian access as $\mathcal{A}^H$. We will say that an algorithm $\mathcal{A}^H$ makes one query to $H$ when it either accesses $e^{-iHt}$ (or a controlled version thereof), or requests one sample of a Gibbs state at $\beta \leq O(\poly n)$.
\end{definition}

Finally, we are ready to prove one of our main results, \cref{th:pseudo-GUE}. Simply stated, it says that Hamiltonians from the ensemble $\pgue$, see Eq.~\eqref{eq:pgue}, are indistinguishable from GUE Hamiltonians.
\begin{theorem}[Hamiltonian indistinguishability]\label{th:ham-indisting}
For any $\tilde{d} = \omega(\poly n)$, the two ensembles of Hamiltonian $\gue$ and $\pgue$ are statistically indistinguishable, meaning any quantum algorithm $\mathcal{A}^H$ which queries $H$ at most a \emph{polynomial} number of times obeys
\begin{equation}
    \abs{\Pr_{U \sim \gue}[\mathcal{A}^H() = 1] - \Pr_{U \sim \pgue}[\mathcal{A}^H() = 1]} < \negl(n).
\end{equation}
\end{theorem}
\begin{proof}
Let us first focus on algorithms which only make queries to the time evolution operators $U$. We will assume our initial state is always some fixed state $\rho_0$, consisting of $n$ system qubits and $n'=O(\poly n)$ auxiliary registers. We will later allow for the possibility of algorithms which request $m=O(\poly n)$ Gibbs states simply by adding $n \cdot m$ more auxiliary systems, each initialized in a Gibbs state at some temperature $\beta=O(\poly n)$. We can freely assume that all of the Gibbs states are requested at the beginning, because they can be injected into the algorithm at any time with swap operators.

Returning to the time-evolution only case, the output state of $\mathcal{A}^H$ can be written
\begin{equation}
    \rho(\Lambda, V; \rho_0) \coloneqq \qty[W_{k+1} (\text{C}^{n_k}U(t_k)) W_{k} (\text{C}^{n_{k-1}}U(t_{k-1})) W_{k-2} (\text{C}^{n_2}U(t_2)) \ldots W_2 (\text{C}^{n_1}U(t_1)) W_1](\rho_0),
\end{equation}
where $k, n'' = O(\poly n)$, $(\text{C}^{n_j}U(t_j))$ is a controlled version of $U=Ve^{-i \Lambda t_j} V^\dagger$, and the $n_j$th auxiliary register is the control qubit. Each of the $W_j$ is an efficient unitary acting jointly on $n+n'=O(\poly n)$ qubits. This subsumes the case where the unitaries $U$ are not controlled, because we can simply manually initialize the $n_k$th auxiliary qubit to be in the state vector $\ket{1}$ so that $\text{C}^{n_k} U = U$. We note that $\text{C}^n U$ can be written $(\text{C}^n V)(\text{C}^n e^{-i \Lambda t})(\text{C}^n V^\dagger)$. We will absorb the controlled Haar random unitaries into each of the $W_j$ unitaries, so for a fixed Haar random unitary $V$ we can also rewrite the output state as
\begin{equation}\label{eq:psi-lamb-v}
    \rho(\Lambda, V; \rho_0) = \qty[W_{k+1}'' (\text{C}^{n_k}e^{-i \Lambda t_k}) W_{k-1}'' \ldots W_2'' (\text{C}^{n_1}e^{-i \Lambda t_1}) W_1'](\rho_0),
\end{equation}
where $W_{k+1}'' = W_{k+1} (\text{C}^{n_k}V)$, $W_j'' = (\text{C}^{n_j} V)^\dagger W_j (\text{C}^{n_{j-1}}V)$, and $W_1'' = (\text{C}^{n_1} V^\dagger) W_1$. Observe that $\text{C}^{n_k} e^{-i\Lambda t}$ can simply be written $\ketbra{0}_{n_k} \otimes \id + \ketbra{1}_{n_k} \otimes e^{-i \Lambda t}$. Crucially, since $p$ and $\tilde{p}_{\tilde{d}}$ are both permutation invariant, the expectation of $\rho(\Lambda,V; \rho_0)$ with respect to $\Lambda$ \emph{depends only on the $2k$-marginal distributions $p^{(2k)}$ and $\tilde{p}^{(2k)}_{\tilde{d}}$}. This implies that
\begin{equation}\label{eq:psi-v-bound}
    \norm{\E_{\Lambda \sim p}[\rho(\Lambda,V; \rho_0)] - \E_{\Lambda \sim \tilde{p}}[\rho(\Lambda,V; \rho_0)]}_1 \leq \delta(p^{(2k)},\tilde{p}^{(2k)}_{\tilde{d}}) \cdot \max_{\Lambda}\norm{\rho(\Lambda,V; \rho_0)}_1 \leq \negl(n).
\end{equation}
Since this holds for any $V$, we can let $V$ be sampled from an inverse-secure pseudorandom unitary ensemble~\cite{ma2024constructrandomunitaries}, which will imply that the expected output states of $\mathcal{A}^U$, for $U \sim \mathcal{E}_{\gue}$ and $U \sim \pgue$ are $\negl(n)$ close in trace distance. However, we note that the use of an inverse-secure pseudorandom unitary may be an overly strong assumption, as a weaker condition is actually sufficient. We refer the interested reader to \cref{App:GUEHaar} for further details.

Now, to return to the original setting in which we also have access to Gibbs states, the only modification is that now we are interested in the quantity
\begin{equation}
    \norm{\E_{\Lambda \sim p}[\rho(\Lambda,V; \rho_0 \otimes \sigma_0)] - \E_{\Lambda \sim \tilde{p}}[\rho(\Lambda,V; \rho_0 \otimes \sigma_0')]}_1,
\end{equation}
where $\sigma_0$ and $\sigma_0'$ are $m=O(\poly n)$ copies of Gibbs states for the Hamiltonian $V\Lambda V^\dagger$. A simple triangle inequality shows that this can be upper bounded by
\begin{equation}
    \underbrace{\norm{\E_{\Lambda \sim p}[\rho(\Lambda,V; \rho_0 \otimes \sigma_0)] - \E_{\Lambda \sim \tilde{p}}[\rho(\Lambda,V; \rho_0 \otimes \sigma_0)]}_1}_{<\negl(n) \,\text{by \cref{eq:psi-v-bound}}} + \underbrace{\norm{\E_{\Lambda \sim p}[\rho_0 \otimes \sigma_0] - \E_{\Lambda \sim \tilde{p}}[\rho_0 \otimes \sigma_0']}_1}_{<\negl(n)\,\text{by \cref{thm:gibbs}}} < \negl(n).
\end{equation}
This completes the proof.
% Alternative proof: integrate over lambda immediately *without* separating into two summands via triangle inequality.
\end{proof}
\section{Efficient simulation of of pseudo-GUE Hamiltonians}\label{App:simulationpseudo-GUE}
In this section, we describe efficient quantum algorithms for simulating evolution under the pseudo-GUE up to exponential times $t$.

We first remind the reader of the notation used throughout this manuscript. We work with several different distributions over eigenvalues, each with specific properties:
\begin{itemize}[noitemsep]
    \item The eigenvalue distribution of the GUE, denoted as $p(\lambda_1,\ldots,\lambda_{d})$ and defined in \cref{eq:eigenvaluedistirbution}, exhibits complex correlations between its eigenvalues. We abbreviate this distribution simply as $p$.
    
    \item The product distribution $\tilde{p}(\lambda_1,\ldots,\lambda_d)\coloneqq\prod_{i=1}^{d}p^{(1)}(\lambda_i)$, where $p^{(1)}$ is the first marginal of $p$ (the Wigner semicircle law). This distribution lacks the correlations present in $p$, as each eigenvalue is sampled independently.
    
    \item A degenerate version $\tilde{p}_{\tilde{d}}$ which has only $\tilde{d}$ unique energies (where typically $\omega(\poly n) < \tilde{d} \ll d$). Each unique energy appears exactly $d/\tilde{d}$ times, and the unique energies are sampled independently from $p^{(1)}$.
    
    \item A computationally efficient version $\tilde{p}_{\tilde{d},k}$ where the $\tilde{d}$ unique energies are only $k$-wise independent (with $k=O(1)$).
\end{itemize}

The concept of $k$-wise independence is crucial: while fully independent random variables require that any subset of variables be independent, $k$-wise independence only requires this for subsets of size at most $k$. Formally, for any indices $i_1,\ldots,i_m$ with $m\leq k$, the joint distribution of $(\lambda_{i_1},\ldots,\lambda_{i_m})$ must match that of $m$ independent copies of $p^{(1)}$. This weaker notion of independence allows for much more efficient implementations.

Correspondingly, we study three main ensembles of Hamiltonians:
\begin{subequations}
\begin{align}
    \gue &= \text{The GUE ensemble (see \cref{app:GUE})} \\
    \mathcal{E}_{\tilde{d}} &\coloneqq \{U\Lambda U^{\dag}\,:\, U\sim \mathcal{U}\,,\, \Lambda\sim \tilde{p}_{\tilde{d}}\} \\
    \mathcal{E}_{\tilde{d},k} &\coloneqq \{U\Lambda U^{\dag}\,:\, U\sim \mathcal{U}\,,\, \Lambda\sim \tilde{p}_{\tilde{d},k}\}
\end{align}    
\end{subequations}
where $\mathcal{U}$ is an ensemble of pseudorandom unitaries. Crucially, only $\mathcal{E}_{\tilde{d},k}$ can be efficiently implemented on a quantum device. The fully independent ensemble $\mathcal{E}_{\tilde{d}}$ requires exponential gate complexity because sampling $\tilde{d}$ truly independent random variables requires $\Omega(\tilde{d})$ resources, and we will always assume $\tilde{d}=\omega(\poly n)$. The $k$-wise independent version circumvents this by using pseudorandom number generators that only guarantee independence up to order $k=O(1)$, which can be implemented with polynomial resources~\cite{joffe1974set,karloff1994construction,zhandry2015secure}.

\subsection{Simulating time evolution}\label{app:time-evol}
\begin{fact}\label{fact:phase-f}
Given query access to a function $f: [2^n] \to [2^m]$, one can implement the unitary
\begin{equation}
    T_f: \ket{x} \mapsto e^{2i \pi f(x)/2^m} \ket{x}
\end{equation}
with a single query to $f$, $m$ auxiliary systems, and $O(m^2)$ gates. The idea is simple: initialize the auxiliary systems in $\ket{1}^{\otimes m}$, and apply the quantum Fourier transform. This results in the mapping
\begin{equation}
    \ket{x} \otimes \ket{1}^{\otimes m} \mapsto \ket{x} \otimes \sum_{z \in [2^m]} \omega_M^{z} \ket{z},
\end{equation}
where $M = 2^m$ and $\omega_M \coloneqq e^{2i \pi/M}$. We then query the function $f: [2^n] \to [2^m]$ on our $n$ system qubits, and subtract the result from the auxiliary system, so the overall unitary now implements
\begin{equation}
    \ket{x} \otimes \ket{1}^{\otimes m} \mapsto \ket{x} \otimes \sum_{z \in [2^m]} \omega_M^{z} \ket{z-f(x)}.
\end{equation}
Observe that the state vector $\sum_{z \in [2^m]} \omega_M^{z} \ket{z-f(x)}$ can be rewritten $\omega_M^{f(x)} \sum_{z \in [2^m]} \omega_M^z \ket{z}$. This is known as the phase kick-back trick. After this step, we can simply apply an inverse quantum Fourier transform on the $m$ auxiliary systems, and the unitary is now
\begin{equation}
    \ket{x} \otimes \ket{1}^{\otimes m} \to e^{2i \pi f(x)/2^m} \ket{x} \otimes \ket{1}^{\otimes m},
\end{equation}
as desired. 
\end{fact}

We continue to present compelling properties
of pseudorandom functions \cite{zhandry_how_2021},
objects that are crucially important for our work and that have been considered many times before in the quantum information literature, for example to show separations in notions of quantum machine learning
\cite{PACLearning}.

\begin{fact}\label{fact:prf}
Assume we have a pseudorandom function $\prf: \mathcal{K} \times [2^n] \to [2^m]$, where $[2^n]=\qty{0,1,\ldots,2^{n-1}}$ and $\mathcal{K}$ is the key space~\cite{zhandry_how_2021}. Furthermore, assume that $\prf_k$ (with $k$ sampled uniformly at random from $\mathcal{K}$) is \emph{computationally} indistinguishable from some other ensemble of functions $F = \qty{f: [2^n] \to [2^m]}$, in the sense that
\begin{equation}\label{eq:indist}
    \abs{\Pr_{k \sim \mathrm{Uniform}(\mathcal{K})}[\mathcal{A}^{\prf_k}() = 1] - \Pr_{f \sim F}[\mathcal{A}^{f}() = 1]} = \negl(n)
\end{equation}
for any efficient quantum algorithm $\mathcal{A}^f$ with query access to the function $f$. For the remainder of this work, we will assume that $F$ is a uniform distribution over all possible functions mapping $[2^n] \to [2^m]$. 

By \cref{fact:phase-f}, there is an \emph{efficient} circuit which implements
\begin{equation}
    T_k: \ket{x} \mapsto e^{2i \pi \prf_k(x)/2^n},
\end{equation}
since $\prf_k$ is efficiently implementable as well. Therefore, $\qty{T_k \mid k \sim \text{Uniform}(\mathcal{K})}$ and $\qty{T_f \mid f \sim F}$ are computationally indistinguishable in the sense of \cref{eq:indist}, since if they were distinguishable, this would provide an efficient algorithm to distinguish between $\qty{\prf_k \mid k \sim \text{Uniform}(\mathcal{K})}$ and $F$. Finally, we can always assume the family $\qty{\prf_k \mid k \sim \text{Uniform}(\mathcal{K})}$ is $k$-wise independent for any $k=O(1)$ using well-known constructions of $k$-wise independent functions~\cite{joffe1974set,karloff1994construction,zhandry2015secure}.
\end{fact}

\begin{fact}\label{fact:efficient-cdf}
The computational complexity of evaluating the inverse cumulative distribution function $F_{p^{(1)}}^{-1}$ is $O(\poly n)$. The reason is the following. The cumulative distribution function of $p^{(1)}(\lambda)= \frac{4-\lambda^2}{2\pi}$ is
\begin{equation}
    F_{p^{(1)}}(x) = \int_{-2}^x \frac{\sqrt{4-\lambda^2}}{2\pi} \dd{\lambda} = \frac{x\sqrt{4-x^2}+4\arcsin(x/2)}{4\pi} + \frac{1}{2}.
\end{equation}
This is an elementary function, which implies 
that it has $O(\poly n)$ computational complexity~\cite{borwein1987pi}. Furthermore, it is known that the complexity of the inverse for any elementary function is \emph{equivalent} to the complexity of the function itself~\cite{borwein1987pi}.

\end{fact}

\begin{algorithm}[H]
\caption{Implementing $\text{C}^{q_c} U(t)$\label{alg:ctime-evol}}
\KwData{$\pru_k$, $\prf_k: [2^n] \to [2^m]$, and $\tilde{d}$; these specify the eigenbasis, spectrum, and ``effective dimension'' of the pseudo-GUE Hamiltonian. We assume $m \geq n^2$.}
\KwIn{An time $\abs{t} = \exp(O(n))$ and a control qubit $q_c$.}
Apply $\pru_k^\dagger$. \Comment{$\pru_k$ and $\pru_k^\dagger$ do not need to be controlled; if $e^{-i\Lambda t}$ is not applied, the net effect is $(\pru_k)(\pru_k^\dagger) = \id$.}

\Begin(\Comment*[h]{Implement $\mathrm{C}^{q_c}e^{-i \Lambda t}$}){ 
Initialize $m=n^2$ auxiliary qubits in $\ket{1}$.

Apply $\ketbra{0}_{q_c} \otimes (H X)^{\otimes m} + \ketbra{1}_{q_c} \otimes (\text{QFT})$, where $\ketbra{0}_{q_c}$ and $\ketbra{1}_{q_c}$ act on the control qubit $q_c$, and the target is the $m$ auxiliary qubits. This prepares $\ket{+}^{\otimes m}$ if the control qubit is $\ket{0}$, and otherwise prepares $\sum_{z \in [2^{m}]} \omega_{M}^z \ket{z}$, where $\omega_{M} \coloneqq e^{2\pi i/M}$ and $M \coloneqq 2^{m}$. Define the function
\begin{equation}\label{eq:phix2}
    \phi(x) \coloneqq \frac{t}{2\pi} (M \cdot F_{p^{(1)}}^{-1})(\prf_k(x[:\tilde{n}])) \mod{M},
\end{equation}
where $F_{p^{(1)}}^{-1}$ is the inverse CDF of $p^{(1)}$, $\tilde{n} \coloneqq \log_2 \tilde{d}$, and $x[:\tilde{n}]$ denotes the first $\tilde{n} \coloneqq \log_2 \tilde{d}$ bits of $x$. 

Call $\phi(x)$ on the system qubits and subtract $\phi(x)$ from the auxiliary systems. When the control qubit in state $\ket{0}$, this implements the unitary
\begin{equation}
    \ket{x} \mapsto \ket{x} \otimes \sum_{z \in [2^m]} \ket{z-\phi(x)} = \ket{x} \otimes \ket{+}^{\otimes m}.
\end{equation}
Otherwise, by \cref{fact:phase-f}, it implements 
\begin{equation}
    \ket{x} \mapsto \omega_M^{\phi(x)} \ket{x} \otimes \sum_{z \in [2^m]} \omega_M^z \ket{z}.
\end{equation}
Apply $\ketbra{0}_C \otimes (X H)^{\otimes m} + \ketbra{1}_C \otimes (\text{QFT}^{-1})$ to uncompute the auxiliary system. 
The end result is the unitary
\begin{equation}
    \ket{x} \mapsto \begin{cases}
        \ket{x} &\qq{when the control qubit is in state vector $\ket{0}$, and} \\
        \omega_M^{\phi(x)} \ket{x} &\qq{otherwise.}
    \end{cases}
\end{equation}
}
Apply $\pru_k$.
\end{algorithm}
\smallskip

We elaborate further on the details of 
the function
$x\mapsto \phi(x) \,\eqref{eq:phix2}$.
First, we note that each component of $\phi(x)$ is efficient: $\prf_k$ is efficient by definition of quantum pseudorandom 
functions~\cite{zhandry_how_2021}, $F_{p^{(1)}}^{-1}$ is efficient by \cref{fact:efficient-cdf}, and the modular multiplication is efficient using well-known circuits~\cite{cho2020quantum,rines2018high}. The role of \cref{eq:phix2} is simple: the first step $F_{p^{(1)}}^{-1} \circ \prf_k$ samples a random energy $\lambda_x \sim p^{(1)}$, a standard technique for non-uniform random variable sampling~\cite{hormann2004automatic,steele1987uniform}. The second step, multiplication by $\frac{M \cdot t}{2\pi}$, ensures the right prefactor so that ultimately the phase $e^{i \lambda_x t}$ is applied, and the modulo by $M$ simply drops terms that contribute phases which are multiples of $2\pi$. Since $F_{p^{(1)}}^{-1} \circ \prf_k$ in reality samples from a ``discretized'' version of $p^{(1)}$ in intervals of size $M^{-1}$, we need to set $M = 2^{n^2}$ so that the effect of this discretization is still exponentially 
suppressed for exponentially long times. Indeed, the effects of the discretization can be suppressed for any time $t=2^{O(n^c)}$ simply by using $m=n^{c+1}$ auxiliary systems. 
Since the the computational complexity of $\prf_k$, $F_{p^{(1)}}^{-1}$, and modular multiplication by $\frac{M \cdot t}{2\pi}$ each scale polynomially in $m$, we see that the overall simulation cost (in $t$) is just $O(\poly \log t)$. Therefore, this construction only becomes 
inefficient if one aims to simulate for extremely long times which scale as $\exp(\omega(\poly n))$. Finally, we note that the implementation of $e^{-i \Lambda t}$ in \cref{alg:ctime-evol} is not permutation-invariant, as the first $\tilde{d}$ bit-strings will all have the same phase, as well as the next $\tilde{d}$, and so on. However, the application of a pseudorandom unitary $\pru_k^\dagger$ before and $\pru_k$ after remedies this problem, as this shuffles the phases around pseudorandomly (to see this, note that the 
phases can be manually shuffled using a pseudorandom permutation 
$\prp_k^\dagger$ before $e^{-i \Lambda t}$, and $\prp_k$ after -- but this 
can 
simply be absorbed into $\pru_k$).

\subsection{Simulating Gibbs states}\label{app:gibbs-sim}

Rejection sampling is a technique used to generate samples from a 
target distribution that may be difficult to sample from directly. In our 
case, the target distribution is the Gibbs state of our pseudo-GUE Hamiltonian. The method works as follows. We propose samples from a simpler distribution (the surrogate distribution), and then we accept or reject these samples based on a comparison with the target distribution.

For our pseudo-GUE Hamiltonian, $Q_i = d^{-1}$ represents our surrogate distribution, which is a uniform distribution over all states, while $G(\Lambda)_i = e^{-\beta \lambda_i}/\tr(e^{-\beta \Lambda})$ represents the actual Gibbs distribution we want to sample from. The key to efficient rejection sampling is choosing a constant $C$ such that $C \cdot Q_i$ always upper bounds $G(\Lambda)_i$. This ensures we never mistakenly reject a sample we should have accepted. The smaller we can make $C$, the more efficient our sampling becomes. The following lemma shows that we can, indeed, choose a relatively small $C$.

\begin{lemma}[Choice of small $C$]\label{lem:c-upper}
For all $\beta = O(\poly n)$, the distribution $(G(\Lambda))_i \coloneqq \frac{e^{-\beta \lambda_i}}{\tr(e^{-\beta \Lambda})}$ obeys
\begin{equation}
    \Pr_{\Lambda \sim p}\qty[\max_i (G(\Lambda))_i \geq \frac{1}{d} \max\qty(\frac{(2\beta)^{3/2}}{2},12)] \leq O(d^{-1/8}).
\end{equation}
\end{lemma}
\begin{proof}
    From \cref{lem:z-bound}, we apply a 
    Chebyshev inequality, to get
    \begin{equation}
        \Pr[Z \leq \frac{d I_1(2\beta)}{\beta} (1 - d^{-1/8})] \leq O(d^{-1/8}).
    \end{equation}
    Since $\max_i P(i) \leq \frac{e^{2\beta}}{Z}$, and $I_1(2\beta) \geq \frac{e^{2\beta}}{2 \sqrt{2\beta}}$ for any $\beta \geq 2$, we get that
    \begin{equation}
        \Pr[\max_i P(i) \geq \frac{(2\beta)^{3/2}}{2d}] \leq O(d^{-1/8})\qq{for $\beta \geq 2$.}
    \end{equation}
    For $\beta \leq 2$, we directly upper bound $\frac{e^{2\beta}}{d I_1(2\beta)/\beta (1-d^{-1/8})} \leq \frac{12}{d}$.
\end{proof}
This lemma essentially tells us that for the vast majority of our pseudo-GUE Hamiltonians, we can choose 
\begin{equation}
C = \max\qty(\frac{(2\beta)^{3/2}}{2}, 12). 
\end{equation}
This choice of $C$ leads to an efficient algorithm, as demonstrated in the following theorem.

\begin{theorem}[Efficiency and correctness of \cref{alg:gibbs2}]
For any $k=O(\poly n)$ and any $\beta=O(\poly n)$, \cref{alg:gibbs2} terminates in polynomial time with overwhelming probability. Furthermore, with overwhelming probability over the choice of pseudo-GUE Hamiltonian, the algorithm prepares the exact Gibbs state.
\end{theorem}
\begin{proof}
The central idea is that the surrogate distribution $Q_i = d^{-1}$ will be a good proposal distribution in the sense that the rejection sampling criterion $(G(\Lambda))_i \leq C \cdot Q_i, \forall i$ will be satisfied with just $C = \max\qty(\frac{(2\beta)^{3/2}}{2},12)$. Fixing $C$ to be this constant will result in our algorithm failing to be correct for a vanishing fraction $O(d^{-1/8})$ of pseudo-GUE Hamiltonians, as shown in \cref{lem:c-upper}. However, those pathological Hamiltonians aside, the correctness of this algorithm follows directly from the correctness of rejection sampling in general.

Each iteration of the sampling will require a number of samples $N$ from the surrogate distribution (before one gets accepted) that follows a geometric distribution with parameter $1/C$~\cite{casella2004generalized,radford2003slice}. Let $N_j$ be the number of samples required in the $j$th iteration. We use the fact that $\E[e^{t N_j}]=\frac{e^t}{C-(C-1)e^t}$. Since each of the $N_j$ are independent of each other, we can use a Chernoff bound to arrive at
\begin{equation}
    \Pr[\sum_{j=1}^k N_j \geq kC^2] \leq \min_{t} \qty(\frac{e^{(1-C^2)t}}{C-(C-1)e^t})^k.
\end{equation}
This has a minimum at $t = \ln(1+C^{-1})$, which we substitute to find
\begin{equation}\label{eq:chernoff}
    \Pr[\sum_{j=1}^k N_j \geq kC^2] \leq \qty(\frac{(1+C^{-1})^{1-C^2}}{C-\frac{(C-1)(C+1)}{C}})^k = \qty(C(1+C^{-1})^{1-C^2})^k,
\end{equation}
and since $C > 12$, $C(1+C^{-1})^{1-C^2} < 10^{-3}$. We can assume WLOG that $k \geq n$, so \cref{eq:chernoff} is upper bounded by $10^{-3k}$ -- if $k<n$, we simply observe that 
\begin{equation}
\Pr[\sum_{j=1}^k N_j \geq nC^2] \leq \Pr[\sum_{j=1}^n N_j \geq nC^2] \leq 10^{-3n} 
\end{equation}
--- so the runtime will be at most $O(\max(k,n) C^2)$ with overwhelming probability in $n$. Therefore, the algorithm runs in time $O(\poly(n,\beta))$ for any $k = O(\poly n)$ with overwhelming probability.
\end{proof}

\begin{algorithm}[H]
\caption{Efficient Gibbs state preparation for the pseudo-GUE\label{alg:gibbs2}}
\KwData{$\pru_k$, $\prf_k: [2^n] \to [2^m]$, and $\tilde{d}$; these specify the eigenbasis, spectrum, and ``effective dimension'' of the pseudo-GUE Hamiltonian. We assume $m \geq n^2$.}
\KwIn{An inverse temperature $\beta = O(\poly n)$.}
$C \gets \max(\frac{(2\beta)^{3/2}}{2}, 12)$

\While{True}{
    $x \sim \text{Uniform}(\qty{0,1,\ldots,2^{n}-1})$ \Comment{Propose a random $x$}
    $\lambda_x \gets F_{p^{(1)}}^{-1}(\prf_k(x[:\tilde{n}])/2^m)$ \Comment{Get the energy $\lambda_x$ associated with $x$}
    
    $\alpha \sim \text{Uniform}([0,e^{2\beta}])$
    
    \If{$e^{-\beta \lambda_x}/C \geq \alpha$}{ 
        \Return $U_k \ketbra{x} U_k^\dagger$ \Comment{Accept $x$}
    }
}
\end{algorithm}

\subsection{Remarks on the Monte-Carlo sign problem}\label{App:SP}
In this section, we prove that GUE Hamiltonians suffer from a strong sign problem, a notorious obstacle in quantum Monte Carlo techniques to compute expectation values of their Gibbs states. While this gives evidence that GUE Gibbs states are hard to classically simulate, it does not rule out heuristic techniques to find basis rotations that ease the sign problem.

We consider the ``average sign" quantity $\nu_1(R)$ which was introduced in Ref.~\cite{Easing}. There it has been shown that to estimate expectation values of Gibbs states to precision $\epsilon$, the magnitude of the ``average sign" determines the variance of the estimator. That is, one must average over a number of samples scaling as
\begin{equation}\label{eq:samples}
s \geq \frac{1}{\E[\nu_1(R)]^2 \epsilon^2}.
\end{equation}

The Gaussian unitary ensemble $\mathcal{E}_{\gue}$
of $d\times d$ Hamiltonians is an ensemble of Hamiltonians for which the $d$ diagonal elements are iid distributed according to
$\mathcal{N}(0,d^{-1/2})$. Both the real and the imaginary parts of the off-diagonal elements will, in turn, be distributed as 
$\mathcal{N}(0,(2d)^{-1/2})$. Let us consider the real part $R=\Re(H)$ only and denote with $R_\urcorner$ the part of $R$ that consists of the strictly positive off diagonal elements of $R$. 

\begin{lemma}[Sign problem] For the real part $R$ of GUE Hamiltonian ensembles, we have 
\begin{equation}
\E[\nu_1(R)] = \mathbb{E}
\qty[d^{-1} \norm{R_\urcorner}_1]
= \frac{\sqrt{d}}{4\sqrt{\pi}} - o(1).
\end{equation}
\end{lemma}

\begin{proof} This is easy to see. Since the ${1}$ norm is the sum of absolute values, 
the  $\norm{R_\urcorner}_{1}$
  are given, up to normalization, by a sum of iid half-normal distributed entries. The expectation of each entry is $\frac{1}{4\sqrt{\pi d}}$, and since there are $d^2-d$ such entries, we have
  \begin{equation}
      \E[\nu_1(R)] = \frac{1}{4\sqrt{\pi}}\frac{d-1}{\sqrt{d}} = \frac{\sqrt{d}}{4\sqrt{\pi}} - o(1).
  \end{equation}
  \end{proof}

\section{Behavior of probes to chaos of GUE versus pseudo-GUE}\label{App:behaviorprobestochaos}

\subsection{Preliminaries}
Let us recall the notation used throughout the manuscript. The eigenvalue distribution of the GUE is denoted as $p(\lambda_1,\ldots,\lambda_{d})$ (defined in \cref{eq:eigenvaluedistirbution}), and abbreviated simply as $p$. We denote as $\tilde{p}(\lambda_1,\ldots,\lambda_d)\coloneqq\prod_{i=1}^{d}p^{(1)}(\lambda_i)$ the independent and identically distributed (iid) distribution built from first marginals of $p$. We also define degenerate version of $\tilde{p}$, which we denote with $\tilde{p}_{\tilde{d}}$. This is a permutation-invariant distribution over $d$ energies which has only $\tilde{d}$ unique energies, each of which is iid according to $p^{(1)}(\lambda)$. Lastly, we denote as $\tilde{p}_{\tilde{d},k}$ the version of $\tilde{p}_{\tilde{d}}$ where the $\tilde{d}$ variables are merely $k$-wise independent (where $k$ is assumed to be $O(1)$). Correspondingly, we use three main ensembles of Hamiltonians: the $\gue$ ensemble reviewed in \cref{app:GUE}, and the two ensembles
\be
\mathcal{E}_{\tilde{d}}\coloneqq\{U\Lambda U^{\dag}\,:\, U\sim \mathcal{U}\,,\, \Lambda\sim \tilde{p}_{\tilde{d}}\}, \quad {\mathcal{E}}_{\tilde{d},k}\coloneqq\{U\Lambda U^{\dag}\,:\, U\sim \mathcal{U}\,,\, \Lambda\sim \tilde{p}_{\tilde{d},k}\},
\ee
where $\mathcal{U}$ is an ensemble of pseudorandom unitaries. Notice that only $\mathcal{E}_{\tilde{d},k}$ is efficiently implementable on a quantum device. Indeed, the ensemble $\mathcal{E}_{\tilde{d}}$ has exponential gate complexity due to the iid  distribution, see Ref.~\cite{kotowski2023extremal}.

\begin{figure}[ht]
    \centering
    \def\svgwidth{0.75\textwidth}
 %   \graphicspath{{}}
%    \input{haar-avg.pdf_tex}
 \includegraphics[width=0.5\columnwidth]{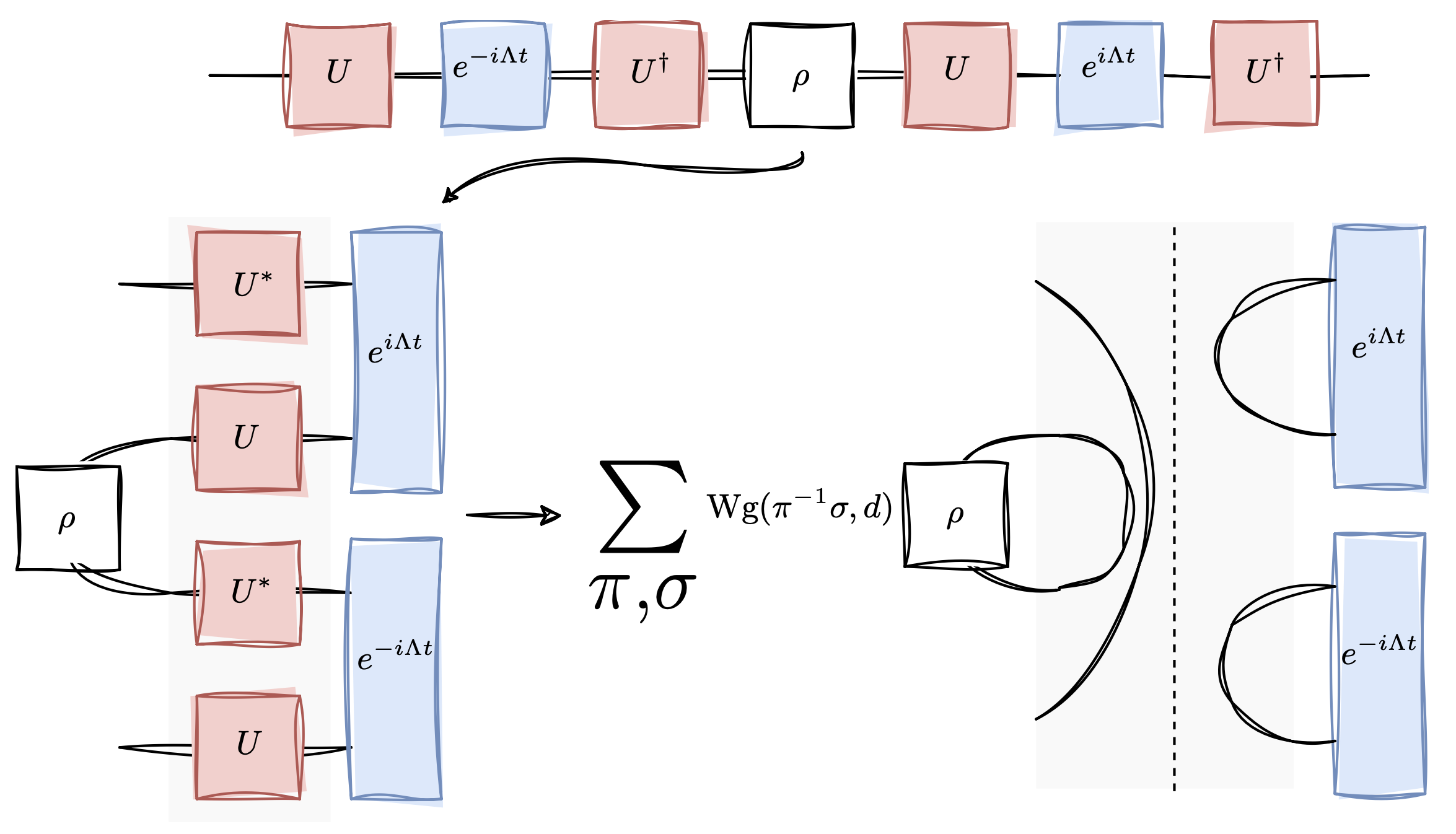}
    \caption{Calculating the Haar average over the eigenbasis for time evolution by a GUE Hamiltonian. The first step simply rewires the diagram into a format such that the ``vectorized form'' of the Haar average can be calculated~\cite{mele2024introduction}. In the last diagram, the sum runs over all permutations $\pi,\sigma \in \mathbb{S}_2$, and the wiring on the left half of the gray box corresponds to $\pi$, while the right side corresponds to $\sigma$. $\text{Wg}(\pi^{-1} \sigma,d)$ is the Weingarten function associated with the permutation $\pi^{-1} \sigma$. Note that the right side of the diagram simply results in coefficients which reweight the terms of the sum according to different spectral form factors; the one show in the figure corresponds to $\abs{\tr(e^{-i \Lambda t})}^2$.}\label{fig:vectorize}
\end{figure}

In this section, we reduce the typical entanglement, magic, and OTOCs of time-evolved states (or operators) to functions of the spectral form factor. For brevity, we will define 
\begin{equation}
Z(X) := \frac{\tr(e^{-i X})}{d} 
\end{equation}
for any operator $X$. We have the following preliminary lemma which allows us to reduce calculations of entanglement, magic, or OTOCs at time $t$ to the problem of reasoning about $Z(\Lambda t)$, where $\Lambda$ is the spectrum of a GUE Hamiltonian. Specifically, in the notation of \cref{lem:haar-avg} below, we will find that entanglement, magic, or OTOCs of time-evolved states can be expressed as $\E_{U \sim \haar}[(f^{(k)}(U e^{-i \Lambda t} U^\dagger))]$ for some function $f^{(k)}$, which then reduces to $\abs{Z(\Lambda t)}^{2k} + O(d^{-1})$.

\begin{lemma}[Haar averages]\label{lem:haar-avg}
Consider any degree-$k$ function $f^{(k)}$ mapping operators to $[0,1]$. Assume $f$ has Lipschitz constant $\eta=O(1)$. Consider any ensemble $\mathcal{E}$ of Hamiltonians for which the associated probability distribution factors into independent distributions over the eigenbases and eigenenergies. Let the distribution over the eigenbases form a relative $\epsilon$-approximate $4k$-design $\mathcal{E}_{4k,\epsilon}$ with $\epsilon=O(d^{-1})$. Then for any fixed diagonal matrix $D$,
\begin{equation}
    \Pr_{U \sim \mathcal{E}_{4k,\epsilon}}\qty[\abs{f^{(k)}(U D U^\dagger)-\E_{U \sim \haar}[(f^{(k)}(U D U^\dagger))]} \geq  d^{-1/12}] \leq O(d^{-1/6}).
\end{equation}
\end{lemma}
\begin{proof}
For a fixed diagonal $D$, we can define a $2k$-degree function $g_D(U) \coloneqq f^{(k)}(U D U^\dagger)$.
Note that by Levy's lemma, $\Pr[g_D(U) \geq \E_U[g_D(U)] + d^{-1/3}] \leq \exp(-\frac{d^{1/3}}{200\eta^2})$.
\begin{equation}
\begin{aligned}
    \E_{U \sim \mathcal{E}_{4k,\epsilon}}[(g_D(U))^2] &\leq (1+\epsilon) \E_{U \sim \haar}[(g_D(U))^2] \\
    &\leq (1+\epsilon) \qty(\qty(\E_{U \sim \haar}[g_D(U)] + d^{-1/3})^2 + \exp(-\frac{d^{1/3}}{200\eta^2})) \\
    &\leq \qty(\E_{U \sim \haar}[g_D(U)])^2 + O(d^{-1/3}).
\end{aligned}
\end{equation}
By a simple Chebyshev inequality, for any fixed $D$,
\begin{equation}
    \Pr_{U \sim \mathcal{E}_{4k,\epsilon}}\qty[\abs{g_D(U)- \E_{U \sim \haar}[(g_D(U))]} \geq d^{-1/12}] \leq O(d^{-1/6}).
\end{equation}
\end{proof}

Now, we turn to calculating $\E_{U \sim \haar}[(f^{(k)}(U e^{-i \Lambda t} U^\dagger))]$ when $f^{(k)}$ represents entanglement, magic, local operator entanglement, and OTOCs of time-evolved states or operators. These Haar averages are shown in \cref{tab:expval-haar}. They are derived by brute forcing the Weingarten calculus calculation illustrated in  \cref{fig:vectorize} using \texttt{SymPy}, computer algebra system~\cite{sympy}. This is the only feasible rigorous approach, as the full expressions include up to $(8!)^2 \sim 10^9$ terms. The basic technique for the calculation of the first two quantities (purity and magic) is outlined below, and no new techniques were used for the calculation of the remaining two quantities.
\begin{itemize}
    \item For purity, we rewrite the subsystem swap operator as a sum $d_A^{-1} \sum_{P \in \mathbb{P}_{n_A}} P^{\otimes 2}$, after which the 
    problem reduces to evaluating $\E_{U \sim \haar}[\tr^2(P \rho(U,t))]$. Paulis in $\mathbb{P}^{(Z)} \coloneqq \expval{Z_1,Z_2,\ldots,Z_n}$ have different behavior than Paulis outside this subgroup, 
    because these are Paulis for which the initial state has non-zero expectation. We find
    \begin{equation}
        \E_{U}[\tr^2(P \rho(U,t))] = \begin{cases}
            1 &\qq{if $P = \id$,} \\
            \abs{Z(\Lambda t)}^4 + \frac{2\Re(Z(\Lambda t)^2 Z(-2\Lambda t)) - 3 \abs{Z(\Lambda t)}^4 + 1}{d} + O(d^{-2}) &\qq{if $P \in \mathbb{P}^{(Z)}$,} \\
            \frac{1-\abs{Z(\Lambda t)}^4}{d} + O(d^{-2}) &\qq{otherwise.}    
        \end{cases}
    \end{equation}
    \item For magic, the technique was similar. 
    We find
    \begin{equation}
        \E_{U}[\tr^4(P \rho(U,t))] = \begin{cases}
            1 &\qq{if $P= \id$,} \\
            \abs{Z(\Lambda t)}^8 + \frac{12\abs{Z(\Lambda t)}^4\Re(Z(\lambda t)^2 Z^*(2\Lambda t))) - 18 \abs{Z(\Lambda t)}^8 + 6\abs{Z(\Lambda t)}^4}{d} + O(d^{-2}) &\qq{if $P \in \mathbb{P}^{(Z)}$,} \\
            \frac{3 \qty(\abs{Z(\Lambda t)}^2 -1)^2}{d^{2}} + O(d^{-3}) &\qq{otherwise.}
        \end{cases}
    \end{equation}
\end{itemize}

\begin{table*}[ht]
    \centering
    \begin{tabular}{|c||c|c|}
    \hline
         \textbf{Quantity}&\textbf{Leading term}&\textbf{Subleading term} \\
    \hline \hline
    Subsystem purity $\tr((\rho(U,t)_A)^2)$; & \multirow{2}{*}{$\abs{Z(\Lambda t)}^4$} & \multirow{2}{*}{$(d_B^{-1}+d_A^{-1})(1-\abs{Z(\Lambda t)}^4)$} \\ 
    $\rho(U,t) \coloneqq U e^{-i \Lambda t}U^\dagger \ketbra{0} U e^{i \Lambda t} U^\dagger$ & & \\
    \hline
    Stabilizer purity $M_{\text{pur}}(\rho(U,t))$; & \multirow{2}{*}{$\abs{Z(\Lambda t)}^8$} & \multirow{2}{*}{$d^{-1}\qty(12\abs{Z(\Lambda t)}^4\Re(Z(\lambda t)^2 Z^*(2\Lambda t))) - 16 \abs{Z(\Lambda t)}^8 + 4)$} \\ 
    $M_{\text{pur}}(\rho) \coloneqq d^{-1}\sum_{P \in \mathbb{P}_n} \tr^4(P \rho)$ & & \\ \hline
    Operator purity $\tr((V(U,t)_{A\cup A'})^2)$; & \multirow{2}{*}{$\abs{Z(\Lambda t)}^8$} & $d^{-1}\Bigl(2\Re(Z(\Lambda t)^4 Z^*(2\Lambda t)^2) + 2 \abs{Z(\Lambda t)}^4 \abs{Z(2\Lambda t)}^2$ \\ 
    $\ket{V(U,t)} \coloneqq (U e^{-i \Lambda t}U^\dagger) V (U e^{i \Lambda t} U^\dagger) \otimes \id \ket{\phi^+}$ &  & $-\,2\abs{Z(\Lambda t)}^8 - 4\abs{Z(\lambda t)}^4 \Re(Z(\Lambda t)^2 Z^*(2\Lambda t)) + 2\Bigr)$\\ \hline
    4-point OTOC; & \multirow{2}{*}{$\abs{Z(\Lambda t)}^4$} & \multirow{2}{*}{$O(d^{-2})$} \\ 
    $\OTOC_4(U,t) = \tr(P_1(U,t) P_2 P_1(U,t) P_2)/d$ & & \\ \hline
    \end{tabular}
    \caption{Average values of four quantities for 
    Haar average eigenbases. For the operator purity, we have assumed $\dim \mathcal{H}_A = \dim \mathcal{H}_B$ for convenience, which is to say that the bipartition we evaluate the purity with respect to exactly divides the system in half. For the 4-point OTOC, we assume $P_1$ and $P_2$ are commuting non-identity Pauli operators.}
    \label{tab:expval-haar}
\end{table*}

Having established the centrality of $Z(\Lambda t)$ and its moments, we will now separate our analysis into two separate cases. First, we will aim to understand the late-time behavior of the $Z(\Lambda t)$ when $\Lambda$ is distributed according to either $\tilde{p}_{\tilde{d}}$ or $\tilde{p}_{\tilde{d},k}$. We will then turn to focus on the case of early time behavior. 

\subsection{Late-time Behavior}
Our approach is as follows. We first calculate the expectation of $\abs{Z(\Lambda t)}^{2k}$ at late times under a $2k$-wise independent distribution. By combining this with various concentration and anticoncentration (e.g., the Paley–Zygmund inequality) inequalities, we can find good bounds on the distribution of $\abs{Z(\Lambda t)}^{2k}$, which in turn allows us to make inferences about the quantities in \cref{tab:expval-haar}.

\begin{lemma}[Late-time expectation of $Z$]
For any integer $k=O(1)$, and any $t \geq \Omega(d)$,
\begin{equation}
    \E_{\Lambda \sim \tilde{p}_{d,2k}}\qty[\abs{Z(\Lambda t)}^{2k}] = \frac{k!}{d^k} + O(d^{-(k+1)}).
\end{equation}
\end{lemma}
\begin{proof}
For simplicity, let us denote $\mathbb{E}:=\E_{\Lambda \sim \tilde{p}_{d,2k}} $ for brevity
\begin{equation}
    \E\qty[\abs{Z(\Lambda t)}^{2k}] = \frac{1}{d^{2k}} \sum_{(p_1,q_1),\ldots,(p_k,q_k)} \E\qty[\exp(-i \sum_{m=1}^k (U_{p_m}-U_{q_m}))],\label{eq:z2k}
\end{equation}
where we define the random variables $U_i \coloneqq t\lambda_i \pmod{2\pi}$, with each of the $\lambda_i$ iid according to $p^{(1)}$. Observe that if $t \geq \Omega(d)$, $U_i$ is \emph{uniformly} distributed over the interval $[0,2\pi]$ up to an error $d^{-1}$ in total variation distance. Therefore, we evaluate \cref{eq:z2k} in the case where $U_i \sim \text{Uniform}(0,2\pi)$; the actual value will differ by at most $O(k/d)=O(d^{-1})$. 

When $U_i \sim \text{Uniform}(0,2\pi)$, then if the argument of the exponential in \cref{eq:z2k} is not zero (i.e., all the $U_i$s cancel out), the expectation is zero. Therefore, the expectation of $\abs{\tr(e^{-i \Lambda t})}^{2k}$ reduces to counting the number of tuples $(p_1,q_1),\ldots,(p_k,q_k)$ such that the they all ``pair off'', so that $\sum_{m=1}^k U_{p_m} - U_{q_m} = 0$. The number of ways to do this is
\begin{equation}
    \sum_{\lambda \vdash k} \binom{d}{\abs{\lambda}} \binom{k}{\lambda_1,\lambda_2,\ldots}^2,
\end{equation}
where $\binom{k}{\lambda_1,\lambda_2,\ldots}$ is a multinomial coefficient and $\abs{\lambda}$ is the number of terms in the partition $\lambda$. Here, the sum runs over all possible repeated indices in $(p_1,\ldots,p_k)$ -- so for instance of $k=3$ and $\lambda=(2,1)$, it means there are two repeated indices and one unique one. The $\binom{d}{\abs{\lambda}}$ term counts the number of ways to pick indices from $1,\ldots,d$ to assign to $p_1,\ldots,p_k$, and the multinomial coefficient counts for the number of ways to shuffle the $p$ and $q$ indices. The leading order term in this corresponds to $\lambda = (1,\ldots,1)$ and yields a contribution $d^k k! + O(d^{k-1})$, and remaining terms are also $O(d^{k-1})$. We can, therefore, conclude
that
\begin{equation}    \E\qty[\frac{\abs{\tr(e^{-i \Lambda t})}^{2k}}{d^{2k}}] = \frac{k!}{d^k} + O(d^{-(k+1)}).
\end{equation}
\end{proof}

\begin{remark}[Generalizing
results of \citet{cotler2017chaos}]
    This is a generalization of a result by \citet{cotler2017chaos}, which has shown that at times $t=\Omega(d)$,
    \begin{equation}\label{eq:sff-12}
    \E\qty[\frac{\abs{\tr(e^{-i \Lambda t})}^{2}}{d}] = 1 + O(d^{-1}) \qq{and} \E\qty[\frac{\abs{\tr(e^{-i \Lambda t})}^{4}}{d^2}] = 2 + O(d^{-1}).
    \end{equation}
    It has, therefore, been conjectured that 
    \begin{equation}
        \E\qty[\frac{\abs{\tr(e^{-i \Lambda t})}^{2k}}{d^k}] = k + O(d^{-1}).
    \end{equation}
    Instead, we have shown that the asymptotic value is $k!$, which is consistent with the known results in \cref{eq:sff-12}.
\end{remark}

\begin{corollary}[Concentration of $\abs{Z(\Lambda t)}$ under $4$-wise independence]\label{cor:z-bound}
Let $\tilde{p}_{\tilde{d},4}$ the $4$-wise independent distribution. Then, for any $t \geq \Omega(\tilde{d})$, it holds that
\begin{equation}
    \Pr_{\Lambda \sim \tilde{p}_{\tilde{d},4}}\qty[\abs{Z(\Lambda t)}^{2k} \geq \tilde{d}^{-2k}] \geq \frac{1}{2} - O(\tilde{d}^{-1}).
\end{equation}
\end{corollary}
\begin{proof}
We make use of the Paley-Zygmud inequality, which reads
\begin{equation}
\Pr\qty[\abs{Z(\lambda t)}^2 \geq \tilde{d}^{-2}] \geq \qty(1 - \frac{1}{\tilde{d}^{2} \E[\abs{Z(\lambda t)}^2]}) \frac{\E[\abs{Z(\lambda t)}^2]^2}{\E[\abs{Z(\lambda t)}^4]} \geq \qty(1 - \frac{1}{\tilde{d}}) \frac{1}{2}-O(\tilde{d}^{-1}) = \frac{1}{2} - O(\tilde{d}^{-1}).
\end{equation}
\end{proof}

\begin{lemma}[Concentration of $\abs{Z(\Lambda t)}$ under complete independence]\label{lem:pd-indep}
For any $t \geq \Omega(\tilde{d})$,
\begin{equation}
    \Pr_{\Lambda \sim \tilde{p}_{\tilde{d}}}\qty[\tilde{d}^{-2k} \leq \abs{Z(\Lambda t)}^{2k} \leq \tilde{d}^{-k/2}] \geq 1-O(\tilde{d}^{-1/2}).
\end{equation}
\end{lemma}
\begin{proof}
When $\tilde{p}_{\tilde{d}}$ is iid, observe that $Z(\Lambda t) = \frac{1}{\tilde{d}} \sum_{j=1}^{\tilde{d}} e^{-i \lambda_j t}$ becomes the average of $\tilde{d}$ iid random variables $X_j \coloneqq e^{-i \lambda_j t}$, where $\lambda_j \sim p^{(1)}$. We will treat $X_j$ as a vector-valued random variable, with components equal to the real and imaginary parts of $e^{-i \lambda_j t}$. We can analytically calculate
\begin{equation}
    \mu \coloneqq \E[X_j] = \mqty[\frac{J_1(2t)}{t} \\ 0],
\end{equation}
where $J_1$ is a Bessel function of the first kind. Also, 
\begin{equation}
    \Sigma \coloneqq \text{cov}(X_j) = \mqty[\frac{1}{2}\qty(1+\frac{J_1(4t)}{2t}) - \frac{J_1(2t)^2}{t^2} & 0 \\ 0 & \frac{1}{2}\qty(1-\frac{J_1(4t)}{2t})].
\end{equation}
Now, we have used the fact that the modified Bessel function of the first kind obeys~\cite{abramowitz1965handbook}
\begin{equation}
    \abs{J_1(x)} \leq \sqrt{\frac{2}{\pi \abs{x}}} + O(\abs{x}^{-3/2}).
\end{equation}
Therefore, for $t \geq \Omega(\tilde{d})$, $\Sigma = \frac{\id}{2} + O(\tilde{d}^{-3/2})$. We first show that the two normal distributions $D_1 = \mathcal{N}(\mu,\Sigma/\tilde{d})$ and $D_2 = \mathcal{N}(0,\id/(2\tilde{d})))$ are close in total variation distance. The KL divergence between two bivariate Gaussians is exactly $\frac{1}{2} \qty(\log \frac{\abs{\Sigma_2}}{\abs{\Sigma_1}} - 2 + \tr(\Sigma_2^{-1} \Sigma_1) + (\mu_2-\mu_1)^T \Sigma_2^{-1} (\mu_2-\mu_1))$. For our case, this means
\begin{equation}
    D_{KL}(D_2 \parallel D_1) = \frac{1}{2} \qty(\log \frac{\abs{\id/2}}{\abs{\Sigma}} - 2 + \tr(2\Sigma) + 2\tilde{d}\frac{J_1^2(2t)}{t^2}) \leq O(\tilde{d}^{-1}).
\end{equation}
This implies that the total variation distance between $D_1$ and $D_2$ is $O(\tilde{d}^{-1/2})$.

The multivariate Berry-Esseen theorem says that $Z(\Lambda t)$ converges to $D_1$ in the sense that for any convex region $U$, $\abs{\Pr[Z(\Lambda t) \in U] - \Pr[D_1 \in U]} \leq O(\tilde{d}^{-1/2})$~\cite{rai2019multivariate}. By a simple triangle inequality, we then also have that
\begin{equation}
    \abs{\Pr[Z(\Lambda t) \in U] - \Pr[D_2 \in U]} \leq O(\tilde{d}^{-1/2}).
\end{equation}
To make use of this result, let us set $U$ to be a disk of radius $\tilde{d}^{-1}$ around the origin and compute the probability
\begin{equation}
    \Pr[D_2 \in U] = \int_0^{\tilde{d}^{-1}} 2\pi r \cdot \frac{2\tilde{d}}{2\pi} \cdot \exp(-\frac{r^2}{2} \cdot 2 \tilde{d}) \dd{r} = 1 - e^{-(\tilde{d}^{-1})} \leq \tilde{d}^{-1}.
\end{equation}
Therefore, we conclude that $\Pr[\abs{Z(\Lambda t)} \leq \tilde{d}^{-1}] \leq O(\tilde{d}^{-1/2})$, so in general,
\begin{equation}
    \Pr\qty[\abs{Z(\Lambda t)}^{2k} \geq \tilde{d}^{-2k}] \geq 1-O(\tilde{d}^{-1/2}).
\end{equation}
Repeating the above calculation and setting $U$ to be a disk of radius $\tilde{d}^{-1/4}$ around the origin, we find
\begin{equation}
    \Pr\qty[\tilde{d}^{-2k} \leq \abs{Z(\Lambda t)}^{2k} \leq \tilde{d}^{-k/2}] \geq 1 - O(\tilde{d}^{-1/2})\,,
\end{equation}
which concludes the proof. 
\end{proof}
Finally, this positions us to prove the main theorem.
\begin{theorem}[Late-time behavior of GUE versus pseudo-GUE]\label{th:late-time}
Let $\rho(H,t) \coloneqq e^{-iHt} \ketbra{0} e^{iHt}$, $\ket{V(H,t)} = (e^{-iHt} V e^{iHt} \otimes \id) \ket{\phi^+}$, and $\OTOC_4(H,t)= \tr(P_1(t) P_2 P_1(t) P_2)$, with $P_1 = e^{-iHt} P_1 e^{iHt}$. Let $\tilde{n} \coloneqq \log_2 \tilde{d}$. We define $\loe(H,t) = S_2(V(H,t)_{A\cup A'})$. For the pseudo-GUE Hamiltonians, assume the eigenbases form a relative $\epsilon$-approximate $16$-design with $\epsilon=O(d^{-1})$, and $\tilde{p}_{\tilde{d},8}$ is $8$-wise independent. Let us then denote $\mathcal{E}_{\tilde{d},8}$ the corresponding pseudo-GUE ensemble. For any $t_1 \geq \Omega(d)$ and any $t_2 \geq \Omega(\tilde{d})$, the following holds.
\small\begin{eqnarray}
        \Pr_{H \sim \gue}\qty[S_2(\rho(H,t_1)_A) \geq \frac{n}{16}] \geq 1-\exp(-O(n)) &&\qq{while} \Pr_{H \sim \pgue}\qty[\tilde{n} \leq S_2(\rho(H,t_2)_A) \leq 5\tilde{n}] \geq \frac{1}{2} - O(\tilde{d}^{-1}), \\
        \Pr_{H \sim \gue}\qty[M_2(\rho(H,t_1)) \geq \frac{n}{32}] \geq 1-\exp(-O(n)) &&\qq{while} \Pr_{H \sim \pgue}\qty[\tilde{n} \leq M_2(\rho(H,t_2)_A) \leq 9\tilde{n}] \geq \frac{1}{2} - O(\tilde{d}^{-1}), \\
        \Pr_{H \sim \gue}\qty[\loe(H,t_1) \geq \frac{n}{16}] \geq 1-\exp(-O(n)) &&\qq{while} \Pr_{H \sim \pgue}\qty[\tilde{n} \leq \loe(H,t_2) \leq 9\tilde{n}] \geq \frac{1}{2} - O(\tilde{d}^{-1}), \\
        \Pr_{H \sim \gue}\qty[|\OTOC_4(H,t_1)| \leq d^{-1/16}] \geq 1-\exp(-O(n)) &&\qq{while} \Pr_{H \sim \pgue}\qty[\tilde{d}^{-5} \leq \OTOC_4(H,t_2) \leq \tilde{d}^{-1}] \geq \frac{1}{2} - O(\tilde{d}^{-1}).
\end{eqnarray}\normalsize
Moreover, if we focus on the ensemble $\mathcal{E}_{\tilde{d}}$, i.e., 
using the iid spectral distribution $\tilde{p}_{\tilde{d}}$
(not just $8$-wise independent), the bounds for the right hand side sharpen further to $1-O(\tilde{d}^{-1/2})$.
\end{theorem}
\begin{proof}
We prove the statement for $S_2(\rho(H,t)_A)$. All other statements follow by similar proofs. 

First, we note that taking probabilities with respect to both $U \sim \mathcal{U}_{\text{PR}}$ and $\Lambda \sim \tilde{p}_{\tilde{d},8}$ amounts to taking a probability with respect to the entire Hamiltonian ensemble $\mathcal{E}_{\tilde{d},8}$. Combining the results of \cref{tab:expval-haar} with \cref{lem:haar-avg,cor:z-bound}, we have that for $t \geq \Omega(\tilde{d})$, one finds 

\begin{equation}\Pr_{H \sim \mathcal{E}_{\tilde{d},8}}\qty[\tr((\rho(H,t)_A)^2) \geq \tilde{d}^{-5}] \geq \frac{1}{2} - O(\tilde{d}^{-1}). 
\end{equation}
The use of the Markov inequality 
also implies that 
\begin{equation}
\Pr_{H \sim \mathcal{E}_{\tilde{d},8}}\qty[\tr((\rho(H,t)_A)^2) \geq \tilde{d}^{-1}] \leq O(\tilde{d}^{-1})
,
\end{equation}
and combining these two, we arrive at
\begin{equation}
    \Pr_{H \sim \mathcal{E}_{\tilde{d},8}}[\tilde{n} \leq S_2(\rho(H,t)_A) \leq 5\tilde{n}] \geq \frac{1}{2} - O(\tilde{d}^{-1}).
\end{equation}
When considering the spectral distribution $\tilde{p}_{\tilde{d}}$, i.e.,  the completely independent one, we can on the one hand plug in the bounds from \cref{lem:pd-indep} instead of the weaker bounds in \cref{cor:z-bound} to arrive at the stronger statement
\begin{equation}
    \Pr_{H \sim \pgue}[\tilde{n} \leq S_2(\rho(H,t)_A) \leq 5\tilde{n}] \geq 1 - O(\tilde{d}^{-1/2}).
\end{equation}
Conversely, for the GUE, we can use a simple Markov inequality, since
\begin{equation}
    \E_{H \sim \gue}[\tr((\rho(H,t)_A)^2)] = \E_{\Lambda \sim p}[\abs{Z(\Lambda t)}^4] + O(d^{-1}) = \E_{\Lambda \sim \tilde{p}}[\abs{Z(\Lambda t)}^4] + O(d^{-1/8}) = O(d^{-1/8}),
\end{equation}
where the second equality follows from \cref{thm:spoof}. Therefore,
\begin{equation}
    \Pr_{H \sim \gue}\qty[S_2(\rho(H,t)_A) \geq \frac{n}{16}] \leq O(d^{-1/16}) = \exp(-O(n)).
\end{equation}
\end{proof}

\subsection{Early-time Behavior}
The approach for the early-time behavior of the GUE and the pseudo-GUE is similar to the late-time approach; we calculate concentration inequalities for $\abs{Z(\Lambda t)}^{2k}$ under the distributions $\tilde{p}_{\tilde{d},4k}$ (\cref{cor:zlambt}) and $\tilde{p}_{\tilde{d}}$ (\cref{lem:zlambd-large}).
\begin{lemma}[Expectation of $Z$ for early times]\label{lem:zlambt}
Let $\tilde{p}_{\tilde{d},2k}$ the $2k$-wise independent degenerate distribution. Then, for any $k=O(1)$, we have
\begin{equation}
    \E_{\Lambda \sim \tilde{p}_{\tilde{d},2k}}\qty[\abs{Z(\lambda t)}^{2k}] = \qty(\frac{J_1(2t)}{t})^{2k} + O(d^{-1}).
\end{equation}
\end{lemma}
\begin{proof}
For the reminder of the proof, let us denote $\E:= \E_{\Lambda\sim \tilde{p}_{\tilde{d},2k}}$. We employ the expansion 
\begin{equation}
    \E\qty[\abs{Z(\Lambda t)}^{2k}] = \frac{1}{d^{2k}} \sum_{p_1,p_2,\ldots,p_{2k}} \E\qty[\exp(-i \sum_{m=1}^k (U_{p_{2m}}-U_{p_{2m+1}}))] .
\end{equation}
This sum can be partitioned into classes depending on how many unique indices there are in $\qty{p_1,\ldots,p_{2k}}$. For instance, one class corresponds to the case where each of the $p_m$ are unique, one class corresponds to the case where each of the $p_m$ are unique, except for one pair $p_{m_1}=p_{m_2}$ which is the same, and so on, until we get to the case where $p_1=p_2=\ldots=p_{2k}$. The first class contains 
\begin{equation}
\prod_{j=0}^{2k-1} (d-j) = \frac{d!}{(d-2k)!}
\end{equation}
many terms, the second class contains $\binom{2k}{2}\frac{d!}{(d-2k+1)!} = O(d^{2k-1})$, and the last class contains exactly $d$ terms. Simply by counting, we see that each class except the first contributes at most $O(d^{-1})$ to the sum. Therefore,
\begin{equation}
    \E\qty[\abs{Z(\Lambda t)}^{2k}] = \frac{d!}{d^{2k} (d-2k)!} \qty(\E[\exp(-i U_1)])^{k} \cdot \qty(\E[\exp(i U_1)])^{k} + O(d^{-1}).
\end{equation}
We can analytically calculate $\E_{U_1 \sim p^{(1)}}[\exp(\pm i U_1)] = \frac{J_1(2t)}{t}$, and also using Stirling's approximation, infer that $\frac{d!}{d^{2k} (d-2k)!} = 1 - O(k^2/d)$. Therefore, we can conclude
\begin{equation}
    \E_{\Lambda\sim \tilde{p}_{\tilde{d},2k}}\qty[\abs{Z(\lambda t)}^{2k}] = \qty(\frac{J_1(2t)}{t})^{2k} + O(d^{-1}),
\end{equation}
as desired.
\end{proof}

\begin{lemma}[Concentration of $Z$ under $4k$-wise independence]\label{cor:zlambt}
Let $\tilde{p}_{\tilde{d},4k}$ the $4k$-wise independent degenerate distribution. Then, for any $k=O(1)$, we have
\begin{equation}
    \Pr_{\Lambda \sim \tilde{p}_{\tilde{d},4k}}\qty[\abs{\abs{Z(\Lambda t)}^{2k} - \qty(\frac{J_1(2t)}{t})^{2k}} \geq d^{-1/4}] \leq O(d^{-1/2}),
\end{equation}
when $\tilde{p}$ is a $4k$-wise independent distribution.
\end{lemma}
\begin{proof}
Simply applying \cref{lem:zlambt} twice, we can see that the variance of $\abs{Z(\Lambda t)}^{2k}$ is at most $O(d^{-1})$. Therefore, since 
\begin{equation}\abs{\abs{Z(\Lambda t)}^{2k} - \E_{\Lambda \sim \tilde{p}}[\abs{Z(\Lambda t)}^{2k}]} \leq \abs{\abs{Z(\Lambda t)}^{2k} - \qty(\frac{J_1(2t)}{t})^{2k}} + O(d^{-1}),
\end{equation}
we can conclude
\begin{equation}
    \Pr_{\Lambda \sim \tilde{p}_{\tilde{d},4k}}\qty[\abs{\abs{Z(\Lambda t)}^{2k} - \qty(\frac{J_1(2t)}{t})^{2k}} \geq d^{-1/4}] \leq O(d^{-1/2}).
\end{equation}
\end{proof}

\begin{lemma}[Concentration of $Z$ under full independence]\label{lem:zlambd-large}
Let $\tilde{p}_{\tilde{d}}$ be the iid degenerate distribution. Then, for any $k=O(1)$ and any $t$, it holds that
\begin{equation}
    \Pr_{\Lambda \sim \tilde{p}_{\tilde{d}}}[\abs{Z(\Lambda t)}^{2k} \geq \tilde{d}^{-2k}] \geq 1 - O(\tilde{d}^{-1/2}).
\end{equation}
\end{lemma}
\begin{proof}
Let $\mu = \frac{J_1(2t)}{t}$, $\sigma_x^2 = \frac{1}{2}\qty(1+\frac{J_1(4t)}{2t}) - \frac{J_1(2t)^2}{t^2}$, and $\sigma_y^2=\frac{1}{2}\qty(1-\frac{J_1(4t)}{2t})$. Setting $X \sim \mathcal{N}(\mu, \sigma_x^2/\tilde{d})$ and $Y \sim \mathcal{N}(0, \sigma_y^2/\tilde{d})$, from \cref{lem:pd-indep}, we have
\begin{equation}
    \abs{\Pr_{\Lambda \sim \tilde{p}_{\tilde{d}}}[\abs{Z(\Lambda t)}^{2k} \leq \tilde{d}^{-2k}] - \Pr[X^2 + Y^2 \leq \tilde{d}^{-2}]} \leq O(\tilde{d}^{-1}).
\end{equation}
We discuss two cases separately. 
\begin{itemize}
    \item For $t < 7/5$, we write 
    \begin{equation}
        \Pr[X^2 + Y^2 \leq \tilde{d}^{-2}] \leq \Pr[X \leq \tilde{d}^{-1}] \leq \Phi\qty(\frac{\tilde{d}^{-1/2} - \tilde{d}^{1/2} \mu}{\sigma_x}),
    \end{equation}
    where $\Phi$ is the CDF of the standard normal distribution. Since $\mu \geq \frac{1}{4}$ and $\sigma_x \leq \frac{1}{2}$ for $t < 7/5$, we can upper bound this with $\Phi\qty(2\tilde{d}^{-1/2}-\frac{\tilde{d}^{1/2}}{2})=O(\tilde{d}^{1/2})$.
    \item For $t \geq 7/5$, we have $\sigma_x,\sigma_y \geq \frac{1}{4}$. For this case, we explicitly spell out
    \begin{equation}
        \Pr[X^2 + Y^2 \leq \tilde{d}^{-2}] = \frac{\tilde{d}}{2\pi \sigma_x \sigma_y} \int_U \exp(-(\vec{r}_1 - \mu)^2/2\sigma_x^2-(\vec{r}_2)^2/2 \sigma_y^2) \dd{\vec{r}},\label{eq:d36}
    \end{equation}
    where $U$ is a disk of radius $\tilde{d}^{-1}$ centered at the origin. We simply upper bound the integrand by $1$. Then, \cref{eq:d36} is upper bounded by $\frac{16\tilde{d}}{2\pi} \cdot 2\pi \tilde{d}^{-2} = 16\tilde{d}^{-1}$.
\end{itemize}
\end{proof}

Combining these two concentration lemmas enables us to prove the main theorem below.
\begin{theorem}[Early time behavior of GUE versus pseudo-GUE, \cref{th:separations} in the main text]\label{App:separationproof}
Let $\rho(H,t) \coloneqq e^{-iHt} \ketbra{0} e^{iHt}$, $\ket{V(H,t)} = (e^{-iHt} V e^{iHt} \otimes \id) \ket{\phi^+}$, and $\OTOC_4(H,t)= \tr(P_1(t) P_2 P_1(t) P_2)$, with $P_1 = e^{-iHt} P_1 e^{iHt}$. Let $\tilde{n} \coloneqq \log_2 \tilde{d}$. We define $\loe(H,t) = S_2(V(H,t)_{A\cup A'})$. For the pseudo-GUE Hamiltonians, assume the eigenbases form a relative $\epsilon$-approximate $16$-design with $\epsilon=O(d^{-1})$, and consider the iid spectral distribution $\tilde{p}_{\tilde{d}}$ and the corresponding ensemble $\mathcal{E}_{\tilde{d}}$. Then, $\exists t_1 = O(1)$ and $\forall t_2$, the following holds.
\small\begin{eqnarray}
        \Pr_{H \sim \gue}\qty[S_2(\rho(H,t_1)_A) \geq \frac{n}{12}] \geq 1-\exp(-O(n)) &&\qq{while} \Pr_{H \sim \pgue}\qty[S_2(\rho(H,t_2)_A) \leq 4\tilde{n}] \geq 1-O(\tilde{d}^{-1/2}), \\
        \Pr_{H \sim \gue}\qty[M_2(\rho(H,t_1)) \geq \frac{n}{12}] \geq 1-\exp(-O(n)) &&\qq{while} \Pr_{H \sim \pgue}\qty[M_2(\rho(H,t_2)_A) \leq 8\tilde{n}] \geq 1-O(\tilde{d}^{-1/2}), \\
        \Pr_{H \sim \gue}\qty[\loe(H,t_1) \geq \frac{n}{12}] \geq 1-\exp(-O(n)) &&\qq{while} \Pr_{H \sim \pgue}\qty[\loe(H,t_2) \leq 8\tilde{n}] \geq 1-O(\tilde{d}^{-1/2}), \\
        \Pr_{H \sim \gue}\qty[|\OTOC_4(H,t_1)| \leq d^{-1/12}] \geq 1-\exp(-O(n)) &&\qq{while} \Pr_{H \sim \pgue}\qty[\tilde{d}^{-4} \leq \OTOC_4(H,t_2)] \geq 1-O(\tilde{d}^{-1/2}).
\end{eqnarray}\normalsize
\end{theorem}
\begin{proof}
We prove the statement for $S_2(\rho(H,t)_A)$. All other statements follow by similar proofs. The proof follows in a nearly identical fashion to \cref{th:late-time}, except we plug in bounds from \cref{cor:zlambt,lem:zlambd-large}. We have that for pseudo-GUE $\mathcal{E}_{\tilde{d}}$
\begin{equation}
    \Pr_{H \sim \mathcal{E}_{\tilde{d}}}\qty[\tr((\rho(H,t_2)_A)^2) \geq \tilde{d}^{-4}] \geq 1 - O(\tilde{d}^{-1/2}),
\end{equation}
for all times $t_2$. Meanwhile,
\begin{equation}
    \Pr_{H \sim \gue}\qty[\abs{\tr((\rho(H,t)_A)^2)-\qty(\frac{J_1(2t)}{t})} \leq d^{-1/12}] \geq 1 - O(d^{-1/6}).
\end{equation}
Since $J_1(2t)$ has infinitely many roots $t_1 \approx 1.92, 3.51, \ldots$, we can choose any one of these roots to get 
\begin{equation}
    \Pr_{H \sim \gue}\qty[\tr((\rho(H,t_1)_A)^2) \leq d^{-1/12}] \geq 1 - O(d^{-1/6}).
\end{equation}
This concludes the proof.
\end{proof}

\section{Applications}

\subsection{Distillation bounds}\label{App:distillationbounds}
In this section, we discuss the proof of \cref{cor:distillationbounds}. As emphasized in the main text, the evolution generated by pseudochaotic Hamiltonians with $\tilde{d} = \Theta(\exp(\poly\log n))$ introduces a stronger concept: pseudoresourceful unitaries. Indeed, \cref{th:separations} shows that, at many instances in time, the behavior of quantum resources, such as entanglement and magic, differs significantly. More rigorously, given an extensive bipartition $A|B$, an ensemble of pseudoentangling unitary operators $\mathcal{E}_{\tilde{U}}$ produces states of the form $U\ket{0}$ for $U \sim \mathcal{E}_{\tilde{U}}$, with entanglement $\Theta(\poly\log n)$, and it is indistinguishable from an ensemble of unitaries producing highly entangled states, i.e., with entanglement $\Theta(n)$.

Thanks to \cref{th:separations}, there exists a time $t^{*}$ such that the ensemble of GUE Hamiltonians generates highly entangled states with overwhelming probability, while the ensemble of pseudo-GUE Hamiltonians $\mathcal{E}_{\tilde{d}}$, with $\tilde{d} = \Theta(\exp(\poly\log n))$, produces states with entanglement $O(\poly\log n)$ with high probability. As such, these pair of ensembles constitute an example of pseudoentangling unitaries. Moreover, the following lemma demonstrates that these two ensembles generate the maximal achievable gap, i.e., $\log n$ vs. $n$, for pseudoentangling unitaries.

\begin{lemma}[Lower bounds for pseudoentangling unitaries]\label{lem:lowerpseudoentanglingunitaries} Let $\mathcal{E}_{\tilde{U}}$ an ensemble of unitaries that (i) produce state vectors $U\ket{0}$ with entanglement upper bounded by $f(n)$ with high probability and (ii) is indistinguishable from an ensemble of unitaries producing high entangled states, i.e., $\Omega(n)$. Then
\be
f(n)=\omega(\log n).
\ee
\begin{proof}
The result follows from the observation that if $f(n)=O(\log n)$, then a simple swap test could distinguish between the ensembles. Hence one should necessarily have $f(n)=\omega(\log n)$. 
\end{proof}
\end{lemma}
The above lemma shows that our pseudo-GUE Hamiltonians constitute an \textit{optimal} ensemble of pseudoentangling unitaries. We can now 
prove \cref{cor:distillationbounds}, which we restate for convenience:
\begin{corollary}[\cref{cor:distillationbounds} in the main text]\label{cor:distillationbounds2}
Consider a general unitary $U$, a reference state $\rho_0$ and let $\rho \coloneqq U \rho_0 U^\dagger$. Consider any efficient ``designer'' quantum algorithm $\mathcal{D}^U$ with query access to $U$, which outputs some efficient LOCC algorithm $\mathcal{A}_{\locc}$, between two parties $A|B$, for distilling $m$ 
copies of a target state per copy of $\rho$, then
\begin{equation}
    m = O(\log^{1+c} S_2(\rho_A))
\end{equation}
where $\rho_A=\tr_B\rho$. This holds for any constant $c>0$ and every extensive bipartition $A|B$. 
\begin{proof}
The proof simply follows from the fact that if such an efficient and agnostic LOCC protocol distilling $\Omega(\log^{1+c}S_2)$ existed, it would immediately distinguish between the ensemble of pseudoentangling unitaries introduced above. First note that for every $c$, we can construct a pseudo-GUE ensemble obeying $S_{2}=O(\log^{1+c/2}n)$ by opportunely tuning $\tilde{d}$. Therefore, after the application of such a hypothetical protocol, a distinguishing method would be provided by counting the output Bell pairs, which, in the case of pseudoentangling unitaries, would be $O(\log^{1+c/2}n)$ (by basic rules of resource conversion~\cite{RevModPhys.91.025001}), and in the case of highly entangling unitaries, would be $\Omega(\log^{1+c}n)$ by hypothesis. Hence, the result follows.
\end{proof}
\end{corollary}

\begin{remark}
In a similar fashion, the pseudo-GUE ensemble generates pseudomagic unitaries, i.e., unitaries that produce quantum states $\ket{\phi}$ with $M_{2}(\ket{\phi}) = O(\poly \log n)$, which is sufficient for magic state distillation bounds given its monotonicity~\cite{PhysRevA.110.L040403}. As such, a lemma identical to \cref{lem:lowerpseudoentanglingunitaries} and a corollary identical to \cref{cor:distillationbounds} follow directly. We omit the proof, as it straightforwardly follows from the techniques developed in Ref.~\cite{gu2024pseudomagic}.
\end{remark}

\subsection{Hamiltonian learning}\label{applemmahardnesslearning}
\begin{lemma}\label{lemlemmahardnesslearning}
    Given black-box query access to an \emph{efficiently implementable} pseudo-GUE Hamiltonian $H$, $\mathcal{E}_{\tilde{d},k}$ with $k=O(1)$, there is some $t^*=O(1)$ such that for any $t \geq t^*$, any \emph{computationally efficient} quantum algorithm which outputs a circuit description of a unitary $U \equiv A^H()$ satisfies
    \begin{equation}
        \Pr_{H \sim \mathcal{E}_{\tilde{d},k}}[\abs{\expval{U^\dagger e^{-iHt}}{0}}^2 \geq 1 - \Omega(1)] \leq o(1).
    \end{equation}
\end{lemma}
\begin{proof}
Let us make use of the key result of Ref.~\cite{kotowski2023extremal}, summarized in \cref{lem:kotowski2023extremal}, i.e., for $H \sim \gue$, there exists some $t^*=O(1)$ such that for any $t \geq t^*$, the circuit complexity of $e^{-iHt} \ket{0}$ is $\omega(\poly n)$ with high probability over the choice of $H$. This implies that there is some non-zero constant $C = \Omega(1)$ such that \emph{any} polynomial size unitary $U$ satisfies
\begin{equation}
    \abs{\expval{U^\dagger e^{-iHt}}{0}}^2 \leq 1 - C
\end{equation}
with high probability. Since any \emph{computationally efficient} quantum algorithm $\mathcal{A}^H$ can output a circuit description of $U \equiv \mathcal{A}^H()$ with size at most $O(\poly n)$, we must have
\begin{equation}
    \Pr_{H \sim \gue}[\abs{\expval{U^\dagger e^{-iHt}}{0}}^2 \geq 1 - C] \leq o(1).
\end{equation}
We finally apply \cref{th:ham-indisting}, i.e.,  that $\mathcal{E}_{\tilde{d},k}$ is indistinguishable from $\gue$ given $\tilde{d}=\omega(\poly n)$. Since $\abs{\expval{U^\dagger e^{-iHt}}{0}}^2$ is efficiently measurable with accuracy up to $1/\poly(n)$, we must have
\begin{equation}
    \Pr_{H \sim \mathcal{E}_{\tilde{d},k}}[\abs{\expval{U^\dagger e^{-iHt}}{0}}^2 \geq 1 - C'] \leq o(1),
\end{equation}
for $C' = C-O(1/\poly n)$. If this were not the case, measuring $\abs{\expval{U^\dagger e^{-iHt}}{0}}^2$ would give an efficient distinguisher between $H \sim \gue$ and $H \sim \mathcal{E}_{\tilde{d},k}$, forbidden by \cref{th:ham-indisting}. This concludes the proof.
\end{proof}

\begin{theorem}[No general Hamiltonian learning. \cref{th:nohamiltonianlearning} in the main text] \label{th:hardnessdiamond}
Consider any efficient quantum algorithm $\mathcal{A}^H(t,\varepsilon)$ that, given black-box access to a Hamiltonian $H$, attempts to output a circuit description $U$ approximating $e^{-iHt}$ within $\varepsilon$ error in diamond norm. There exists a time $t=O(1)$ and error $\varepsilon=O(1)$, such that even when restricted to Hamiltonians $H$ that are $O(1)$-sparse in the computational basis and have efficiently implementable time evolutions $e^{-iHt}$, there exist Hamiltonians $H$ for which any efficient quantum algorithm $\mathcal{A}^H$ fails to produce a correct $\varepsilon$-approximation.
\begin{proof}
We make use of \cref{lemlemmahardnesslearning} which says that there exists a efficiently implementable Hamiltonian, specifically a pseudochaotic Hamiltonian belonging to $\mathcal{E}_{\tilde{d},k}$ (with $d=\omega(\poly n)$) such that 
\be
|\langle0|U^{\dag}e^{-iHt}|0\rangle|\le 1-\Omega(1)
\ee
for every $U$ ouput of a efficient quantum algorithm $\mathcal{A}^{H}$. Consider the diamond norm distance $\|e^{-iHt}-U\|_{\diamond}$, then the following inequalities hold
\be
\|e^{-iHt}-U\|_{\diamond}\ge \|e^{-iHt}\ket{0}-U\ket{0}\|=\sqrt{2-2|\langle0|U^{\dag}e^{-iHt}|0\rangle|}\ge\Omega(1)
\ee
which prove the statement. 

Let us now show that the result still holds for sparse Hamiltonians in the computational basis. Let us consider the following ensembles of Hamiltonians 
\be
\mathcal{E}_1\coloneqq\qty{\sum_{i}\lambda_i\ketbra{i}\,:\, \lambda_{1},\ldots,\lambda_{d}\sim \tilde{p}}
\ee
where $\tilde{p}=\prod_{i=1}^{d}p^{(1)}(\lambda)$ is the product of iid random variables distributed according to the Wigner semi-circle distribution. We then consider the ensemble of efficiently implementable Hamiltonians
\be
\mathcal{E}_2\coloneqq\qty{\sum_{i}\lambda_i\ketbra{i}\,:\, \lambda_{1},\ldots,\lambda_{d}\sim \tilde{p}_{d,k}}
\ee
where $\tilde{p}_{d,k}$ is both computationally indistinguishable from $\tilde{p}$, and is exactly $k$-wise independent. Time evolution under this ensemble is efficiently implementable for $k=O(1)$ (see \cref{fact:prf}). Let us consider the output $U=\mathcal{A}^{H}()$ of an efficient quantum algorithm and consider the following quantity
\be
\mathbb{E}_{\ket{\psi}}[\abs{\expval{U^\dagger e^{-iHt}}{\psi}}^2]=\frac{\abs{\tr(U^{\dag}e^{-iHt})}^2}{d^2}+O(d^{-1})
\label{eq:averageoverlap}
\ee
where the average is 
over any $2$-state design (e.g., random stabilizer states). 
In Ref.~\cite{kotowski2023extremal}, it has been proven that for any efficient algorithm $\mathcal{A}^H$,
\be
\Pr_{H\sim\mathcal{E}_{\text{gauss}}}[\abs{\tr(U^{\dag}e^{-iHt})}\ge \alpha' d]\le \exp(-\Omega(d))\qc \mathcal{E}_{\text{gauss}}\coloneqq\qty{\sum_{i}\lambda_i\ketbra{i}\,:\, \lambda_{1},\ldots,\lambda_{d}\sim \mathcal{N}(0,1)}
\ee
with $\alpha' = \Theta(1)$. Despite the fact that theorem uses energies $\lambda_i$ which are normally distributed, their proof can be replicated simply using the Wigner semi-circle law instead of the standard normal. That is, the following holds as well:
\be
\Pr_{H\sim\mathcal{E}_{1}}[|\tr(U^{\dag}e^{-iHt})|< \alpha d]>1- \exp(-\Omega(d))
\ee
with $\alpha = \Theta(1)$. 
%Therefore, we have that $e^{-iHt}$ has exponential gate complexity with high probability over the ensemble $H\sim\mathcal{E}_1$.
The latter thus implies that the average overlap in \cref{eq:averageoverlap} obeys
\be
\Pr_{H\sim \mathcal{E}_1}[\mathbb{E}_{\ket{\psi}}[|\langle\psi|U^{\dag}e^{-iHt}|\psi\rangle|^2]\le 1-\Omega(1)]\ge 1-\exp(-\Omega(d)).
\ee
However, since the average overlap can be efficiently measured (up to inverse polynomial precision, with failure probability exponentially small in $n$) by randomly drawing a stabilizer state and then estimating the overlap with a standard swap test, this also implies that for the ensemble of efficiently implementable (yet computationally indistinguishable) Hamiltonians belonging to $\mathcal{E}_2$, we have
\be
\Pr_{H\sim \mathcal{E}_2}[\mathbb{E}_{\ket{\psi}}[|\langle\psi|U^{\dag}e^{-iHt}|\psi\rangle|^2]\le 1-\Omega(1)]\ge 1-\negl(n).
\ee
%This implies that there exists a state $\ket{\phi}$, such that $e^{-iHt}\ket{\phi}$ has exponential gate complexity for the state preparation. 
Since $\abs{\expval{U^\dagger e^{-iHt}}{\psi}} \leq \norm{e^{-iHt} - U}_\diamond$ for any $\ket{\psi}$, we conclude that there exists a Hamiltonian in $\mathcal{E}_2$ such that
\be
\|e^{-iHt}-U\|_{\diamond}\ge \Omega(1)
\ee
which concludes the proof. 
\end{proof}
\end{theorem}

\subsection{Scrambling property testing}\label{App:scramblingtesting}
In this section, we lower bound the query access to a unitary $U$ to solve the following decision problem, defined as \textit{scrambling property testing}.
\begin{definition}[Scrambling property tester] An algorithm $\mathcal{A}$ is a scrambling property tester if, given $M$ copies to the unitary generated by a Hamiltonian $H$ or, more simply, a unitary $U$ and two threshold values $k,K\in [0,1]$ with $k<K$, obeys the following conditions.
\begin{itemize}
    \item If $\abs{\OTOC_4(U)}\le k$ then $\mathcal{A}$ accepts with probability $1/3$.
    \item If $\abs{\OTOC_4(U)}> K$ then $\mathcal{A}$ accepts with probability $2/3$.
\end{itemize}
\end{definition}

\begin{theorem}[Sample-complexity lower bounds for scrambling tester. \cref{cor:scramblingtesting} in the main text.]
    Given $k>2^{-\Omega(n)}$, any scrambling tester, with query access to possibly controlled version of $U$ and its conjugate, requires $\Omega(K^{-\frac{1}{2+c}})$ (for $c>0$) many samples. 
\begin{proof}
Thanks to \cref{th:separations}, we know that there exists a $t^{*}$ such that
\be
\Pr_{H\sim {\gue}}[\abs{\OTOC(H, t^{*})}\le 2^{-\Omega(n)}]\ge 1-2^{-\Omega(n)}.\label{Eq:11}
\ee
Let us define the following two ensembles of Hamiltonians 
\begin{equation}
    \begin{aligned}
        \mathcal{E}^{t^*}_{\le k}&\coloneqq\qty{H\,:\, \abs{\OTOC_4(H, t^*)}\le k }, \\
\mathcal{E}_{> k}^{ t^*}&\coloneqq \qty{H\,:\, \abs{\OTOC_4(H, t^*)}> k }.        
    \end{aligned}
\end{equation}
The ensemble of Hamiltonians are distinguished by the value of the OTOC attained by the time evolution $e^{-iH t^{*}}$ of at most $k$ or at least $k$. 
Let us define $\pi\coloneqq\Pr_{H\sim \gue}[H\in \mathcal{E}_{\ge k}^{t^*}]$ and notice that \cref{th:separations} implies that if $k>2^{-\Omega(n)}$, there exists $t^{*}$ such that $\pi<2^{-\Omega(n)}$. This implies that
\be
\abs{\Pr_{U\sim {\gue}}[\mathcal{A}^{H}()=1]-\Pr_{U\sim \mathcal{E}_{<k}^{t^*}}[\mathcal{A}^{H}()=1]}<2^{-\Omega(n)}.
\ee
Similarly, for the pseudo-GUE ensemble of Hamiltonians $\mathcal{E}_{\tilde{d}}$, we can define
\be
{\mathcal{E}}_{\tilde{d},\ge}\coloneqq\{H\in \mathcal{E}_{\tilde{d}}\,:\, \OTOC(H,t^{*})\ge \Omega(\tilde{d}^{-1})\}.
\ee
By virtue of \cref{th:separations}, we have that
\be
\left|\Pr_{U\sim {\mathcal{E}}_{\tilde{d}}}[\mathcal{A}^{H}()=1]-\Pr_{U\sim {\mathcal{E}}_{\tilde{d},\ge}}[\mathcal{A}^{H}()=1]\right|<O(\tilde{d}^{-1}).
\ee
Thanks to \cref{th:pseudo-GUE}, we have that the pseudo-GUE ensemble is indistinguishable from the true GUE ensemble, which means that for every algorithm $\mathcal{A}^{H}$ which makes at most $M$ queries to (possibly controlled versions) of $e^{-iHt^{*}}$ and its conjugate $e^{iHt^{*}}$  cannot distinguish the two ensembles unless $M\ge\Omega(\tilde{d})$. More specifically, employing the result of \cref{lem:degen}, we can write
\be
\abs{\Pr_{U\sim {\gue}}[\mathcal{A}^{H}()=1]-\Pr_{U\sim {\mathcal{E}}_{\tilde{d}}}[\mathcal{A}^{H}()=1]}<O\left(\frac{M^2}{\tilde{d}}\right).
\ee
Hence, by a simple application of triangle inequalities, we find that
\be
\abs{\Pr_{U\sim \mathcal{E}_{\tilde{d},\ge}}[\mathcal{A}^{H}()=1]-\Pr_{U\sim \mathcal{E}_{<k}^{t^*}}[\mathcal{A}^{H}()=1]}<O\left(\frac{M^2}{\tilde{d}}\right).
\ee
Therefore, to solve the scrambling property testing problem with probability $\Omega(1)$ one needs a number of samples 
\be
M=\Omega(\tilde{d}^{1/2})\label{eq212}
\ee 
to distinguish between the ensemble $\mathcal{E}_{<k}^{t^*}$ and $\mathcal{E}_{\tilde{d},\ge}$. However, since the ensemble $\mathcal{E}_{<k}^{t^*}$ is defined by $\abs{\OTOC(H,t^{*})}<k$, while  the ensemble $\mathcal{E}_{\tilde{d},\ge}$ has $\abs{\OTOC(H, t^{*})}\ge \Omega(\tilde{d}^{-1})$, the result follows. To see this, let us suppose for the sake of contradiction that we need $O(K^{-1/(2+c)})$ for some constant $c>0$ to solve the scrambling property testing problem. We can then construct a pseudo-GUE ensemble with $\tilde{d}=K^{-1}$ and distinguish it from $\mathcal{E}_{<k}^{(t^*)}$ using only $O(K^{-1/(2+c)})$. However, by virtue of \cref{eq212}, we need $M=\Omega(K^{-1/2})$ and this is a contradiction. This concludes the proof. 
\end{proof}
\end{theorem}
\begin{remark}[Features of entanglement and stabilizer
Rényi entropies]
Analogous results for the local operator entanglement $\loe(H,t)$, the entanglement entropy $S_{2}$, or the 2-stabilizer Rényi entropy $M_2$ can be derived using a similar approach. The only ingredient required are \cref{th:separations} and \cref{th:pseudo-GUE}.
\end{remark}

\begin{open}[Preventing scrambling property testers]
The above no-go result prevents a scrambling property tester for values of $k^{-1} = \omega(\poly(n))$; however, it leaves open the possibility of efficient testing for $k^{-1} = O(\poly(n))$. A naive algorithm, which measures the 4-point OTOC by applying $U, U^{\dag}, P, Q$ in the correct order to a uniformly randomly selected bitstring $\ket{x}$, and then measuring the overlap with $\ket{x}$ itself, achieves a sample complexity of $O(k^{-2})$. An interesting open question is whether it is possible to design a more efficient property tester, possibly matching the lower bound $\Omega(k^{-1/2})$.
\end{open}

\subsection{Testing spectral statistics}\label{App:testingESS}
In this section, we give a proof of \cref{cor:ESSth}. Most of our previous results focused on properties of individual Hamiltonians sampled from an ensemble. This ``single-shot'' setting is natural when studying physical systems: we typically have access to a single realization of a quantum system and want to determine its properties through measurements. However, certain physical properties are inherently statistical in nature and can only be meaningfully defined for ensembles of systems. A prime example is level repulsion --- while we can measure the energy gaps of a single Hamiltonian, determining whether an ensemble exhibits level repulsion requires examining the statistical distribution of these gaps across many samples.

This ensemble testing scenario connects to the rich literature of distribution testing in classical computer science~\cite{goldreich1998testing,batu2000testing,goldreich2011complexity}. In the classical setting, given sample access to an unknown distribution, one aims to determine whether it has certain properties (e.g., uniformity, independence) or is ``far'' from having these properties. The quantum setting introduces several new challenges: while classical distribution testing typically assumes independent samples from the distribution, quantum measurements on copies of the same state are constrained by no-cloning. Moreover, in our setting, we don't have direct access to the eigenvalue distribution --- we can only access it through measurements of time evolution and thermal states. Our results on pseudochaotic Hamiltonians have direct implications for quantum ensemble testing. If two ensembles are computationally indistinguishable yet have drastically different statistical properties (like level statistics), this establishes fundamental limits on our ability to test these properties efficiently.

First, we prove the fact that the spectral statistics of the spoofing GUE are Poissonian.
\begin{lemma}[Poissonian statistics of $\tilde{p}$]\label{lem:poissonian}
We define the normalized gap $\hat{s}$ to be $d \cdot s$, where $s$ is the difference between any given energy $\lambda_i$ and the next largest energy $\lambda_j$. For the iid distribution $\tilde{p}_{d}$, $P(\hat{s}=0) = \frac{8}{3\pi^2}$.
\end{lemma}
\begin{proof}
The PDF of the gap is
\begin{equation}\label{eq:ps}
    P(s) = d \int_{-2}^{-2+s} p^{(1)}(\lambda) p^{(1)}(\lambda+s) (1+F(\lambda)-F(\lambda+s))^{d-2} \dd{\lambda},
\end{equation}
where $F(\lambda) := \int_{-2}^\lambda p^{(1)}(x) \dd{x}$ is the CDF of the Wigner semi-circle distribution. The justification of this equation is that in order for the gap to be $s$, if one energy is $\lambda$ we require another energy to be at $\lambda+s$, and we require all other energies to be outside the window $[\lambda,\lambda+s]$, which has probability $(1+F(\lambda)-F(\lambda+s))^{d-2}$; a similar derivation is presented in Ref.~\cite[Sec. 2.3]{livan2018introduction}. Evaluating at $s=0$, we have
\begin{equation}
    P(s=0) = d \int_{-2}^2 f(\lambda)^2 \dd{\lambda} = d \int_{-2}^2 \frac{4-\lambda^2}{4\pi^2} \dd{\lambda} = \frac{8}{3\pi^2} d.
\end{equation}
To show that the distribution is monotonically decreasing, it suffices to differentiate \cref{eq:ps} with respect to $s$ and observe that it is everywhere non-positive.
\end{proof}

To formally define what it means to test for a property of an \emph{ensemble} of Hamiltonians, we first define black-box query access to an ensemble of Hamiltonians.
\begin{definition}[Black-box ensemble access]\label{def:black-box-ensemble}
We say an algorithm has black-box access to an ensemble of Hamiltonians $\mathcal{E}$ if it can make $K$ queries to the ensemble. This means it can request polynomially many Hamiltonians $H_i$ randomly sampled from $\mathcal{E}$, for $i=1,\ldots,k$, and it has access to
\begin{itemize}[noitemsep]
    \item $g_i$ samples of the Gibbs states $e^{-\beta H_i}/\tr(e^{-\beta H_i})$ for any $\beta \leq O(\poly n)$, and
    \item $q_i$ applications of $e^{-iH_i t}$ (as well as its controlled version) for any $\abs{t} \leq \exp(O(n))$,
\end{itemize}
where $K = \sum_{i=1}^k g_i + q_i$. We will denote algorithms $\mathcal{A}$ with black-box access to an ensemble of Hamiltonians as $\mathcal{A}^{\mathcal{E}}$.
\end{definition}

\begin{corollary}[No efficient level repulsion testing. \cref{cor:ESSth} in the main text]
There exists two classes of Hamiltonians $\mathcal{E}_1$ and $\mathcal{E}_2$ that satisfy the following:
\begin{itemize}[noitemsep]
    \item they are $O(1)$-sparse in the computational basis;
    \item the eigenvalue spectral statistics (ESS) for $\mathcal{E}_1$ exhibit level repulsion, while the ESS for $\mathcal{E}_2$ does not;
    \item yet, any algorithm with black-box ensemble access to the ensembles $\mathcal{E}_1$ or $\mathcal{E}_2$ that makes at most $K=O(\poly n)$ queries to either ensemble satisfies
    \begin{equation}\label{eq:e27}
        \abs{\Pr[\mathcal{A}^{\mathcal{E}_1}()]-\Pr[\mathcal{A}^{\mathcal{E}_2}()]} \leq \negl(n).
    \end{equation}
\end{itemize}
\begin{proof}
We first define the two ensembles of Hamiltonians. We select both ensembles to consist of Hamiltonians that are diagonal in the computational basis. The first ensemble follows the correlated eigenvalue statistics of the GUE, and the second has entirely uncorrelated energies distributed according to $\tilde{p}_{\tilde{d}}$ with $\tilde{d} = \omega(\poly n)$:
\begin{equation}
    \begin{aligned}
    \mathcal{E}_1&\coloneqq\qty{\sum_{i}\lambda_i\ketbra{i}\,:\, \lambda_{1},\ldots,\lambda_{d}\sim p },\\
    \mathcal{E}_2&\coloneqq\qty{\sum_{i}\lambda_i\ketbra{i}\,:\, \lambda_{1},\ldots,\lambda_{d}\sim \tilde{p}_{\tilde{d}}}.
    \end{aligned}
\end{equation}
These ensembles trivially satisfy the first claim of the corollary, and the absence of level repulsion for the second ensemble follows from \cref{lem:poissonian}. 

The remainder of the proof consists in showing \cref{eq:e27}. This follows essentially from applying a triangle inequality (i.e., a standard `hybrid argument'~\cite{fischlin2021hybrid}) to accesses of individual Hamiltonians. To make it explicit, observe that the output of an algorithm $\mathcal{A}^{\mathcal{E}_1}$ which makes $K=O(\poly n)$ queries to the ensemble $\mathcal{E}_1$ can be written
\begin{equation}
    \E_{H_1 \sim \mathcal{E}_1} \E_{H_2 \sim \mathcal{E}_1} \ldots \E_{H_k \sim \mathcal{E}_1} \qty[\mathcal{A}^{H_1,\ldots,H_k}()],
\end{equation}
where $\mathcal{A}^{H_1,\ldots,H_k}$ denotes an algorithm which has black-box query access to each of the $H_1,\ldots,H_k$ Hamiltonians. By assumption that $K=O(\poly n)$, observe that $k=O(\poly n)$, and $\mathcal{A}^{H_1,\ldots,H_k}$ can make at most $\ell=O(\poly n)$ queries to each individual Hamiltonian. Now, holding each of the $H_1,\ldots,H_{k-1}$ fixed, observe \cref{th:ham-indisting} implies that
\begin{equation}
    \norm{\E_{H_k \sim \mathcal{E}_1} \qty[\mathcal{A}^{H_1,\ldots,H_k}()]-\E_{H_k \sim \mathcal{E}_2} \qty[\mathcal{A}^{H_1,\ldots,H_k}()]}_1 \leq \negl(n).
\end{equation}
Applying the triangle inequality $k=O(\poly n)$ times, we conclude that 
\begin{equation}
    \norm{\E_{H_1,\ldots,H_k \sim \mathcal{E}_1}\qty[\mathcal{A}^{H_1,\ldots,H_k}()] - \E_{H_1,\ldots,H_k \sim \mathcal{E}_2}\qty[\mathcal{A}^{H_1,\ldots,H_k}()]} \leq k \cdot \negl(n) = \negl(n),
\end{equation}
and since $\E_{H_1,\ldots,H_k \sim \mathcal{E}_1}\qty[\mathcal{A}^{H_1,\ldots,H_k}()] = \E[\mathcal{A}^{\mathcal{E}_1}]$ (with an identical relation for $\mathcal{E}_2$), this shows the desired statistical indistiguishability of the two ensembles, which proves \cref{eq:e27}.
\end{proof}
\end{corollary}

\section{Indistinguishability of late-time GUE from Haar random unitaries}\label{App:GUEHaar}
In this section, we rigorously demonstrate an intuitive result: late-time evolution generated by the GUE ensemble is indistinguishable from Haar-random unitaries for any computationally bounded observer. Here, late time is defined as $t = \omega(\poly n)$. Despite being intuitive, this result has never been rigorously proven. In Ref.~\cite{cotler2017black}, the authors conjectured that there should exist some late time at which evolution under the GUE would resemble Haar-random unitaries, but they also claimed that GUE evolution should deviate from Haar-random unitaries for extremely long times. This conjecture was based on the intuition that while the spectral distribution of Haar-random unitaries exhibits level repulsion, the long-time evolution under GUE erases this repulsion, leading to eigenphases uniformly distributed on the circle. While this intuition is accurate, it is precisely the reason why late-time GUE evolution resembles Haar-random unitaries. We will first show that the spectral distribution of Haar-random unitaries is `spoofed' by a uniform independent and identically distributed (iid) distribution, which is reached by the late-time evolution under the GUE. Thus, not only do we rigorously prove this result, but we also refute the reasoning about the deviation from Haar-random behavior: for any $t = \omega(\poly n)$, the GUE evolution is indistinguishable from Haar-random unitaries for any computationally bounded observer.

\subsection{Proof of indistinguishability}

Haar random unitaries are also known as the Circular Unitary Ensemble (CUE), introduced by Dyson in Ref.~\cite{dysonCUE1962}. As every matrix basis-invariant ensemble, we can write it as
\be
\haar=\{U\Phi U^{\dag}\,:\, U\sim\haar,\,\Phi\sim q(\phi_1,\ldots,\phi_{d})\}\label{eq:tautdefinitionCUE}
\ee
where is a diagonal unitary with eigenphases $\phi_i\in(0,2\pi]$ distributed according to the CUE spectral distribution~\cite{mehta2004random}, given below
\be
q(\phi_1,\ldots, \phi_d)=\frac{1}{d!}\prod_{i,j}|e^{i\phi_i}-e^{i\phi_j}|^2
\ee
i.e., exhibiting level repulsion. While seemingly tautological, the expression in \cref{eq:tautdefinitionCUE} is the key to the simple proof of indistinguishability. Indeed, intuitively, it is sufficient to show that the unitaries $\Phi \sim q$ and $e^{-i\Lambda t}$ with $\Lambda \sim p$ are indistinguishable, as both cases feature a random eigenbasis. In the following, we will primarily focus on proving the statement regarding diagonal unitaries, employing the techniques and ideas developed in \cref{App:pseudochaoticconstruction}. Finally, in \cref{th:GUE=Haarlatetimes}, we present the rigorous result establishing the indistinguishability.

Thanks to the result of Ref.~\cite{diaconis1980finite} and the fact that $q(\phi_1,\ldots,\phi_d)$ is permutation-invariant, there exists an iid probability distribution that approximates every $k$-marginal of $q$, with an error scaling as $O(k^2/d)$. To find the single eigenvalue distribution $\tilde{q}(\lambda)$, it is useful to express the marginals of $q$ as shown in Ref.~\cite{mehta2004random}.
\be
q^{(k)}(\phi_1,\ldots,\phi_k)=\frac{(d-k)!}{d!}\frac{1}{(2\pi)^k}\det(d\mathbb{I}+C)
\ee
where $C$ is a $k\times k$ matrix with components
\be
C_{ij}=\begin{cases}
    0 &\qq{if $i=j$}\\
    \frac{\sin(d(\phi_i-\phi_j))}{\sin(\phi_i-\phi_j)} &\qq{otherwise.} \label{def:matrixC}
\end{cases}
\ee
Let us define $\tilde{q}(\phi_1,\ldots,\phi_d)=\frac{1}{(2\pi)^d}$ the uniform distribution over all the eigenphases. Let us show that it well-approximates the $k$-marginals of $q$ up to $O(k^2/d)$ error in total variation distance. 
\begin{lemma}[Upper bound to probabilities of events]
Let $E$ be the event where $\abs{\sin(\phi_i-\phi_j)}<\frac{1}{\sqrt{d}}$, then 
\be
\Pr_{q^{(k)}}[E],\Pr_{\tilde{q}^{(k)}}[E]\le \frac{1}{\sqrt{d}}
\ee
\begin{proof}
Let us start from $q^{(k)}$, to get
\ba
\Pr_{q^{(k)}}[E]&=&\sum_{i=1}^{k}\sum_{j=i+1}^{k}\Pr_{q^{(k)}}\left[|\sin(\phi_i-\phi_j)|< \frac{1}{\sqrt{d}}\right]\\
\nonumber
&=&\frac{k(k-1)}{2}\Pr_{q^{(2)}}\left[|\sin(\phi_1-\phi_2)|< \frac{1}{\sqrt{d}}\right]\\
\nonumber
&\le&\frac{k(k-1)}{2}\Pr_{q^{(2)}}\left[|\phi_1-\phi_2|< \frac{1}{\sqrt{d}}\right]\\
\nonumber
&=&\frac{k(k-1)}{2(1-1/d)}\frac{1}{4\pi^2}\int_{0}^{2\pi}\de\phi_1\int_{\phi_1-\frac{1}{\sqrt{d}}}^{\phi_1+\frac{1}{\sqrt{d}}}\de\phi_2\left[1-\left(\frac{\sin[d(\phi_1-\phi_2)]}{d \sin(\phi_1-\phi_2)}\right)^2\right].
\ea
\nonumber
Since $0 \leq 1-\left(\frac{\sin[d(\phi_1-\phi_2)]}{d \sin(\phi_1-\phi_2)}\right)^2 \leq 1$, we can immediately upper bound the value of the inner integral with $2/\sqrt{d}$, and find
\ba
\Pr_{q^{(k)}}[E]\le O\left(\frac{k^2}{\sqrt{d}}\right).
\ea
The analogous statement for $\tilde{q}^{(k)}$ follows from
\be
\Pr_{\tilde{q}^{(k)}}[E]\le \frac{k(k-1)}{2}\frac{1}{(2\pi)^2}\int_{0}^{2\pi}\de\phi_1\int_{\phi_1-\frac{1}{\sqrt{d}}}^{\phi_1+\frac{1}{\sqrt{d}}}\de\phi_2=O\left(\frac{k^2}{\sqrt{d}}\right).
\ee
\end{proof}
\end{lemma}

\begin{lemma}[Closeness of marginals]\label{Applem:closenessmarginalsCUE}
The marginals $q^{(k)}$ and $\tilde{q}^{(k)}$ are close in total variation distance, i.e.,  \be
\delta(q^{(k)},\tilde{q}^{(k)})=O\left(\frac{k}{d^{1/4}}\right).
\ee
\begin{proof}
The proof will be, as in \cref{app:spoof}, divided into three steps. First, we define
\begin{equation}
    \begin{aligned}
    q_{\cut}^{(k)}(\phi_1,\ldots,\phi_k)&\coloneqq\begin{cases}
    0 &\qq{if $E$} \\
    \frac{q(\phi_1,\ldots,\phi_k)}{1-\Pr_{q^{(k)}}[E]} &\qq{otherwise}
    \end{cases} \\
        \tilde{q}^{(k)}_{\cut}(\phi_1,\ldots,\phi_d)&\coloneqq \begin{cases} 0 &\qq{if $E$} \\
    \frac{\tilde{q}(\phi_1,\ldots,\phi_d)}{1-\Pr_{\tilde{q}^{(k)}}[E]} &\qq{otherwise}
    \end{cases}
    \end{aligned}
\end{equation}
We bound $D_{KL}(q^{(k)} \parallel q_{\cut}^{(k)})$ as
\be
D_{KL}(q_{\cut}^{(k)} \parallel q^{(k)})=\int \dd{\phi_1} \cdots \dd{\phi_k} q_{\cut}^{(k)}(\phi_{1},\ldots, \phi_k)\ln\left(\frac{q^{(k)}(\phi_1,\ldots,\phi_k)}{q_{\cut}^{(k)}(\phi_1,\ldots,\phi_k)}\right)=-\ln(1-\Pr[E])=O\left(\frac{k^2}{\sqrt{d}}\right).
\ee
Similarly, it holds that
\be
D_{KL}(\tilde{q}_{\cut}^{(k)}\parallel\tilde{q}^{(k)})=O\left(\frac{k^2}{\sqrt{d}}\right)
\ee
Finally, let us bound $D_{KL}(\tilde{q}_{\cut} \parallel q_{\cut})$. Similarly to \cref{lem:tilde-rho}, let us use the fact that $|\ln\det(\mathbb{I}+d^{-1}C)|\le(d^{-1}\tr(C)+d^{-2}\|C\|_{2}^{2}/2)$. The matrix $C$, defined in \cref{def:matrixC}, has null diagonal. Hence, we find that $\tr(C)=0$. Then, over the domain of $\tilde{q}^{(k)}_{\cut}$, we have that $|C_{ij}|\le \sqrt{d}$, hence
\be
\norm{C}_2^2 =\sum_{i,j} \abs{C_{ij}}^2\le k^2d,
\ee
from which it follows that $\abs{\ln\det(\mathbb{I}+d^{-1}C)}\le \frac{k^2}{2d}$. Therefore,
\begin{equation}
    \begin{aligned}
        D_{KL}(\tilde{q}^{(k)}_{\cut} \parallel q^{(k)}_{\cut})&=\int \dd{\phi_1} \cdots \dd{\phi_k} \tilde{q}_{\cut}\ln\left(\frac{{q}^{(k)}_{\cut}}{\tilde{q}^{(k)}_{\cut}}\right)\\
&=\int\de\phi_1\cdots\de\phi_k \tilde{q}_{\cut}\ln\left(\frac{(d-k)!d^k}{d!}\det(\mathbb{I}+d^{-1}C)\right)\\
&\leq \ln \left(\frac{(d-k)!d^k}{d!}\right)+\frac{k^2}{2d}\\
&=\frac{k(k-1)}{d}+\frac{k^2}{2d}=O\left(\frac{k^2}{d}\right).
    \end{aligned}
\end{equation}
where we used that $\ln \left(\frac{(d-k)!d^k}{d!}\right)\le \frac{k(k-1)}{d}$ for $k\le d/2$, see also \cref{lem:tilde-rho}. Therefore, we can finally bound the TV distance between $\tilde{q}^{(k)}$ and $q^{(k)}$. Since we have bounds on the KL-distance, we use the general inequality
\be
\delta(P,Q)\le\sqrt{\frac{1}{2}D_{KL}(P\|Q)}
\ee
Hence:
\ba
\delta(\tilde{q}^{(k)},q^{(k)})&\le& \delta(\tilde{q}^{(k)},\tilde{q}^{(k)}_{\cut}) +\delta({q}^{(k)}_{\cut},q^{(k)})+\delta(\tilde{q}^{(k)}_{\cut},q^{(k)}_{\cut})\\
&\le& \sqrt{\frac{1}{2}D_{KL}(\tilde{q}^{(k)}\|\tilde{q}^{(k)}_{\cut})} +\sqrt{\frac{1}{2}D_{KL}({q}^{(k)}_{\cut}\|q^{(k)})}+\sqrt{\frac{1}{2}D_{KL}(\tilde{q}^{(k)}_{\cut}\|q^{(k)}_{\cut})}\\
&=&O\left(\frac{k}{d^{1/4}}+\frac{k}{d^{1/2}}\right)=O\left(\frac{k}{d^{1/4}}\right)
\nonumber
\ea
which concludes the proof.
\end{proof}
\end{lemma}
\begin{lemma}[Wigner semi-circle law approaches uniform distribution]\label{lem:wignerdysontouniform} Let $t\in[0,\infty)$. Let $p^{(1)}(\lambda)=\frac{1}{2\pi}\sqrt{4-\lambda^2}$ be the marginal of the GUE eigenvalue distribution and $q^{(1)}=\frac{1}{2\pi}$ be the uniform distribution. Let $p^{(1)}_{+}(\phi)$ be the probability density induced by $p^{(1)}(\lambda)$ via $\phi=\lambda t$ for $\phi\in[0,2\pi)$, then
\be
\delta\left(q^{(1)}(\phi), p^{(1)}_{+}(\phi)\right)\le O(t^{-1/8}).
\ee
\end{lemma}
\begin{proof}
Let $U=[\lambda_{\min},\lambda_{\max}] \coloneqq \qty{\lambda \in [-2,2]: p^{(1)}(\lambda) \geq 1/\sqrt{t}}$. Define
\begin{eqnarray}
    p^{(1)}_{\cut}(\lambda) = \begin{cases}
        \frac{p^{(1)}(\lambda)}{1-\int_U p^{(1)}(\lambda) \dd{\lambda}} &\qq{for $\lambda \in U$} \\
        0 &\qq{otherwise.}
    \end{cases}
\end{eqnarray}
Note that \cref{lem:e2} proved that $\int_U p^{(1)}(\lambda) \dd{\lambda} \leq 2t^{-1/4}$. By identical reasoning to \cref{thm:spoof}, this implies that $\delta(p^{(1)}(\lambda), p^{(1)}_{\cut}(\lambda)) \leq O(t^{-1/8})$. Also, since the mapping $\phi = \lambda t \pmod{2\pi}$ is deterministic, if we define $p_{\cut}(\phi)$ to be the distribution induced by $\phi = \lambda t \pmod{2\pi}$ for $\lambda \sim p^{(1)}_{\cut}$, we have $\delta(p^{(1)}_{+}(\phi), p^{(1)}_{\cut}(\phi)) \leq O(t^{-1/8})$ as well. Now, it remains to upper bound $\delta(p_{\cut}^{(1)}(\phi), q^{(1)}(\phi))$. 

Let us define 
\begin{equation}
    \begin{gathered}
    \lambda_{1}(\phi,t) \coloneqq \operatorname{argmin}\qty{\lambda \in U: \lambda t := \phi \pmod{2\pi}} \\
    \lambda_{2}(\phi,t) \coloneqq \operatorname{argmax}\qty{\lambda \in U: \lambda t := \phi \pmod{2\pi}}.
    \end{gathered}
\end{equation}
Now, observe that
\begin{equation}\label{eq:riemann}
    p^{(1)}_{\cut}(\phi) = \frac{1}{t} \qty(p^{(1)}_{\cut}(\lambda_1(\phi,t)) + p^{(1)}_{\cut}(\lambda_1(\phi,t) + 2\pi/t) + p^{(1)}_{\cut}(\lambda_1(\phi,t) + 4\pi/t) + \ldots + p^{(1)}_{\cut}(\lambda_2(\phi,t))).
\end{equation}
Note that the right side of this equation is simply a Riemann sum which approximates $\frac{1}{2\pi} \int_{\lambda_1(\phi,t)}^{\lambda_2(\phi,t)} p^{(1)}_{\cut}(\lambda) \dd{\lambda}$. Furthermore, it is well-known that the error of this approximation is upper bounded by $\frac{(\lambda_2(\phi,t) - \lambda_1(\phi,t))^2}{m} \cdot \max_{\lambda \in U} \abs{\dv{\lambda} p^{(1)}_{\cut}(\lambda)}$, where $m$ is the number of summands in \cref{eq:riemann}. Since $m=\Theta(t)$, and $\max_{\lambda \in U} \abs{\dv{\lambda} p^{(1)}_{\cut}(\lambda)} \leq \frac{\sqrt{t}}{1-\int_U p^{(1)}(\lambda) \dd{\lambda}} \leq O(\sqrt{t})$, we have that
\begin{equation}
    \abs{p_{\cut}^{(1)}(\phi) - \frac{1}{2\pi} \int_{\lambda_1(\phi,t)}^{\lambda_2(\phi,t)} p^{(1)}_{\cut}(\lambda) \dd{\lambda}} \leq O(t^{-1/2}).
\end{equation}
Finally, we observe that since $\int_{\lambda_{\min}}^{\lambda_{\max}} p^{(1)}_{\cut}(\lambda) \dd{\lambda} = 1$,
\begin{equation}
\int_{\lambda_1(\phi,t)}^{\lambda_2(\phi,t)} p^{(1)}_{\cut}(\lambda) \dd{\lambda} = 1 - \int_{\lambda_{\min}}^{\lambda_1(\phi,t)} p^{(1)}(\lambda) \dd{\lambda} - \int_{\lambda_{2}(\phi,t)}^{\lambda_{\max}} p^{(1)}(\lambda) \dd{\lambda} \geq 1 - O(t^{-1/2}),
\end{equation}
because $\lambda_1(\phi,t) - \lambda_{\min} \leq O(t^{-1/2})$ and $\lambda_{\max} - \lambda_2(\phi,t) \leq O(t^{-1/2})$. In summary,
\begin{equation}
    \abs{p^{(1)}_{\cut}(\phi) - \frac{1}{2\pi}} \leq O(t^{-1/2})
\end{equation}
for all $\phi$. This implies $\delta(p^{(1)}_{\cut}(\phi), q^{(1)}(\phi)) \leq O(t^{-1/2})$, hence $\delta(p^{(1)}_{+}(\phi), q^{(1)}(\phi)) \leq O(t^{-1/8})$.
\end{proof}
\begin{corollary}
 Let $t\in(-\infty,0]$. Let $p^{(1)}_{-}(\phi)$ be the probability density induced by $p^{(1)}(\lambda)$ via $\phi=\lambda t$ for $\phi\in[0,2\pi)$, then
\be
\delta\left(q^{(1)}(\phi), p^{(1)}_{-}(\phi)\right)\le O(t^{-1/8}).
\ee
\begin{proof}
The proof is identical to the one from \cref{lem:wignerdysontouniform} by replacing $t\mapsto |t|$.
\end{proof}
\end{corollary}
\begin{corollary}
Define $\bar{p}_{\pm t}(\phi)=\prod_{i=1}^{d}p^{(1)}_{\pm}(\phi)$ where $p^{(1)}_{\pm}(\phi)$ is the probability density induced by $p^{(1)}(\lambda)$ via $\phi=\pm\lambda t$ for $t\in[0,\infty)$ (see \cref{lem:wignerdysontouniform}). Then, the total variation distance with $k$-marginals $\tilde{q}^{(k)}$ of the uniform distribution $\tilde{q}(\phi_1,\ldots,\phi_d)=(2\pi)^{-d}$ is also bounded
\be
\delta(\tilde{q}^{(k)},\bar{p}_{\pm t}^{(k)})=O\left(\frac{k}{t^{1/8}}\right)
\ee
\begin{proof}
Since both distributions, $\bar{p}_{\pm}$ and $\tilde{q}$ are product distributions, their marginals are product as well. Using $\Theta(k)$ triangle inequalities, the result follows.
\end{proof}
\end{corollary}

We are finally ready to prove the main result of this section: late-time evolution under the GUE is indistinguishable from Haar-random unitaries for any $t = \omega(\poly n)$. To do this, we combine the following key facts: (i) the polynomial marginals of the spectral CUE distribution are uniformly distributed, (ii) the marginals of the spectral distribution of the GUE are approximated by independent variables following the semicircle law, and (iii) as shown in \cref{lem:wignerdysontouniform}, at late times the Wigner semicircle distribution converges to a uniform distribution over the unit circle.

\begin{theorem}[Late-time GUE is indistinguishable from Haar random]\label{th:GUE=Haarlatetimes}  Consider the ensemble of unitaries generated by GUE, as $\mathcal{U}_{t}=\{e^{-iHt}\,:\, H\sim \gue\}$. Then the ensemble $\mathcal{U}_t$ is adaptively statistically indistinguishable from Haar random unitaries for any $t=\omega(\poly n)$, even if one has access to the adjoint evolution. In particular, any quantum algorithm requires $k=\Omega(\min(t^{1/16}, d^{1/4}))$ queries to the time evolution and/or its inverse to distinguish the two ensembles.
\begin{proof}
Similarly to \cref{thm:spoof}, the most general algorithm aiming at processing $U\sim \mathcal{E}$ for $\mathcal{E}$ a general unitary ensemble will output the following quantum state
\be
\rho(U,U^{\dag}; \rho_0)\coloneqq[W_{k+1}(C^{n_k}U)W_{k}(C^{n_k-1}U^{\dag})\cdots W_{2}(C^{n_1}U)W_1](\rho_0),\quad U\sim\mathcal{E}
\ee
where $\rho_0$ is a $n+n'$ qubit state with $n'=O(\poly(n))$ and $C^{n_j}U_j(t_j)$ is a control version of $U,U^{\dag}$ which acts on the first $n$ qubits, and the $n_j$-th auxiliary system is the control qubit. Notice that the order, as well as the presence, of the adjoint is completely general. First of all, notice that we can write the Haar ensemble as $\haar=\{Ue^{i\Phi} U^{\dag}\,:\, \Phi\sim q\}$ and the ensemble of GUE unitaries as $\mathcal{U}_t=\{U e^{i\Lambda t} U^{\dag}\,:\, \Lambda\sim p\}$. We have already shown in \cref{thm:spoof} that the ensemble $\mathcal{U}_t$ is statistically indistinguishable from $\tilde{\mathcal{U}}_t=\{U e^{i\Lambda t} U^{\dag}\,:\, \Lambda\sim \tilde{p}\}$ where $\tilde{p}=\prod_{i=1}^{d}\tilde{p}(\lambda_i)$. In particular, we have shown
\be\label{App:eqfordesigns}
\norm{\E_{U\sim\mathcal{U}_t}\rho(U,U^{\dag};\rho_0)-\E_{U\sim\tilde{\mathcal{U}}_t}\rho(U,U^{\dag};\rho_0)}_{1}=O(\delta(p^{(k)},\tilde{p}^{(k)}))=O\left(\frac{\sqrt{k}}{d^{1/8}}\right)
\ee
In particular, without loss of generality we can consider $t\in[0,\infty)$ because of the presence of $U^{\dag}$ with reverse the time. Moreover, defining $\widetilde{\haar}\coloneqq\{Ue^{i\Phi} U^{\dag}\,:\, \Phi\sim \tilde{q}\}$ with $\tilde{q}=\prod_i{q}^{(1)}(\phi_i)$, thanks to the permutational invariance (similarly to \cref{thm:spoof}), we can write:
\be
\|\E_{U\sim\widetilde{\haar}}\rho(U,U^{\dag};\rho_0)-\E_{U\sim\haar}\rho(U,U^{\dag};\rho_0)\|_{1}=O(\delta(q^{(k)},\tilde{q}^{(k)}))=\frac{k}{d^{1/4}}
\ee
where the latter equality follows from \cref{Applem:closenessmarginalsCUE}. Finally, let us note that we can rewrite the set $\tilde{\mathcal{U}}_t$ as
\be
\tilde{\mathcal{U}}_{t}\coloneqq\{Ue^{i\Phi}U^{\dag}\,: \Phi\sim \bar{p}_{+t}\}
\ee
as we can evaluate the following TV distance 
\be
\|\E_{U\sim\tilde{\mathcal{U}}_t}\rho(U,U^{\dag};\rho_0)-\E_{U\sim\widetilde{\haar}}\rho(U,U^{\dag};\rho_0)\|_{1}\le O(k\delta(\bar{p}^{(k)}_{\pm t},\tilde{q}^{(k)}))=O\left(\frac{k^2}{t^{1/8}}\right)
\ee
where the first equality follows from the fact that both ensemble consist of unitaries written in a Haar random basis; hence, thanks to the permutational invariance, the distance depends only on the $k$-marginal, as usual. Moreover, notice that we used $\Theta(k)$ many triangle inequality to always compare $\bar{p}_{+}$ or $\bar{p}_{-}$ with the uniform distribution $\tilde{q}$.

We finally can evaluate the statistical distance between unitaries generated by the GUE ensemble and Haar random unitaries: we can just then use triangle inequalities:
\ba
\|\E_{U\sim\mathcal{U}_t}\rho(U,U^{\dag};\rho_0)-\E_{U\sim\haar}\rho(U,U^{\dag};\rho_0)\|_{1}&\le& \|\E_{U\sim\widetilde{\haar}}\rho(U,U^{\dag};\rho_0)-\E_{U\sim{\haar}}\rho(U,U^{\dag};\rho_0)\|_{1}\\
&+& \|\E_{U\sim\tilde{\mathcal{U}}_t}\rho(U,U^{\dag};\rho_0)-\E_{U\sim {\mathcal{U}_t}}\rho(U,U^{\dag};\rho_0)\|_{1}\\
&+&\|\E_{U\sim\tilde{\mathcal{U}}_t}\rho(U,U^{\dag};\rho_0)-\E_{U\sim\widetilde{\haar}}\rho(U,U^{\dag};\rho_0)\|_{1}\\
&=&O\left(\frac{\sqrt{k}}{d^{1/8}}\right)+O\left(\frac{k}{d^{1/4}}\right)+O\left(\frac{k^2}{t^{1/8}}\right)\\&=&O\left(\frac{k}{d^{1/4}}\right)+O\left(\frac{k^2}{t^{1/8}}\right)
\ea
where we used the fact that $t=o(d^{4})$. This concludes the proof. 
\end{proof}
\end{theorem}
As a simple corollary, we can prove that late-time GUE forms an approximate unitary $k$-design.
\begin{corollary}[Late-time GUE is a unitary $k$ design]\label{Appcor:unitarykdesign}
Consider the ensemble of unitaries generated by GUE, as $\mathcal{U}_{t}=\{e^{-iHt}\,:\, H\sim \gue\}$. Then the ensemble $\mathcal{U}_t$ is a $O\left(\frac{k^2}{t^{1/8}}\right)$-approximate unitary $k$-design. 
\be
\norm{\mathbb{E}_{U\sim \mathcal{U}_t}U^{\otimes k}(\cdot) U^{\dag\otimes k}-\mathbb{E}_{U\sim \haar}U^{\otimes k}(\cdot) U^{\dag\otimes k}}_{\diamond}\le O\left(\frac{k^2}{t^{1/8}}\right)
\ee
In particular for any $t=\omega(\poly n)$, it is a $o\left(\frac{1}{\poly n}\right)$-approximate unitary $k$-design for any $k=O(\poly n)$. 
\begin{proof}
The proof comes directly from the definition of the diamond norm 
\begin{equation}
    \begin{gathered}
        \norm{\mathbb{E}_{U\sim \mathcal{U}_t}U^{\otimes k}(\cdot) U^{\dag\otimes k}-\mathbb{E}_{U\sim \haar}U^{\otimes k}(\cdot) U^{\dag\otimes k}}_{\diamond} \\
= \max_{n'} \max_{\rho}\norm{\mathbb{E}_{U\sim \mathcal{U}_t}[U^{\otimes k}\otimes \id_{n'}(\rho) U^{\dag\otimes k}\otimes \id_{n'}]-\mathbb{E}_{U\sim \haar}[U^{\otimes k}\otimes \id_{n'}(\rho) U^{\dag\otimes k}\otimes \id_{n'}]}_1\label{App:eqfordesigns2}
    \end{gathered}
\end{equation}
where $\rho$ is $(n+n')$-qubit state with $n'$ ancilla. It is then sufficient to notice that \cref{App:eqfordesigns2} is a particular case of \cref{App:eqfordesigns}. 
\end{proof}
\end{corollary}

\subsection{Discussion: Constructing pseudorandom unitaries from pseudo-GUE}

Let us now comment on the above result and outline its consequences, merging it with the other findings of this paper. Specifically, we have shown that there exists an ensemble of efficiently implementable unitaries that spoofs the GUE ensemble for any computationally bounded observer. Additionally, as demonstrated in \cref{App:pseudochaoticconstruction}, the time complexity of the algorithm is only polylogarithmic in the time $t$ of evolution under the GUE, meaning our ensemble of pseudochaotic Hamiltonians can be exponentially fast-forwarded. 

By combining (1) the indistinguishability of late-time GUE and Haar-random unitaries, (2) the existence of efficiently implementable pseudo-GUE Hamiltonians, and (3) the fact that these pseudo-GUE Hamiltonians are exponentially fast-forwardable, we conclude that our ensemble of pseudo-GUE Hamiltonians generates (for superpolynomially large times) an ensemble of pseudorandom unitaries, i.e., efficiently implementable unitaries that are indistinguishable from Haar-random unitaries. Moreover, the ensemble generated by pseudo-GUE Hamiltonians is also ``inverse-secure," meaning that it remains indistinguishable even if the algorithm has access to adjoint (inverse) operations.
Two remarks are in order: first, constructing pseudorandom unitaries via GUE evolution has a distinct advantage. Indeed, we have shown that we can tune the effective dimension of the pseudo-GUE ensemble, denoted as $\tilde{d}$ throughout the manuscript, which in turn allows us to adjust the entanglement, magic, and scrambling properties of the generated unitaries. In particular, by taking $\tilde{d} = \exp(\poly\log n)$, the ensemble of unitaries generated remains a pseudorandom ensemble, but it exhibits only polylogarithmic entanglement and polylogarithmic magic. Hence, this ensemble qualifies not only as a pseudorandom ensemble of unitaries, but also as a pseudoentangled and pseudomagic ensemble of unitaries. This is a novel contribution to the literature, as the known constructions of inverse-secure pseudorandom unitaries do not exhibit these properties.
Second, while constructing pseudorandom unitaries via GUE evolutions is certainly possible, our construction so far relies on the existence of (inverse-secure) pseudorandom unitaries. In \cref{App:pseudochaoticconstruction}, we use pseudorandom unitaries to spoof the Haar-random basis featured by GUE Hamiltonians. Therefore, the current method demonstrates the existence of pseudoentangled and pseudomagic ensembles of unitaries that are also pseudorandom unitaries, but it is not sufficient by itself to show the existence of pseudorandom unitaries, proven only recently~\cite{ma2024constructrandomunitaries}. 

However, a closer look at our methodology reveals that while pseudorandom unitaries are sufficient to spoof the Haar-random basis of the GUE, they may not be strictly necessary. In fact, we require a weaker condition. Specifically, we need an ensemble $\mathcal{E}$ of unitaries such that for any fixed diagonal unitary $\Phi$ (in the computational basis), the following holds:
\be
\left|\Pr_{U\sim \mathcal{E}}\mathcal{A}^{U\Phi U^{\dag}}(\cdot) - \Pr_{U\sim \haar}\mathcal{A}^{U\Phi U^{\dag}}(\cdot)\right| = \negl(n)\,,\label{Eq:weaker?}
\ee
for any quantum algorithm $\mathcal{A}$, with query accesses to $U\Phi U^{\dag}$, running in polynomial time.

While it is clear that \cref{Eq:weaker?} represents a weaker condition compared to the usual pseudorandom unitary condition, the question of whether it is strictly weaker remains open. 

\begin{open}
Does there exist an ensemble of unitaries satisfying \cref{Eq:weaker?}, which is not computationally indistinguishable from Haar-random unitaries (i.e., not qualifying as pseudorandom unitaries)? In other words, is \cref{Eq:weaker?} strictly weaker than the standard definition of pseudorandom unitaries?
\end{open}

A positive answer to the above question would provide an alternative and independent proof of the existence of pseudorandom unitaries through our method of spoofing the GUE ensemble.

\end{document}